\documentclass[10pt,a4paper]{article}
\usepackage[utf8]{inputenc}
\usepackage{amsmath, amsthm}
\usepackage{amsfonts}
\usepackage{amssymb}
\usepackage{hyperref}	

\usepackage[left=1cm,right=1cm,top=1cm,bottom=1.3 cm]{geometry}

\usepackage{enumerate}

\usepackage{enumerate}
\newcommand{\barL}{\underline{L}}
\theoremstyle{plain}
\newtheorem{thm}{Theorem}[section]

\newtheorem*{thmunb}{Theorem}

\newtheorem{lem}[thm]{Lemma}
 
\newtheorem{prop}[thm]{Proposition}

\newtheorem{cor}[thm]{Corollary}
\newtheorem*{cor*}{Corollary}

\newtheorem{conjecture}[thm]{Conjecture}

\newtheorem{theo}{Theorem}
\newtheorem{coreo}{Corollary}
\newtheorem{theotwo}{Theorem}
\newtheorem{theothree}{Theorem}

\theoremstyle{remark}
\newtheorem{rmk}{Remark}[section]

\newcommand{\A}{\mathcal{A}}
\newcommand{\Oo}{\tilde{O}_1}
\newcommand{\Ot}{\tilde{O}_2}
\newcommand{\Oth}{\tilde{O}_3}

\newcommand{\R}{\mathcal{R}}
\newcommand{\T}{\mathcal{T}}
\theoremstyle{definition}

\newtheorem{defn}{Definition}[section]
\newtheorem{hyp}{Assumption}

\newcommand{\RR}{\mathbb{R}}
\newcommand{\CC}{\mathbb{C}}

\theoremstyle{plain}

\theoremstyle{remark}

\theoremstyle{definition}

\newcommand{\LB}{\mathcal{LB}}
\newcommand{\uu}{u_{max}}
\newcommand{\barM}{\overline{M}}

\newcommand{\CH}{\mathcal{CH}_{i^+}}
\newcommand{\CHone}{\mathcal{CH}_{i^+_1}}
\newcommand{\CHtwo}{\mathcal{CH}_{i^+_2}}

\usepackage{bm}

\usepackage{color}

\usepackage{float}

\addtocounter{tocdepth}{-2}
\usepackage{graphicx}

\numberwithin{equation}{section}
\begin{document}
	\title{Mass inflation and the $C^2$-inextendibility of spherically symmetric charged scalar field dynamical black holes} 
	\author{Maxime Van de Moortel}
	\maketitle
	
	\abstract
	
	It has long been suggested that the Cauchy horizon of dynamical black holes is subject to a weak null singularity, under the mass inflation scenario. 
	We study in spherical symmetry the Einstein--Maxwell--Klein--Gordon equations and \textit{while we do not directly show mass inflation}, we obtain a ``mass inflation/ridigity'' dichotomy. More precisely, \color{black} we prove assuming (sufficiently slow) decay of the charged scalar field  on the event horizon, that the Cauchy horizon emanating from time-like infinity $\CH$ can be partitioned as $\CH= \mathcal{D} \cup \mathcal{S}$ for two (possibly empty) disjoint connected sets $\mathcal{D}$ and $\mathcal{S}$ such that\begin{itemize}
		
		\item  $\mathcal{D}$ (the dynamical set) is a future set on which the Hawking mass blows up (mass inflation scenario).

		\item  $\mathcal{S}$ (the static set) is a past set isometric to a Reissner--Nordstr\"{o}m Cauchy horizon i.e.\ the radiation is zero on $\mathcal{S}$.
	\end{itemize} As a consequence of this result,  we prove that  the entire Cauchy horizon $\CH$ is \textit{globally} \underline{$C^2$-inextendible}, extending a previous local result established by the author\color{black}. To this end, we establish a novel classification of Cauchy horizons into three types: dynamical ($\mathcal{S}=\emptyset$), static ($\mathcal{D}=\emptyset$) or mixed. As a side benefit, we prove that there exists a trapped neighborhood of the Cauchy horizon, thus the apparent horizon \textit{cannot} cross the Cauchy horizon, which is a result of independent interest. 
	
	Our main motivation is to prove the $C^2$ Strong Cosmic Censorship Conjecture for a realistic model of spherical collapse in which charged matter emulates the repulsive role of angular momentum. In our case, this model is the Einstein--Maxwell--Klein--Gordon system on space-times with  \textbf{one} asymptotically flat end. 
	As a consequence of the $C^2$-inextendibility of the Cauchy horizon, we prove the following statements, in spherical symmetry:

	\begin{enumerate}
		\item Two-ended asymptotically flat space-times are $C^2$-future-inextendible i.e.\  $C^2$ Strong Cosmic Censorship is true for Einstein--Maxwell--Klein--Gordon, assuming the decay of the scalar field on the event horizon at the expected rate.
		
		\item  In the one-ended case, under the same assumptions, the Cauchy horizon emanating from time-like infinity is $C^2$-inextendible. This result suppresses the main obstruction to $C^2$ Strong Cosmic Censorship in spherical collapse.
	\end{enumerate} The remaining obstruction in the one-ended case is associated to ``locally naked'' singularities emanating from the center of symmetry, a phenomenon which is also related to the Weak Cosmic Censorship Conjecture.
	
	

	\section{Introduction}
	
	\paragraph{Context of the problem} We study in spherical symmetry the Einstein--Maxwell-Klein-Gordon system, featuring a charged scalar field of charge $q_0 \neq 0$ and mass $m^2$, which we allow to be either massive ($m^2 \neq 0$) or massless ($m^2 =0$):
	\begin{equation} \label{1} Ric_{\mu \nu}(g)- \frac{1}{2}R(g)g_{\mu \nu}= \mathbb{T}^{EM}_{\mu \nu}+  \mathbb{T}^{KG}_{\mu \nu} ,    \end{equation} 
	\begin{equation} \label{2} \mathbb{T}^{EM}_{\mu \nu}=2\left(g^{\alpha \beta}F _{\alpha \nu}F_{\beta \mu }-\frac{1}{4}F^{\alpha \beta}F_{\alpha \beta}g_{\mu \nu}\right),
	\end{equation}
	\begin{equation}\label{3}  \mathbb{T}^{KG}_{\mu \nu}= 2\left( \Re(D _{\mu}\phi \overline{D _{\nu}\phi}) -\frac{1}{2}(g^{\alpha \beta} D _{\alpha}\phi \overline{D _{\beta}\phi} + m ^{2}|\phi|^2  )g_{\mu \nu} \right), \end{equation} \begin{equation} \label{4}\nabla^{\mu} F_{\mu \nu}= q_{0}\frac{i (\phi \overline{D_{\nu}\phi} -\overline{\phi} D_{\nu}\phi) }{2} , \; F=dA ,
	\end{equation} \begin{equation} \label{5} g^{\mu \nu} D_{\mu} D_{\nu}\phi = m ^{2} \phi,
	\end{equation} where $D:=\nabla+iq_0 A$. This model has been extensively studied in the past c.f. \cite{extremeJonathan} \cite{HodPiran1}, \cite{HodPiran2}, \cite{Kommemi}, \cite{KonoplyaZhidenko}, \cite{OrenPiran}.
	
	We are interested in black hole solutions arising from regular spherically symmetric, asymptotically flat initial data with one or two ends. In the one-ended case (spherical collapse), charged matter ($q_0 \neq 0$) is indispensable, else the Maxwell field  $F_{\mu \nu}$ is zero\color{black}. In the two-ended case, if $\phi \equiv 0$, all non-trivial solutions coincide with a Reissner--Nordstr\"{o}m black hole (see section \ref{RNsolution}). The Reissner--Nordstr\"{o}m Cauchy horizon, which is also the boundary of the maximal globally hyperbolic development, is smoothly extendible; it is a well-known fact that this poses a threat to determinism. In the context of gravitational collapse, a resolution, later known as ``Strong Cosmic Censorship'', was proposed by Penrose in \cite{PenroseSCC}. The strongest version of Strong Cosmic Censorship was often conjectured in the past: we express it in modern terminology as \begin{conjecture}[$C^0$ version of the Strong Cosmic Censorship Conjecture] \label{C0SCC} The maximal globally hyperbolic development of generic regular data for \eqref{1}, \eqref{2}, \eqref{3}, \eqref{4}, \eqref{5} is inextendible as a continuous Lorentzian manifold.\end{conjecture} The belief associated to Conjecture \ref{C0SCC} was that the Reissner--Nordstr\"{o}m Cauchy horizon would ``disappear'' under the effect of any dynamical perturbation and would be replaced by a space-like singularity analogous to the Schwarzschild's.

	In \cite{Moi}, the author studied dynamical black holes solutions of \eqref{1},  \eqref{2}, \eqref{3}, \eqref{4}, \eqref{5} with $q_0 \neq 0$, assuming decay of the scalar field on the black hole event horizon. It was proven that the Cauchy horizon $\CH$, now defined as the null boundary emanating from time-like infinity $i^+$, is non-empty (c.f. Figure \ref{Fig1} for the one-ended case, or Figure \ref{Fig2} for the two-ended case), thus the above belief was \underline{false}. Moreover, it was also shown in \cite{Moi} that space-time is extendible as continuous Lorentzian manifold  in the case $m^2=0$ and in \cite{Moi3Christoph} for the case $m^2 \neq 0$: thus, Conjecture \ref{C0SCC} is also \underline{false}. The approach adopted in \cite{Moi} and \cite{Moi3Christoph} (see also \cite{MoiThesis}) is to prove semi-local stability estimates in a neighborhood of $i^+$, to obtain a portion of Cauchy horizon which is $C^0$-extendible. While the $C^0$ version of Strong Cosmic Censorship is false, a modified version -- where $C^2$-inextendibility replaces $C^0$-inextendibility -- is conjectured to hold: \begin{conjecture}[$C^2$ version of the Strong Cosmic Censorship Conjecture] \label{C2SCC} The maximal globally hyperbolic development of generic regular data for \eqref{1}, \eqref{2}, \eqref{3}, \eqref{4}, \eqref{5} is inextendible as a $C^2$ Lorentzian manifold.\end{conjecture} 
	Consistently with Conjecture \ref{C2SCC}, the author proved \cite{Moi} in the charged and massive case that a \underline{small} piece of the Cauchy horizon near time-like infinity \color{black} is $C^2$-inextendible, due to the blow up of a curvature component. Note, however, that this is insufficient to conclude the $C^2$-inextendibility of the entire Cauchy horizon $\CH$, as this would require global estimates, far from time-like infinity $i^+$, which were not available in \cite{Moi}. Thus, the proof of Conjecture \ref{C2SCC} required further developments (in fact, an intermediate formulation between Conjecture \ref{C0SCC} and Conjecture \ref{C2SCC} is conjectured to hold: the ``$H^1$ \textit{Strong Cosmic Censorship}'', which states that the maximal globally hyperbolic development of generic data is \textit{not} extendible as a continuous Lorentzian manifold with locally square-integrable Christoffel symbols. This formulation of Strong Cosmic Censorship is particularly interesting, as the weakest known solutions of the Einstein equations lie in this low-regularity class, c.f.\ the introduction of \cite{Christo5}. Nevertheless, it is notoriously difficult to prove this version of the conjecture, due to the absence of ``known'' geometric quantities at the $H^1$ level. Therefore, we will not discuss this issue further in the paper, and we refer the reader to \cite{KerrStab} and \cite{JonathanStab} for a detailed presentation of the different issues involved).

	\paragraph{The main result and motivation} In the present article, we bridge this gap and provide a global approach to the properties of the Cauchy horizon emanating from time-like infinity $\CH$, assuming the decay of the scalar field on the event horizon. Our main result is that the mass inflation scenario holds, except in a degenerate situation where the radiation is trivial on $\CH$ i.e.\ $\CH$ is isometric to its Reissner--Nordstr\"{o}m counterpart. As a consequence, we prove that the entire Cauchy horizon $\CH$ is $C^2$-inextendible, \textit{even in the degenerate situation}, establishing the blow up of Ricci curvature. This is because the blue-shift effect, a common cause for both mass inflation and the Ricci blow up, is always effective under our decay assumptions, see the discussion below. Our motivation to study a charged matter model is to understand the properties of black holes arising from spherical collapse, mathematically modelled as solutions of the Einstein equations with one-ended asymptotically flat initial data (i.e.\ diffeomorphic to $\RR^3$). We are specifically interested in the formation and the characteristics of dynamical Cauchy horizons, as they constitute the most prominent obstruction to Strong Cosmic Censorship, see above. As a consequence of the $C^2$-inextendibility of $\CH$, we prove Conjecture \ref{C2SCC}  in the two-ended case, under our decay assumptions c.f.\ Theorem \ref{rough2}. We also prove Conjecture \ref{C2SCC} in the one-ended case, if we additionally assume the absence of \textit{locally naked singularities} emanating from the center c.f.\ Theorem \ref{rough1cond}.

	\paragraph{Previous works on uncharged models} A restricted class of one-ended black holes was studied by Christodoulou in \cite{Christo1}, \cite{Christo2}, \cite{Christo3}, as spherically symmetric solutions of the Einstein-(uncharged)-scalar-field. However, the model studied by Christodoulou does not allow for the formation of Cauchy horizons, therefore the study of Strong Cosmic Censorship in this context is limited. A more suitable spherically symmetric model, the Einstein--Maxwell-(uncharged)-scalar-field, was first analyzed by Dafermos. In this new model, the Maxwell field plays the repulsive role of angular momentum and, as shown in \cite{MihalisPHD}, \cite{Mihalis1}, the Cauchy horizon is non-empty and $C^0$-extendible, under assumptions on the exterior that were later retrieved in \cite{PriceLaw}. The study of the Einstein--Maxwell-(uncharged)-scalar-field model culminated with the work of Luk and Oh \cite{JonathanStab}, \cite{JonathanStabExt}, who proved the $C^2$ version of Strong Cosmic Censorship in this spherically symmetric setting. However, in \cite{JonathanStab}, \cite{JonathanStabExt}, Strong Cosmic Censorship is proven for asymptotically two-ended space-time, which are ill-suited to study gravitational collapse, due to the absence of a center of symmetry in the Penrose diagram. This is because the Einstein--Maxwell-(uncharged)-scalar-field model is too restrictive: in fact, all regular solutions with a non-trivial Maxwell field are two-ended space-times, while gravitational collapse space-times are one-ended, with a regular center of symmetry.
	\paragraph{Cauchy horizons and weak null singularities} In the present manuscript, we study the global properties of the black hole interior for the Einstein--Maxwell--Klein--Gordon model, focusing on the characteristics of the Cauchy horizon (see  also \cite{r=0} for a global study focusing on the structure of singularities) to establish Conjecture \ref{C2SCC}. As explained above, the Cauchy horizon of the static Reissner--Nordstr\"{o}m black hole is smoothly extendible, which represents a priori an obstruction to Strong Cosmic Censorship. Nevertheless, the ``mass inflation scenario'', first suggested in the pioneering works \cite{Ori}, \cite{Poisson}, \cite{PoissonIsrael}, dictates that the Cauchy horizon of generic \textit{dynamical} black holes features a so-called weak null singularity and is thus $C^2$-inextendible. $C^2$-inextendibility is roughly equivalent to a blow up of curvature, in turn caused by the blue-shift effect (discovered by Penrose in \cite{Penroseblue}) which amplifies ingoing radiation near the Cauchy horizon. In fact, the blue-shift is also responsible for mass inflation, if moreover the \textit{outgoing} radiation is non-trivial. This explains why under our decay assumptions, the Cauchy horizon is \textit{always} weakly singular, in the sense that the curvature blows up and $C^2$-inextendibility holds, but in some  \textit{degenerate situations} when outgoing radiation is trivial, mass inflation does not occur. In vacuum, we mention a breakthrough of Dafermos and Luk in \cite{KerrStab}, who proved that the Cauchy horizon of small perturbations of Kerr is always non-empty. Whether this Cauchy horizon is weakly singular or not is still open; however, we indicate the remarkable construction of a large class of weakly singular Cauchy horizons in vacuum by Luk in \cite{JonathanWeakNull}.

	\paragraph{An approach to Strong Cosmic Censorship}	The first step in the proof of Conjecture \ref{C2SCC}, undertaken in \cite{Moi}, is to prove the generic existence of weak null singularities \textit{locally}, namely on a small portion of the Cauchy horizon near time-like infinity. In view of the weak nature of those singularities (which still make $C^2$ norms blow up) note however that quantitative \underline{stability} estimates are proven in \cite{Moi} at lower regularity i.e.\ in the $C^0$ norm and were crucial to the proof. The next step, which we accomplish in the present paper, is to prove that a weak null singularity is present \textit{globally} on the entire Cauchy horizon. The strategy differs radically from the local approach: it is impossible, a priori, to ``propagate the estimates'' of \cite{Moi}, as no ``smallness parameter'' is exploitable in this space-time region, far away from time-like infinity. Note that this problem can be entirely by-passed for uncharged matter models, see \cite{JonathanStab}:  for the Einstein--Maxwell-uncharged-scalar-field model, the propagation of weak null singularities on the Cauchy horizon is immediate, due to very special monotonicity properties which do not hold in more complex settings. In contrast, a comprehensive understanding of the global properties of the Cauchy horizon is useful to prove Conjecture \ref{C2SCC} for charged models or in more general contexts.
	
	\paragraph{Global properties of Cauchy horizons} In our approach, we establish a novel classification of Cauchy horizons into three categories: dynamical type, mixed type, or static type. Using this classification, we prove that in all three cases: \begin{itemize}
		\item The Cauchy horizon is ``trapped'', thus the apparent horizon \underline{cannot} cross the Cauchy horizon.
		\item The Cauchy horizon is (globally) $C^2$-inextendible.
		\item The maximal development is $C^2$-future-inextendible, under assumptions \footnote{In the two-ended case, no additional assumption is required. In the one-ended case, we obtain the result assuming additionally the absence of ``locally naked singularity'' emanating from the center of symmetry, a slightly stronger statement than the Weak Cosmic Censorship Conjecture.} which are conjectured to be generic.
	\end{itemize}
	In fact, only Cauchy horizons of dynamical type are expected to be generic. Nevertheless, the Reissner--Nordstr\"{o}m Cauchy horizon is of static type, and it is also possible to construct Cauchy horizons of mixed type (see Appendix \ref{appendix} and \cite{Ori}). The main difference between those three types, is the presence (or not) of non-trivial radiation on the Cauchy horizon: \begin{enumerate}
		\item On dynamical type Cauchy horizons the radiation is non-zero near time-like infinity\color{black}. The Hawking mass blows up.
		\item On static type Cauchy horizons the radiation is everywhere zero: thus, \color{black} a static Cauchy horizon is an isometric copy of the Reissner--Nordstr\"{o}m Cauchy horizon. The Hawking mass is finite (in fact constant).
		\item On mixed type Cauchy horizons the radiation is zero \color{black} up to a transition time $u_T$ and non-zero \color{black} at times between $u_T$ and $u_T+\epsilon$ for a small $\epsilon$. The Hawking mass blows up at times larger than $u_T$ but is finite at times smaller than $u_T$.
	\end{enumerate}
	As a result, we prove that the Hawking mass must eventually blow up on the Cauchy horizon under our assumptions, except if the Cauchy horizon is of static type, which is a \textit{degenerate} situation where all gauge invariants quantities coincide with their Reissner--Nordstr\"{o}m analogues: in particular the Hawking mass and the charge of the Maxwell field are constant. \begin{rmk} \label{remarktransverse}
		Note however that in the static type case, the ``tangential'' radiation is zero but the \underline{transverse} radiation is generically non-trivial. This is why Cauchy horizons of static type are \textbf{still} subject to a weak null singularity (thus $C^2$-inextendible), as this transverse radiation is blue-shifted, like in the other two cases. There is no inconsistency: Cauchy horizons of static types are isometric to Reissner--Nordstr\"{o}m's, but they are embedded differently in the interior space-time.
	\end{rmk}
	
	\paragraph{Strategy of the proof} The main challenge is to prove that Cauchy horizons which are neither of dynamical type, nor of static type obey the pattern of mixed type, namely that there exists only one transition from the static behavior (in the past) towards the dynamical behavior (in the future). The proof starts with data on the event horizon obeying decay estimates at the expected rates, from which we obtain local estimates on a outgoing cone close enough to time-like infinity, using the results of \cite{Moi}. Then, we resurrect a staticity condition \eqref{MihalisPHDcondition}, first discovered by Dafermos in \cite{MihalisPHD}. This condition propagates to the past, and with the help of additional quantitative estimates, one can establish the classification of Cauchy horizons. We must also prove, in the dynamical type and mixed type cases, that the Hawking mass blows up; we rely also on quantitative estimates, as no monotonicity property is available, in contrast with the previously considered uncharged models. Finally, we establish, both in the static type case, and at the early times of mixed type, that a weak null singularity, namely a blow up of a curvature component is present, despite the finiteness of the Hawking mass.
	
	\paragraph{Outline of the introduction} In section \ref{apriori}, we give a detailed description of the Einstein--Maxwell--Klein--Gordon matter model and we enumerate all the possible a priori Penrose diagrams, following \cite{Kommemi} in the two-ended case, and \cite{MihalisSStrapped} in the two-ended case. Then, we state our main result in section \ref{firstversion}. In section \ref{previouswork}, we mention the previous results in the case of uncharged matter models, in the two-ended case. In section \ref{connected}, we mention connected problems and great conjectures related to the black hole interior. Finally in section \ref{outline}, we give an outline of the proof and of the paper.

	\subsection{The Einstein--Maxwell--Klein--Gordon system, and a priori Penrose diagrams} \label{apriori}

	We consider the Einstein--Maxwell--Klein--Gordon equations, namely the Einstein equation in the presence of a charged scalar field (either massive, or massless) given by \eqref{1}, \eqref{2}, \eqref{3}, \eqref{4}, \eqref{5}, where $D:= \nabla+  iq_{0}A$ is the gauge derivative, $q_0 \neq 0$ is a coupling constant, also called the charge of the scalar field, $m^2 \in \RR$ is the mass of the scalar field, $\nabla$ is the Levi-Civita connection of $g$ and $A$ is the potential one-form.
	
	This matter model satisfies the dominant energy condition and the null condition; some general properties can be derived a priori from those two facts. Using ``soft estimates'', it is possible to give an inventory of the possibilities, a priori, for the interior structure of the black hole. However, such an argument cannot provide information on what is the ``generic behavior'', as a more thorough analysis is necessary (involving quantitative estimates) to obtain any more precise statement. We quote the result of the preliminary analysis, using a soft argument, in the one-ended case:
	\begin{theotwo}[A priori boundary characterization of one-ended spherically symmetric black holes, Kommemi, \cite{Kommemi}] \label{oneendedapriori}
		
		We consider the maximal development $(M=\mathcal{Q}^+ \times_r \mathcal{S}^2,g_{\mu \nu}, \phi,F_{\mu \nu})$ of smooth, spherically symmetric, containing no anti-trapped surface, one-ended initial data satisfying the Einstein--Maxwell--Klein--Gordon system, where $r: \mathcal{Q}^+ \rightarrow [0,+\infty)$ is the area-radius function. Then the Penrose diagram of $\mathcal{Q}^+$ is given by Figure \ref{Fig1}, with boundary $\Sigma \cup \Gamma$ in the sense of manifold-with-boundary --- where $\Sigma$ is space-like, and $\Gamma$, the center of symmetry, is time-like with $r_{|\Gamma}=0$ --- and boundary $\mathcal{B}^+$ induced by the manifold ambient $\RR^{1+1}$: $$ \mathcal{B}^+ = b_{\Gamma} \cup \mathcal{S}^1_{\Gamma} \cup \mathcal{CH}_{\Gamma} \cup  \mathcal{S}^2_{\Gamma} \cup  \mathcal{S}  \cup \mathcal{S}_{i^+}  \cup \CH \cup i^{+} \cup \mathcal{I}^+ \cup i^0,$$  where $i^0$ is space-like infinity, $\mathcal{I}^+$ is null infinity, $i^{+}$ is time-like infinity, and \begin{enumerate}
			\item $\CH$ is a connected (possibly empty) half-open null ingoing segment emanating from $i^{+}$. The area-radius function $r$ extends as a strictly positive function on $\CH$, except maybe at its future endpoint.
			\item $\mathcal{S}_{i^+}$ is a connected (possibly empty) half-open null ingoing segment emanating (but not including) from the end-point of $\CH \cup i^{+}$. $r$ extends continuously to zero on $\mathcal{S}_{i^+}$.
			\item $b_{\Gamma}$ is the center end-point i.e.\ the unique future limit point of $\Gamma$ in $\overline{\mathcal{Q}^+}-\mathcal{Q}^+$.
			\item $ \mathcal{S}^1_{\Gamma} $ is a connected (possibly empty) half-open null outgoing segment emanating from $b_{\Gamma}$. $r$ extends continuously to zero on $ \mathcal{S}^1_{\Gamma} $.
			\item $\mathcal{CH}_{\Gamma}$ is a connected (possibly empty) half-open null outgoing segment emanating from the future end-point of $b_{\Gamma} \cup  \mathcal{S}^1_{\Gamma} $. $r$ extends as a strictly positive function on $\mathcal{CH}_{\Gamma}$, except maybe at its future endpoint.
			\item $\mathcal{S}^2_{\Gamma}$ is a connected (possibly empty) half-open null outgoing segment emanating from the future end-point of $\mathcal{CH}_{\Gamma}$. $r$ extends continuously to zero on $\mathcal{S}^2_{\Gamma}$. 	\item $\mathcal{S} $ is a connected (possibly empty) achronal curve that does not intersect null rays emanating from $b_{\Gamma}$ or $i^+$. $r$ extends continuously to zero on $\mathcal{S}$.
		\end{enumerate} We also define the black hole region $ \mathcal{BH}:= \mathcal{Q}^+ \backslash J^{-}(\mathcal{I}^+) \neq \emptyset$, and the event horizon $\mathcal{H}^+ = \overline{J^{-}(\mathcal{I}^+)} \backslash J^{-}(\mathcal{I}^+) \subset  \mathcal{Q}^+$. \begin{figure} [H]
			
			\begin{center}
				
				\includegraphics[width= 97mm, height=70 mm]{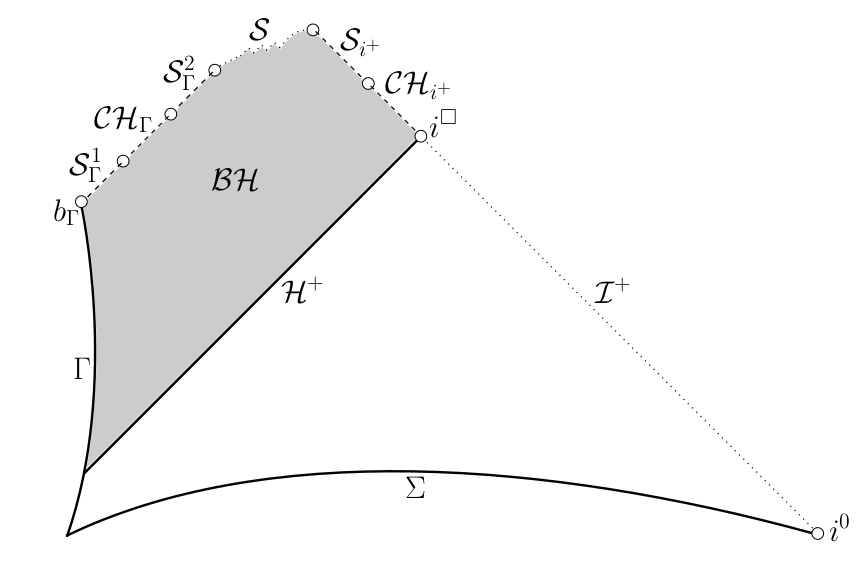}
				
			\end{center}
			
			\caption{General Penrose diagram of a one-ended charged spherically symmetric black hole, Figure from \cite{Kommemi}.}
			\label{Fig1}
		\end{figure}	
	\end{theotwo}

	We briefly discuss the global geometry of trapped surfaces. Each sphere corresponds to a point in the Penrose diagram. At any point, we define the outgoing null derivative of the area-radius function $r$. Then, we call the regular region the set of points for which the outgoing null derivative of $r$ is strictly positive, denoted $\R$ , the trapped region  the set of points for which the outgoing null derivative of $r$ is strictly negative, denoted $\T$, and the apparent horizon the set of points for which the outgoing null derivative of $r$ is zero, denoted $\A$. The structure of the trapped region can be very complex in general, see Figure \ref{Fig5}, if we just use the preliminary result of \cite{Kommemi}. To establish any non-trivial qualitative property on the apparent horizon $\A$ requires quantitative estimates. While the global properties of  $\A$ differ in the one or two-ended case, the properties of $\A$ in the vicinity of the Cauchy horizon $\CH$ are similar in both cases, as we will show.
	\begin{figure}[H]
		
		\begin{center}
			
			\includegraphics[width=97 mm, height=70 mm]{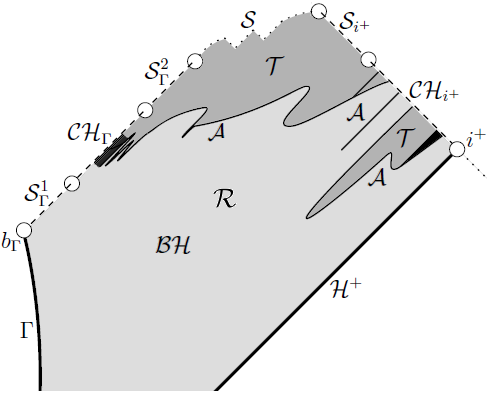}
			
		\end{center}
		
		\caption{General structure of the trapped region, Figure from \cite{Kommemi}.}
		\label{Fig5}
	\end{figure}
	In the two-ended case, the analogue of the ``no anti-trapped surface'' assumption is the admissibility condition (see Definition \ref{admissibilitydef}), satisfied on $\Sigma$ if the outgoing derivative of the area radius is negative near one end, and its ingoing derivative is negative near the other end. Now we present the analogue of Theorem \ref{oneendedapriori} for two-ended admissible space-times:
	\begin{theothree}[A priori boundary characterization of two-ended spherically symmetric black holes, Dafermos \cite{Mihalisnospacelike}, Kommemi \cite{Kommemi}]  \label{twoendedapriori}
		We consider the maximal development $(M=\mathcal{Q}^+ \times_r \mathcal{S}^2,g_{\mu \nu}, \phi,F_{\mu \nu})$ of smooth, spherically symmetric, two-ended admissible initial data satisfying the Einstein--Maxwell--Klein--Gordon system. Then the Penrose diagram of $\mathcal{Q}^+$ is given by Figure \ref{Fig2}, with boundary $\Sigma$ space-like and boundary $\mathcal{B}^+$ induced by the manifold ambient $\RR^{1+1}$: $$ \mathcal{B}^+ =  \mathcal{S}  \cup \mathcal{S}_{i^+}  \cup \CH \cup i^{+} \cup \mathcal{I}^+ \cup i^0,$$ where the definition of the boundary components are analogous \color{black} to those of Theorem \ref{oneendedapriori}, and moreover, see Figure \ref{Fig2} : 
		$$ i^{+}= i^+_1 \cup  i^+_2,$$
		$$ \mathcal{S}_{i^+}= \mathcal{S}_{i^+_1} \cup  \mathcal{S}_{i^+_2},$$
		$$ \mathcal{CH}_{i^+}= \mathcal{CH}_{i^+_1} \cup  \mathcal{CH}_{i^+_2},$$
		$$ \mathcal{I}^+=\mathcal{I}^+_1 \cup  \mathcal{I}^+_2,$$
		$$ i^{0}= i^0_1 \cup  i^0_2,$$
		We  define $ \mathcal{BH}:= \mathcal{Q}^+ \backslash (J^{-}(\mathcal{I}^+_1)  \cap \mathcal{Q}^+ \backslash (J^{-}(\mathcal{I}^+_2)\neq \emptyset$, and $\mathcal{H}^+_i = \overline{J^{-}(\mathcal{I}^+_i)} \backslash J^{-}(\mathcal{I}^+_i) \subset  \mathcal{Q}^+$, for $i=1,2$.
		\begin{figure}[H]
			
			\begin{center}
				
				\includegraphics[width=129 mm, height=60 mm]{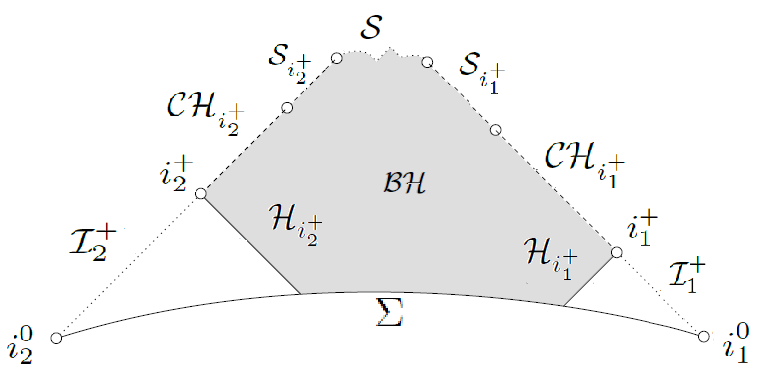}
				
			\end{center}
			
			\caption{General Penrose diagram of a two-ended charged spherically symmetric black hole, Figure from \cite{MihalisSStrapped}.}
			\label{Fig2}
		\end{figure}
	\end{theothree}

	\subsection{First version of the main results} \label{firstversion}
	In this section, we give a first account of our results. More precise statements can be found in section \ref{results}. We start with the $C^2$-inextendibility results, in relation with Conjecture \ref{C2SCC} and we differentiate between the two-ended case -- for which the situation is more straightforward -- and the one-ended case, which is our main interest, as we are motivated by Strong Cosmic Censorship in spherical collapse.
	
	All our results assume that the black hole exterior settles down towards a sub-extremal Reissner--Nordstr\"{o}m black hole, at quantitative rates precisely stated in Theorem \ref{previous}. The sub-extremality condition is conjectured to be generic, c.f. \cite{Kommemi} and the discussion in section \ref{otherext}. The quantitative rates that we assume are also conjectured to be generic, see the discussion in section \ref{decayconj}.
	\subsubsection{Inextendibility in the two-ended case}
	
	\begin{theo}  \label{rough2}
		Given a two-ended solution $(M,g,F,\phi)$ as in Theorem \ref{twoendedapriori} , we assume that both black hole exteriors settle down quantitatively towards a sub-extremal Reissner--Nordstr\"{o}m metric. Then $(M,g)$ is $C^2$-future-inextendible.
	\end{theo}
	
	If we accept that the quantitative decay of the scalar field is generic (see Theorem \ref{previous} for the precise assumptions), then Theorem \ref{rough2} implies directly Conjecture \ref{C2SCC}, i.e.\ the $C^2$ version of Strong Cosmic Censorship for two-ended black holes.
	
	\subsubsection{Inextendibility in the one-ended case}
	In the one-ended case, the situation is more complicated, due to new boundaries emanating from the center of symmetry $\Gamma$, c.f.\ Figure \ref{Fig1}. Nevertheless, one can still prove that the Cauchy horizon is $C^2$ inextendible in the one-ended setting:
	\begin{theo} \label{rough1}
		Given a one-ended solution $(M,g,F,\phi)$ as in Theorem \ref{oneendedapriori}, we assume that the exterior of the black hole settles down quantitatively towards a sub-extremal Reissner--Nordstr\"{o}m metric. Then $\CH$ is $C^2$ inextendible.
	\end{theo}
	While the $C^2$-inextendibility of the Cauchy horizon is valid both in the one-ended and in the two-ended case, it is not sufficient to obtain the $C^2$ version of Strong Cosmic Censorship in the one-ended case. This is because there exists an additional obstruction, coming from the hypothetical extendibility of an outgoing Cauchy horizon $\mathcal{CH}_{\Gamma}$ emanating from the center $\Gamma$. Nevertheless, $\mathcal{CH}_{\Gamma}$ is conjectured to be empty for generic solutions, see section \ref{connected}. If this additional obstruction is not present, we can prove the $C^2$-future-inextendibility of the space-time, as in the two-ended case:\begin{theo} \label{rough1cond}
		Given a one-ended solution $(M,g,F,\phi)$ as in Theorem \ref{oneendedapriori} satisfying the assumptions of Theorem \ref{rough1}, suppose additionally that $\mathcal{CH}_{\Gamma}=\emptyset$. Then $(M,g)$ is $C^2$-future-inextendible.
	\end{theo}
	If we accept that both the quantitative decay of the scalar field and the property $\mathcal{CH}_{\Gamma}=\emptyset$ are generic, then Theorem \ref{rough1} implies directly Conjecture \ref{C2SCC}, i.e.\ the $C^2$ version of Strong Cosmic Censorship for \textbf{one}-ended black holes.
	\subsubsection{Classification of Cauchy horizons, quantitative estimates and strength of the singularity}
	As an important step in our $C^2$-inextendibility proof, we introduce a new classification of the Cauchy horizon into three types. Our main result states that the Cauchy horizon can be divided into one ``static'' connected component which is isometric to Reissner--Nordstr\"{o}m and one ``dynamical'' component -- always to the future of the static one -- which is weakly singular, in the sense that the Hawking mass blows up. A Cauchy horizon is called of dynamical type if its static component is empty, of static type if its dynamical component is empty, and of mixed type otherwise:
	
	\begin{theo} \label{classificationrough}
		Given a one-ended solution $(M,g,F,\phi)$ as in Theorem \ref{oneendedapriori} satisfying the assumptions of Theorem \ref{rough1}, we can classify $\CH$ into three types: \begin{enumerate}
			\item Dynamical type: the Hawking mass blows up everywhere on $\CH$.
			\item Static type: $\CH$ is isometric to a Reissner--Nordstr\"{o}m Cauchy horizon and the Hawking mass is constant.
			\item Mixed type: $\CH$ is the union of two connected components: a ``static component'' including $i^+$, which is isometric to a portion of a Reissner--Nordstr\"{o}m Cauchy horizon, and a ``dynamical'' one on which the Hawking mass blows up.
		\end{enumerate}
	\end{theo}
	\begin{rmk}
		The same statement is true for two-ended solutions as in Theorem \ref{rough2}, if we replace $\CH$ by $\CHone$ or $\CHtwo$.   
	\end{rmk}
	\begin{rmk}
		There exists examples of Cauchy horizons of static type and of mixed type, but it is conjectured that only Cauchy horizons of dynamical type are generic. Proving this result would seemingly necessitate a fully developed scattering theory in the black hole interior, for the Einstein--Maxwell--Klein--Gordon system, which is yet to be discovered.
	\end{rmk}
	
	Note that the main difficulty in Theorem \ref{classificationrough} is to prove that for \textit{any} non-static portions -- i.e.\ for \textit{any} non-trivial ingoing radiation -- the Hawking mass blows up. Since these portions can be quite far from time-like infinity $i^+$, we rely on tailored quantitative estimates and a new continuation criterion to establish the classification of Theorem \ref{classificationrough}.
	
	This classification helps to prove the blow up of curvature, the key ingredient to the $C^2$-inextendibility theorems:
	
	\begin{cor*}
		Given a one-ended solution $(M,g,F,\phi)$ satisfying the assumptions of Theorem \ref{rough1}, quantitative estimates hold in a neighborhood on $\CH$ and $Ric(X,X)$ blows up on $\CH$, for a null radial geodesic vector field $X$ transverse to $\CH$.
	\end{cor*}
	\subsubsection{The trapped region near the Cauchy horizon}
	In addition to $C^2$-inextendibility, we also prove another property of independent interest: the Cauchy horizon is surrounded by the trapped region, see Figure \ref{Figtrapped1}. In particular, the Penrose diagram does \underline{not} contain a ``secondary event horizon'', i.e.\ an outgoing null affine complete hyper-surface reaching the Cauchy horizon. The existence of a trapped neighborhood also implies that the scenario where $\A$ crosses the Cauchy horizon, as depicted in Figure \ref{Fig5}, is \underline{ruled out} under our assumptions.
	\begin{figure}[H]
		
		\begin{center}
			
			\includegraphics[width=97 mm, height=85 mm]{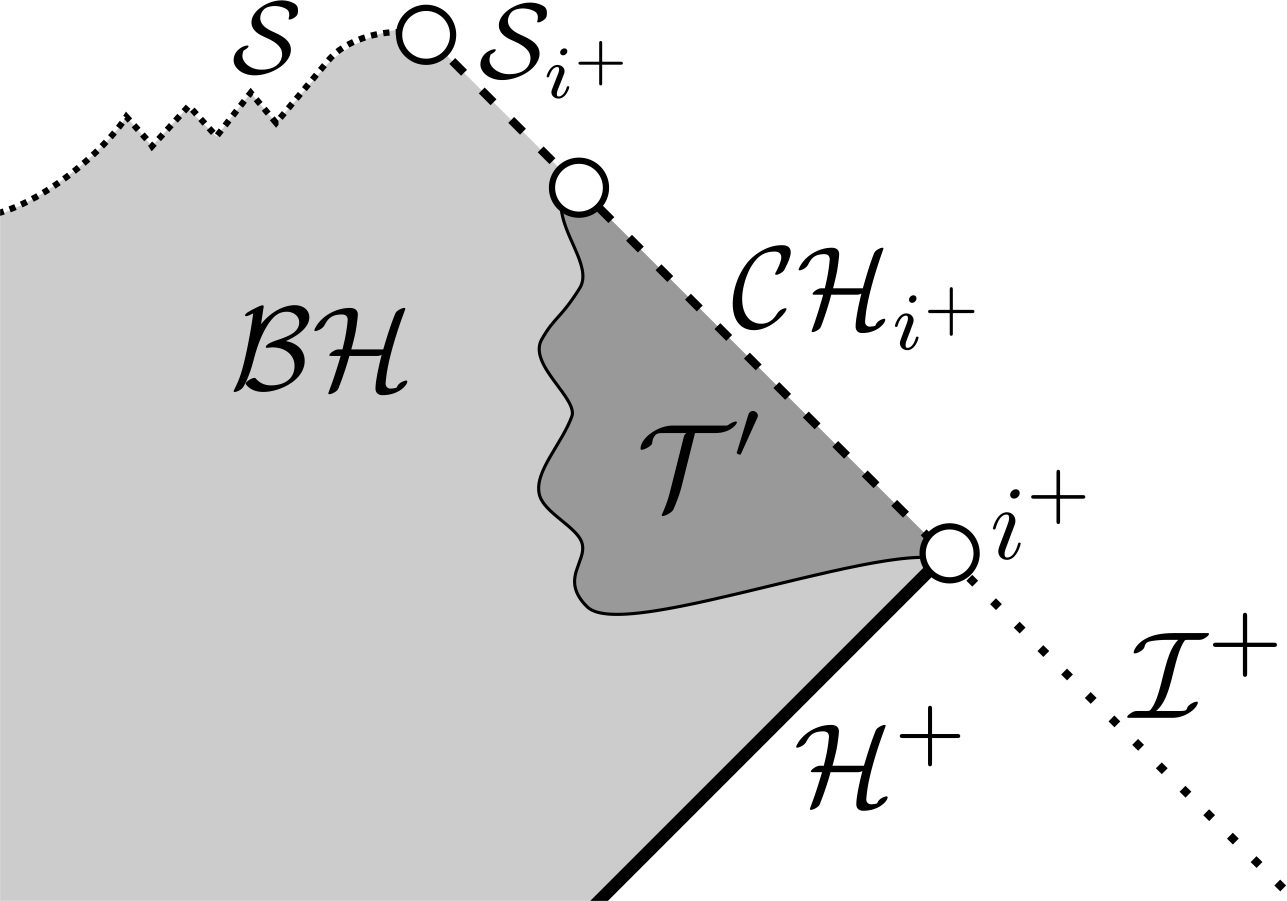}
			
		\end{center}
		
		\caption{Existence of a trapped neighborhood $\mathcal{T}'$ surrounding the Cauchy horizon $\CH$ given by Theorem \ref{trappedrough}.}
		\label{Figtrapped1}
	\end{figure}
	\begin{theo} \label{trappedrough}
		Given a one-ended solution $(M,g,F,\phi)$ as in Theorem \ref{oneendedapriori} satisfying the assumptions of Theorem \ref{rough1}, there exists a neighborhood $\mathcal{T}'$ of $\mathcal{CH}_{i^+}$ inside the trapped region, as in Figure \ref{Figtrapped1}. Therefore, $\A$ has no limit point on $\mathcal{CH}_{i^+}$. \end{theo} \begin{rmk}
		The analogous statement is of course true for two-ended solutions  satisfying the assumptions of Theorem \ref{rough2}.
	\end{rmk}

	\subsubsection{A blow-up criterion which propagates the weak null singularity}
	We now present this continuation criterion: as long as it is satisfied, the Cauchy horizon is static, but when if it fails, then the Hawking mass blows up -- and this blow up is propagated to the future as we shall see. Instead of formulating a continuation criterion, as is traditional in non-linear PDEs, we state a breakdown criterion:
	
	\begin{theo} \label{blowuprough} Given a one-ended solution $(M,g,F,\phi)$ as in Theorem \ref{oneendedapriori} satisfying the assumptions of Theorem \ref{rough1}, assume that the following estimate is true over one outgoing cone $C_{u_0}$ reaching $\CH$: \begin{equation} \label{staticrough}
			\int_{C_{u_0}} \frac{dr}{ \frac{2\rho}{r}-1} <+\infty,
		\end{equation} where $r$ is the area-radius function and $\rho$ is the Hawking mass.
		
		Then on \textbf{all} outgoing cones to the future of $C_{u_0}$ reaching $\CH$, the Hawking mass blows up \underline{point-wise}  towards $\CH$. \end{theo}  \begin{rmk}
		The analogous statement is of course true for two-ended solutions  satisfying the assumptions of Theorem \ref{rough2}.
	\end{rmk} \begin{rmk} \label{rmkblowuprough}
		As we will show, Assumption\eqref{staticrough} implies \textit{a posteriori} \color{black} the non-triviality of ``ingoing radiation'' i.e.\ the field on the Cauchy horizon cannot be identically zero, a fact conjectured to be generic. This is crucial for the mass to blow up.
	\end{rmk}\begin{rmk} \label{monotonicityremark}
		By the Raychaudhuri equation and the null energy condition,  \eqref{staticrough} is propagated to the future. Nonetheless, this ```soft fact'' is useless on its own, and quantitative estimates are necessary to obtain the blow up of the Hawking mass.
	\end{rmk}
	Note that if the Hawking mass blows up  towards $\CH$ on an outgoing cone $C_{u}$ then, in view of the finiteness of $r$, \eqref{staticrough} is satisfied on $C_{u}$. Theorem \ref{blowuprough} shows in particular that the converse is true: \eqref{staticrough} on $C_u$ \textit{implies, under our assumptions}, that the Hawking mass blows up  towards $\CH$ on $C_{u}$ and in fact on all the outgoing cones $C_{u'}$ to the future of $C_u$. \color{black}
	
	This result is the corner stone of  the classification of the Cauchy horizon from  Theorem \ref{classificationrough}.

	\begin{coreo} Given a one-ended solution  $(M,g,F,\phi)$ as in Theorem \ref{oneendedapriori} satisfying the assumptions of Theorem \ref{rough1}, assume the Hawking mass blows up on one outgoing cone $C_{u_0}$ reaching $\CH$.
		
		Then on all outgoing cones reaching $\CH$ to the future of $C_{u_0}$, the Hawking mass blows up point-wise towards $\CH$. \end{coreo}
	\subsection{Previous results for Einstein--Maxwell--Klein--Gordon black holes}	\label{mypreviouswork}
	
	The present paper is preceded by the work of the author \cite{Moi}, \cite{MoiThesis} on the black holes solutions of Einstein--Maxwell--Klein--Gordon. In \cite{Moi}, the non-emptiness of the Cauchy horizon was proven, together with a stability result and quantitative estimates, which laid the groundwork for our present results, and for the study of one-ended solutions in general: \begin{thm}[Stability of the Reissner--Nordstr\"{o}m Cauchy horizon, \cite{Moi}] \label{C0stab} Given a one-ended solution $(M,g,F,\phi)$ as in Theorem \ref{oneendedapriori}, we assume that the exterior of the black hole settles down quantitatively towards a sub-extremal Reissner--Nordstr\"{o}m metric. Then $$\CH \neq \emptyset,$$ and stability estimates are true.
		Moreover, in the case $m^2=0$, $\CH$ is $C^0$-extendible. \end{thm} \begin{rmk}
		Theorem \ref{C0stab} is a semi-\textit{local} result, in a neighborhood of time-like infinity $i^+$, hence it can be formulated in terms on a characteristic initial value problem, with data on the event horizon and an ingoing null cone. In particular, the topology of the manifold is irrelevant, which is why Theorem \ref{C0stab} also applies for two-ended solutions as in Theorem \ref{twoendedapriori}.
	\end{rmk}
	
	Note however that those stability estimates are proven in a weak $L^{\infty}$ norm, consistent with a hypothetical blow up of higher order norms. Indeed, the author proved also in \cite{Moi} the $C^2$ \underline{instability} of $\CH$, using the stability estimates of Theorem \ref{C0stab} in a crucial way. The main estimates of \cite{Moi} show the blow up of some curvature component on a portion of $\CH$ near time-like infinity, which forms a local obstruction to $C^2$-inextendibility:
	
	\begin{thm}[Instability of the Reissner--Nordstr\"{o}m Cauchy horizon, \cite{Moi}] \label{C2instab}  Given a one-ended solution $(M,g,F,\phi)$ as in Theorem \ref{oneendedapriori}, we assume that the exterior of the black hole settles down quantitatively towards a sub-extremal Reissner--Nordstr\"{o}m metric. Then $Ric(X,X)$ blows up on $\CH \cap \mathcal{V}$, where $\mathcal{V}$ is a neighborhood of $i^+$ and $X$ is an outgoing radial null geodesic vector field.	
		
		Moreover, $\phi$ blows up in $H^1_{loc}$ i.e.\ the (non-degenerate) energy of the scalar field on any outgoing trapped cone is infinite. \end{thm}
	
	\begin{rmk}
		The assumptions of Theorem \ref{C2instab} are, in fact, slightly more demanding than those of Theorem \ref{C0stab} in that they require the convergence to Reissner--Nordstr\"{o}m is ``not too fast'', c.f.\ Theorem \ref{previous} for precise assumptions.
	\end{rmk}
	
	While the estimates in \cite{Moi} are local, in a sense that they are valid only on a portion of $\CH$, the result of the present paper is concerned with the \underline{entire} Cauchy horizon $\CH$. While we use local results of \cite{Moi} as a starting point towards global considerations, our proof requires new ideas that go beyond the local aspects near time-like infinity. 
	
	The instability of Theorem \ref{C2instab} relies on the blue-shift of ingoing radiation. Originally, the blue-shift instability was first discovered as a linear mechanism and a consequence of the application of geometric optics in the black hole interior \cite{McNamara}, \cite{Penroseblue}, \cite{JanGaussian}. However, to prove  Conjecture \ref{C2SCC}, it is crucial to work with a \textit{local} version of the blue-shift effect, which is harder to establish but subsists in the non-linear setting, and is then responsible for the blow up of $Ric(X,X)$, see \cite{Moi}. 
	
	The assumptions on the quantitative stability of the black hole exterior were retrieved by the author \cite{Moi2} in the massless charged case $m^2=0$ and in the weakly charged case. While the proof is carried out for the (non-linear) Maxwell-charged-scalar field system \eqref{4}, \eqref{5} on a fixed Reissner--Nordstr\"{o}m background, it should not be difficult to combine the techniques of \cite{Moi2} with those of Luk--Oh \cite{JonathanStabExt} to address \color{black} the full spherically symmetric system \eqref{1}, \eqref{2}, \eqref{3}, \eqref{4}, \eqref{5}, as most of the new difficulties reside in the interaction between the Maxwell field and the charged scalar field: \begin{thm}[Quantitative decay estimates for charged scalar fields with small data, \cite{Moi2}] \label{decaychargedext}
		For regular, spherically symmetric, and small Cauchy data for \eqref{4}, \eqref{5} on a fixed Reissner--Nordstr\"{o}m background, the scalar field decays on the event horizon $\mathcal{H}^+$ at an inverse polynomial rate, in the standard advanced time coordinate $v$ defined by \eqref{gauge2}: $$ \int_{v}^{\infty}|\phi|_{\mathcal{H}^+}^2(v') dv' \lesssim v^{-3+2\delta(q_0e)},$$
		$$ |\phi|_{\mathcal{H}^+}(v_n)  \lesssim v_n^{-2+\delta(q_0e)},$$ where $(v_n)_{n\in \mathbb{N}}$ is a dyadic sequence and  \color{black} $\delta(q_0e)=1-\sqrt{1-4(q_0e)^2}+O( |q_0e|^{\frac{1}{2}})$ as $q_0e \rightarrow 0$ for $e$ the asymptotic charge of the Maxwell field.
	\end{thm} \begin{rmk}
		The upper bound of Theorem \cite{Moi2} corresponds, at the first order, to the decay which is conjectured to be sharp in the literature, i.e.\ $|\phi|_{\mathcal{H}^+}(v) \sim  v^{-1-\sqrt{1-4(q_0e)^2}}$, see \cite{HodPiran1} and the discussion in section \ref{decayconj}.
	\end{rmk}
	The decay mechanism for a charged scalar field is more complex than for its uncharged counterpart. Indeed, in the case of the (uncharged) wave equation, the dynamics are governed by Price's law $|\phi|_{\mathcal{H}^+}(v) \sim  v^{-3}$, see \cite{PriceLaw}, \cite{Pricepaper}. In contrast, in the charged case, the decay rate depends on $q_0e$, i.e.\ the product of the asymptotic Maxwell charge $e$ (a quantity determined in evolution) with the coupling constant $q_0$. This is due to the presence of an inverse square (or ``scale critical'') potential in the charged equation.  Very little is known for such a model in general; to the best of the author's knowledge, decay rates in time depending on parameters or dynamical quantities had never been exhibited before, even for the simplest of such systems i.e.\ the wave equation on Minkowski in the presence of an inverse square potential. See however the series of work \cite{scalecrit1}, \cite{scalecrit2}, \cite{scalecrit2.5}, \cite{scalecrit3}, \cite{scalecrit4} for relatively recent progress on the latter equation, including global well-posedness results.

	\subsection{Previous inextendibility results in the two-ended uncharged case} \label{previouswork}
	
	The Einstein--Maxwell equations in the presence of uncharged matter allow for the existence of Cauchy horizons, but the Maxwell field is static. Therefore, the solutions of these equations are not directly relevant to the dynamics of gravitational collapse; yet they have been studied in the past for the insights they provide on the local behavior of space-time near time-like infinity $i^+$. Here, we present results on two models: the Einstein--Maxwell-null-dust 
	and the Einstein--Maxwell-(uncharged)-scalar-fiel model. 
	The existence of weak null singularities was first revealed for the dust model \cite{Hiscock}, as was the blow-up of the Hawking mass \cite{Ori}, \cite{Poisson}, \cite{PoissonIsrael} -- the famous ``mass inflation scenario''. Nevertheless, the dynamics of dust are governed by a trivial transport equation so it is desirable to study a more sophisticated model.
	
	The wave equation, which governs scalar fields, obeys more complex dynamics, and is more similar to the Einstein equations. Consequently, the non-emptiness of $\CH$, first proven by Dafermos \cite{MihalisPHD}, is non trivial for the Einstein--Maxwell-(uncharged)-scalar-fiel model and constitutes a first essential step. In the same work \cite{MihalisPHD}, \cite{Mihalis1}, Dafermos proves the instability of $\CH$, due to the blow up of the Hawking mass, using the special monotonicity properties of the uncharged model. Note that for his model, the Hawking mass is monotonic so, once a weak null singularity is proved to occur, its propagation is immediate. Finally, the full proof of the $C^2$ version of Strong Cosmic Censorship for two-ended space-times was achieved by Luk and Oh \cite{JonathanStab}, \cite{JonathanStabExt}, who also brought new important insights on the behavior of uncharged scalar fields on the black hole exterior, including inverse polynomial lower bounds on the decay of the scalar field. We now give a detailed account of these different results.

	\subsubsection{Weak null singularities and classification of the Cauchy horizon for the dust model } \label{dustweaknull}
	
	In this section, we discuss spherically symmetric solutions of the Einstein equations in the presence of dust. This will be the opportunity to discuss the classification of Theorem \ref{classificationrough} in a very simplified context (see also Appendix \ref{appendix}) where explicit computations are possible. The Einstein--Maxwell-(uncharged)-null-dust equations are as follows:
	\begin{equation} \label{dust1} Ric_{\mu \nu}(g)- \frac{1}{2}R(g)g_{\mu \nu}= \mathbb{T}_{\mu \nu}^{dust}+\mathbb{T}_{\mu \nu}^{EM},     \end{equation} 	\begin{equation}  \label{dust1.5} \mathbb{T}^{EM}_{\mu \nu}=2\left(g^{\alpha \beta}F _{\alpha \nu}F_{\beta \mu }-\frac{1}{4}F^{\alpha \beta}F_{\alpha \beta}g_{\mu \nu}\right),
	\end{equation}
	\begin{equation} \label{Maxwelldust} \nabla^{\mu} F_{\mu \nu}=0,
	\end{equation} 
	\begin{equation} \label{dust2}
		\mathbb{T}_{\mu \nu}^{dust}= f^2_R \partial_{\mu}u \partial_{\nu}u+
		f^2_L \partial_{\mu}v \partial_{\nu}v,
	\end{equation}
	\begin{equation} \label{dust3}
		g^{\mu \nu}\partial_{\mu}u \partial_{\nu}u=0, \hskip 4 mm g^{\mu \nu}\partial_{\mu}v \partial_{\nu}v=0, 
	\end{equation}
	\begin{equation} \label{dust5}
		g^{\mu \nu}\partial_{\mu}u \partial_{\nu}f_R +\frac{1}{2} (\Box_{g}u)f_R=0,
	\end{equation}
	\begin{equation} \label{dust6}
		g^{\mu \nu}\partial_{\mu}v \partial_{\nu}f_L +\frac{1}{2} (\Box_{g}v)f_L=0.
	\end{equation}As we discussed before, these solutions are necessarily two-ended, a global restriction which nonetheless does not affect the behavior near time-like infinity $i^+$. As written \eqref{dust1}, \eqref{dust1.5}, \eqref{Maxwelldust}, \eqref{dust2}, \eqref{dust3}, \eqref{dust5}, \eqref{dust6} feature a cloud of ingoing null dust of density $f_L$ and a cloud of outgoing null dust of density $f_R$, i.e.\ $f_L$ is transported in the $u$ direction and $f_R$ is transported in the $v$ direction where $u$ and $v$ are eikonal functions (as prescribed by \eqref{dust3}). 
	
	Using $(u,v)$ as a double null coordinate system in the Penrose diagram, it is interesting to work with the null lapse $\Omega^2=-g(\partial_u, \partial_v)$, and $\partial_u r$, $\partial_v r$, where $r$ is the area-radius function. In this gauge, the metric takes the form $$ g=-\Omega^2 du dv+r^2 (d\theta^2+\sin(\theta)^2 d\varphi^2).$$   \begin{rmk} \label{staticMaxwell}
		As the dust is uncharged, \eqref{Maxwelldust} is a \textit{homogeneous} Maxwell equation. In spherical symmetry, this implies that the Maxwell field is ``static'' i.e.\ that $F_{\mu \nu}=  \frac{e}{2r^2}  \cdot \Omega^2 du \wedge dv $, where $e \in \RR$ is the \textit{constant} charge of the black hole.
	\end{rmk} 
	In \cite{Hiscock}, Hiscock studied  \eqref{dust1}, \eqref{dust1.5}, \eqref{Maxwelldust}, \eqref{dust2}, \eqref{dust3}, \eqref{dust5}, \eqref{dust6} in the case of purely ingoing  dust i.e.\ $f_R=0$ and $f_L$ decays at a polynomial rate $v^{-p}$ ($v$ is defined by \eqref{gauge2}) on the event horizon. In Hiscock's model, \color{black}  the Cauchy horizon is already $C^2$-\underline{inextendible}, due to the blow up of one curvature component (see the comments below).  Moreover, certain Christoffel symbols blow up for Hiscock's solution i.e.\ there exists a ``reasonable'' coordinate system which is \footnote{This statement does \textbf{not} prove that the metric is $C^1$-inextendible but does give the insight that a breakdown occurs already at the $C^1$ level.} not $C^1$. Nevertheless, in the absence of outgoing \color{black} radiation, the Hawking mass and the Kretschmann scalar are finite. In fact,  the non-staticity condition \eqref{staticrough} is \textit{violated} everywhere and the Cauchy horizon is isometric to a Reissner--Nordstr\"{o}m Cauchy horizon. This situation corresponds to what we called a Cauchy horizon of \underline{static} type, in the language of Theorem \ref{classificationrough}.\color{black}

	We now come back to the general case. \color{black} The relations between the mass $\rho$ and the gradient of $r$ (see section \ref{geometricframework}) allow us to formulate the non-staticity condition \eqref{staticrough} as  $\frac{|\partial_u r|}{\Omega^2} \in L^1(C_{u_0},dv)$, for $C_{u_0}$ an outgoing cone reaching $\CH$: \begin{equation} \label{staticityrough2}
		\int_{C_{u_0}} \frac{4|\partial_u r|}{\Omega^2} dv =	\int_{C_{u_0}} \frac{dr}{ \frac{2\rho}{r}-1} <+\infty.
	\end{equation}

	Now, there are three possibilities, according to the behavior of the cloud of outgoing dust $f_R$, entirely and trivially determined by its initial data $f_R^0$ on an ingoing cone $\underline{C}_{v_0}$ (the behavior of the ingoing dust $f_L$ is irrelevant to this discussion): \begin{enumerate}[I]
		\item \label{I} $f_R^0 \equiv 0$ on $\underline{C}_{v_0} $: then \eqref{staticityrough2} is violated everywhere on the Cauchy horizon. This situation corresponds to a Cauchy horizon of \underline{static} type, see Definition \ref{rigiddef} (Hiscock's space-times are such examples).
		\item \label{II}For all $u_s \in \RR$, $\{f_R^0 \neq 0\} \cap \{ u \leq u_s \} \neq \emptyset$: then \eqref{staticityrough2} is satisfied \footnote{This is consequence of the Raychaudhuri equation and the dominant energy condition, as the quantity $\frac{4|\partial_u r|}{\Omega^2}$ is monotonic c.f. \eqref{RaychUdust}. } everywhere on the Cauchy horizon. This situation corresponds to a Cauchy horizon of \underline{dynamical} type, see Definition \ref{dynamicaldef}.
		\item \label{III}There exists $u_T \in \RR $ such that $f_R^0 \equiv 0$ on $\underline{C}_{v_0} \cap \{ u\leq u_T\}$, but $\{f_R^0 \neq 0\} \cap \{u_T< u \leq u_{T}+\epsilon\} \neq \emptyset$ for all $\epsilon>0$. Then \eqref{staticityrough2} is satisfied only on outgoing cones $C_{u'}$ with $u'>u_T$. This corresponds to a Cauchy horizon of \underline{mixed} type, see Definition \ref{mixeddef}.
	\end{enumerate} \begin{rmk}
		The correspondence between respectively statements \ref{I}, \ref{II}, \ref{III} and Definitions \ref{rigiddef}, \ref{dynamicaldef}, \ref{mixeddef} is not a priori obvious, but it follows from the (comparatively easier) Proposition \ref{propappendix}. The main mechanism is provided by the Raychaudhuri equation \eqref{RaychUdust}, which essentially dictates that $r_{\CH}$ is constant on $(-\infty,u]$ if and only if $\int_{-\infty}^{u} (f^0_R)^2(u') du' =0$.
	\end{rmk}
	Note that the propagation of $f_R$ is a trivial translation by \eqref{dust5}, thus zero data corresponds to zero radiation at the Cauchy horizon. These three types of Cauchy horizons are easy to construct for the dust model, see Appendix \ref{appendix}.
	\color{black}
	\begin{rmk}
		Note that the classification of the Cauchy horizon in the case of dust is immediate. However, in the presence of a scalar field, that has non-trivial reflectivity, this classification requires a machinery of quantitative estimates, to finally reach the result of Theorem \ref{classificationrough} and the continuation criterion of Theorem \ref{blowuprough}, in turn responsible for $C^2$-inextendibility.
	\end{rmk}
	
	Still under Hiscock's assumption that $f_L$ decays at an polynomial\footnote{In fact, the space-time is $C^2$ extendible and the Hawking mass finite if $f_R$ decays exponentially at a sufficiently fast rate. This phenomenon explains why in the cosmological setting, mass inflation is not expected for a certain range of parameters  c.f.\  \cite{Costa}.} rate $v^{-p}$ on the event horizon\color{black}, it is important to notice that in the three cases \ref{I}, \ref{II} and \ref{III}, the Cauchy horizon is \textbf{$C^2$-inextendible} due to the blow up \footnote{We emphasize however that this blow up was not formulated in either \cite{Poisson}, \cite{PoissonIsrael} or \cite{Ori}. This modern formulation is due to Luk and Oh \cite{JonathanStab}.} of the transverse curvature component $Ric(X,X)$, for a null outgoing radial geodesic vector field $X$. This is because the ingoing radiation $f_L$ is blue-shifted by the Cauchy horizon, a phenomenon which is present even in the static case \ref{I} of Hiscock; a similar logic governs the charged scalar field model, see  Remark \ref{remarktransverse}.

	The next natural question is: ``what happens to the mass in either of the cases \ref{II} or \ref{III} ?'' (for case \ref{I} we already saw that the Hawking mass is finite). \color{black}Poisson and Israel in \cite{Poisson}, \cite{PoissonIsrael} and Ori in \cite{Ori} discovered that in case \ref{II} and case \ref{III}, still under Hiscock's assumption that $f_L$ decays at an polynomial rate\color{black}, the Hawking mass $\rho$ blows up on $\CH$, in contrast with the Hiscock model. In Appendix \ref{appendix}, we revisit their computation and establish a connection with our new classification.

	\subsubsection{Global $C^2$-inextendibility and Strong Cosmic Censorship in the two-ended case}   \label{unchargedsection}
	In this section, we mention previous results in spherical symmetry for the Einstein--Maxwell-(uncharged)-scalar-field:  \begin{equation} \label{EMSF1} Ric_{\mu \nu}(g)- \frac{1}{2}R(g)g_{\mu \nu}= \mathbb{T}^{EM}_{\mu \nu}+  \mathbb{T}^{SF}_{\mu \nu} ,    \end{equation} 
	\begin{equation}\label{EMSF2} \mathbb{T}^{EM}_{\mu \nu}=2\left(g^{\alpha \beta}F _{\alpha \nu}F_{\beta \mu }-\frac{1}{4}F^{\alpha \beta}F_{\alpha \beta}g_{\mu \nu}\right),
	\end{equation} \begin{equation} \label{EMSF3} \mathbb{T}^{SF}_{\mu \nu}= 2\left( \Re(\partial_{\mu}\phi \partial_{\nu}\phi) -\frac{1}{2}(g^{\alpha \beta} \partial_{\alpha}\phi \partial_{\beta}\phi  )g_{\mu \nu} \right), \end{equation} \begin{equation} \label{EMSF4} \nabla^{\mu} F_{\mu \nu}=0, \; F=dA , \end{equation} \begin{equation} \label{EMSF5}  g^{\mu \nu} \partial_{\mu} \partial_{\nu}\phi =0.
	\end{equation} \begin{rmk} \label{staticMaxwell2}
		The scalar field is uncharged, hence $F_{\mu \nu}=  \frac{e}{2r^2}  \cdot \Omega^2 du \wedge dv $, $e\in \RR$ as in the dust case, c.f.\ Remark \ref{staticMaxwell}.
	\end{rmk}
	
	Generalizing the results on null dust to a scalar field is, needless to say, a complex task. This is because scalar fields obey more sophisticated dynamics, involving a mechanism of transmission-reflection. A non-linear scattering theory of the system \eqref{EMSF1}, \eqref{EMSF2}, \eqref{EMSF3}, \eqref{EMSF4}, \eqref{EMSF5} in the interior black hole -- even in spherical symmetry -- is not currently available (see however \cite{ChristophYakov} for results on the \textit{linear} theory for the wave equation on a Reissner--Nordstr\"{o}m interior).
	
	Nevertheless, it is still possible to study the equations \eqref{EMSF1}, \eqref{EMSF2}, \eqref{EMSF3}, \eqref{EMSF4}, \eqref{EMSF5} as a system of coupled non-linear PDEs and employ stability methods to establish the decay of the scalar field, from which we show that the metric converges to Reissner--Nordstr\"{o}m towards time-like infinity $i^+$. 
	
	The first result in this direction is due to Dafermos \cite{MihalisPHD}, \cite{Mihalis1}, who proved the stability of the Reissner--Nordstr\"{o}m Cauchy horizon in spherical symmetry under decay assumptions on the scalar field on the event horizon:
	\begin{thmunb}[Dafermos \cite{MihalisPHD}, \cite{Mihalis1}]
		Assume that for $p>1$ \color{black}, the asymptotic behavior of the event horizon is given by: \begin{equation} \label{Dafermosbound}
			D^{-1} \cdot v^{-p} \leq 	\partial_v \phi \leq  D \cdot v^{-p},
		\end{equation} for some $D>0$, in the advanced time coordinate $v$ defined by gauge \eqref{gauge1}. Then $$\CH \neq \emptyset,$$ and the space-time is $C^0$-extendible\color{black}. Moreover, on all outgoing cones reaching $\CH$, the Hawking mass blows up point-wise towards $\CH$ and the space-time is $C^2$-inextendible\color{black}.
	\end{thmunb} \begin{rmk}
		In fact, the assumption $p>\frac{1}{2}$ is sufficient to prove that $\CH \neq \emptyset$ and the mass inflation, but not to obtain $C^0$-extendibility of the metric (even though the area-radius $r$ extends as a continuous scalar under this weaker assumption). Note that this discussion is purely academic, since for Dafermos' model we have $p=3$, see \cite{Newrp}, \cite{PriceLaw}. \color{black}
	\end{rmk}In reality, the work of Dafermos consists in two distinct results: the Reissner--Nordstr\"{o}m Cauchy horizon is $C^0$ stable but is $C^1$ \textbf{un}stable, in the sense that the Hawking mass blows up on $\CH$. Both results were a priori surprising. A posteriori, the stability result is due to the repulsive effect of the charge of the Maxwell field (which back-reacts by the Einstein equations), and the instability is due to the (linear) amplification of ingoing radiation near $\CH$ -- the (already mentioned) blue-shift effect. It is remarkable that the linear $C^1$ instability persists in the non-linear setting, in part thanks to the strength of the $C^0$ stability estimates. In turn, the blow up of the Hawking mass implies the blow up of the Kretschmann scalar, thus the space-time is $C^2$-future-inextendible. However, the blow up of the Hawking mass relies on a monotonicity argument, which is not robust and also requires the lower bound of \eqref{Dafermosbound}, which has been conjectured but not verified for any non-linear solution in the black hole exterior. Nonetheless, upper bounds consistent with \eqref{Dafermosbound}, the so-called Price's law, were established by Dafermos and Rodnianski \cite{PriceLaw}. These bounds are sufficient to prove that $\CH$ is $C^0$-extendible and thus falsify the $C^0$ version of Strong Cosmic Censorship in spherical symmetry:
	\begin{thmunb}[Dafermos \cite{MihalisPHD}, \cite{Mihalis1}, Dafermos--Rodnianski \cite{PriceLaw}]
		Conjecture \ref{C0SCC} is \underline{false} for the Einstein--Maxwell-(uncharged)-scalar-field model $(q_0=0)$ in spherical symmetry.
	\end{thmunb} The full proof of $C^2$-future-inextendibility for generic spherically symmetric two-ended Cauchy data was ultimately achieved by Luk and Oh \cite{JonathanStab}, \cite{JonathanStabExt}. Remarkably, they do not prove directly the blow up of the Hawking mass: instead, they rely on the blow up of the geometric quantity $Ric(X,X)$, for $X$ a null radial geodesic vector field transverse to $\CH$, which is sufficient to guarantee $C^2$-inextendibility:
	
	\begin{thmunb}[Luk--Oh \cite{JonathanStab}, \cite{JonathanStabExt}]
		Conjecture \ref{C2SCC} is \underline{true} for the Einstein--Maxwell-(uncharged)-scalar-field model in spherical symmetry.
	\end{thmunb}
	One of the key elements of Luk and Oh's proof is to establish that Price's law is sharp, at least in the $L^2$ sense. To reach this conclusion, they established the first lower bounds for the wave equation on a black hole, and in the non-linear setting. Note that lower bounds and even precise tails were later obtained, on a fixed Reissner--Nordstr\"{o}m background by Angelopoulous, Aretakis and Gajic \cite{Newrp}, \cite{Latetime}.

	\subsection{Connected problems, conjectures and additional results} \label{connected}
	\subsubsection{Asymptotic decay on the black hole exterior} \label{decayconj}
	
	In this sub-section, we discuss the conjectured decay rate at which a black hole is expected to settle down towards a sub-extremal Reissner--Nordstr\"{o}m space-time for large times, and we present some related heuristic or numerical works.
	
	The decay of charged scalar fields on spherically symmetric black holes was first considered in \cite{HodPiran1}, where the authors provided a heuristic argument to conjecture the correct late time tail. They argued that the main difference with uncharged fields is that the decay rate depends on the black hole charge, as opposed to the universal rate prescribed by Price's law in the uncharged case. The results of \cite{HodPiran1} were also later backed up by the numerics of Oren and Piran \cite{OrenPiran}:
	\begin{conjecture} [Decay of charged scalar fields, Hod and Piran \cite{HodPiran1}, Oren and Piran \cite{OrenPiran}] \label{chargedconj} For smooth, regular, generic admissible data for which the black hole is non-empty , we have, in the charged massless case $m^2=0$:
		$$ \phi_{|\mathcal{H}^+}(v) \sim  C_H(q_0e,D) \cdot e^{ \frac{iq_0 e}{r_+} v} \cdot  v^{-2+\delta(q_0e)},$$ where $e$ is asymptotic charge of the black hole at time-like infinity, $\delta(q_0e):= 1-\Re(\sqrt{1-4(q_0e)^2}) \in [0,1)$ and $v$ is the standard advanced time null coordinate defined by the gauge condition \eqref{gauge2}.
	\end{conjecture}
	
	The upper bound corresponding to conjecture \ref{chargedconj} was proven mathematically in \cite{Moi2}, on a fixed Reissner--Nordstr\"{o}m background, for small charge $q_0e$ and for a rate $p= 2-\delta(q_0e)+o( \sqrt{|q_0e|})$ as $q_0e \rightarrow 0$, see Theorem \ref{decaychargedext}.
	
	Now we turn to the case of a massive uncharged scalar field, studied in \cite{KoyamaTomimatsu} heuristically, and backed up by the numerics of Burko and Khanna \cite{BurkoKhanna}. It was also argued in \cite{KonoplyaZhidenko} that the same tail holds for a massive charged scalar field:
	
	\begin{conjecture} [Decay of uncharged massive scalar fields \cite{BurkoKhanna}, \cite{KoyamaTomimatsu} or charged massive scalar fields \cite{KonoplyaZhidenko}] \label{conjecturemassive} For smooth, regular, generic admissible data for which the black hole is non-empty, we have, in the massive  case $m^2 \neq 0$, $q_0 \in \RR$: 
		$$ \phi_{|\mathcal{H}^+}(v) \sim C(M \cdot m,D) \cdot  e^{ \frac{iq_0 e}{r_+} v} \cdot \sin( m v + \omega_{err}(v)) \cdot v^{-\frac{5}{6}},$$ where $\omega_{err}(v)= -\frac{3m}{2} (2\pi M)^{\frac{2}{3}}  v^{\frac{1}{3}} + \omega(m \cdot M) $ and $v$ is the standard advanced time null coordinate defined by the gauge condition \eqref{gauge2}.
	\end{conjecture}

	\subsubsection{Weak Cosmic Censorship and the spherical trapped surface conjecture}
	
	In addition to the Strong Cosmic Censorship, one of the most discussed open problems in General Relativity is the Weak Cosmic Censorship Conjecture. Its statement is that ``naked'' singularities are non generic. A ``naked singularity'' can be defined in modern terms as a space-time for which null infinity $\mathcal{I}^+$ is incomplete: we can then formulate the conjecture: \begin{conjecture}[Weak Cosmic Censorship Conjecture for the Einstein--Maxwell--Klein--Gordon model] \label{WCC}
		Among all the data admissible from Theorem \ref{oneendedapriori}, there exists a generic sub-class for which $\mathcal{I}^+$ is complete.
	\end{conjecture} 
	Conjecture \ref{WCC} was solved in the special case $F \equiv0$, $m^2=0$ in the monumental series of Christodoulou \cite{Christo1}, \cite{Christo2}, \cite{Christo3}, but is still an open problem in general. His proof of Weak Cosmic Censorship relies on a local approach near a singular $b_{\Gamma}$. Christodoulou proves in the special case $F \equiv0$, $m^2=0$ the general statement that a sequence of trapped surfaces must asymptote to $b_{\Gamma}$. We formulate the analogous result in the charged case as a conjecture, directly implying Conjecture \ref{WCC}: 
	\begin{conjecture}[Spherical trapped surface conjecture, as formulated in \cite{Kommemi}] \label{trappedsurfaceconj}
		Among all the data admissible from Theorem \ref{oneendedapriori}, there exists a generic sub-class for which if the maximal future development has $\mathcal{Q}^+ \cap J^{-}(\mathcal{I}^+) \neq \emptyset$, then the apparent horizon $\A$ has a limit point on $b_{\Gamma}$. If that is the case, we further conjecture that $\mathcal{S}_{\Gamma}^{1}= \mathcal{CH}_{\Gamma} =\mathcal{S}_{\Gamma}^{2}=\emptyset$.\color{black}
	\end{conjecture}\begin{rmk}
		The statement $\mathcal{S}_{\Gamma}^{1}= \mathcal{CH}_{\Gamma} =\mathcal{S}_{\Gamma}^{2}=\emptyset$ corresponds to the absence of a ``locally naked singularity'' emanating from $b_{\Gamma}$, the end-point of the center of symmetry. This statement is slightly stronger than Conjecture \ref{WCC}.
	\end{rmk}
	
	This conjecture is important for the present manuscript, as the main assumption of our result in Theorem \ref{rough1cond} is that $\mathcal{CH}_{\Gamma}=\emptyset$. However, Conjecture \ref{trappedsurfaceconj} is related to the behavior of space-time in the vicinity of $b_{\Gamma}$, therefore, by causality, that behavior cannot be influenced by the late time tail on the event horizon, which is our only assumption. Therefore, a completely different approach would be required to solve Conjecture \ref{trappedsurfaceconj} -- together with Conjecture \ref{WCC} -- and show that the assumption of Theorem \ref{rough1} is indeed satisfied generically.  
	
	\subsubsection{The breakdown of weak null singularities and the r=0 singularity conjecture}
	
	Another interesting problem is to characterize the singularities in the black hole interior during gravitational collapse. In the present paper, we focus on the Cauchy horizon and proved the presence of a global weak null singularity under assumptions conjectured to be generic. With a different focus, the author has also proven in \cite{r=0} that, during gravitational collapse -- i.e.\ for one-ended solutions as in Theorem \ref{oneendedapriori} -- the weakly singular Cauchy horizon \underline{necessarily} breaks down:
	\begin{thm}[Breakdown of weak null singularities, \cite{r=0}] \label{breakdown} For initial data as in Theorem \ref{oneendedapriori}, assume there exists \textbf{one} trapped cone reaching $\CH$ on which the Hawking mass $\rho$ blows up, while the matter fields are bounded.
		Then $$\mathcal{S}^1_{\Gamma} \cup \mathcal{CH}_{\Gamma} \cup  \mathcal{S}^2_{\Gamma} \cup  \mathcal{S}  \neq \emptyset,$$ i.e.\ $\CH \cup \mathcal{S}_{i^+}$ cannot close off the space-time at $b_{\Gamma}$, in particular the Penrose diagram of Figure \ref{Fig3} is impossible.
	\end{thm}
	\begin{figure}[H]
		
		\begin{center}
			
			\includegraphics[width=86 mm, height=75 mm]{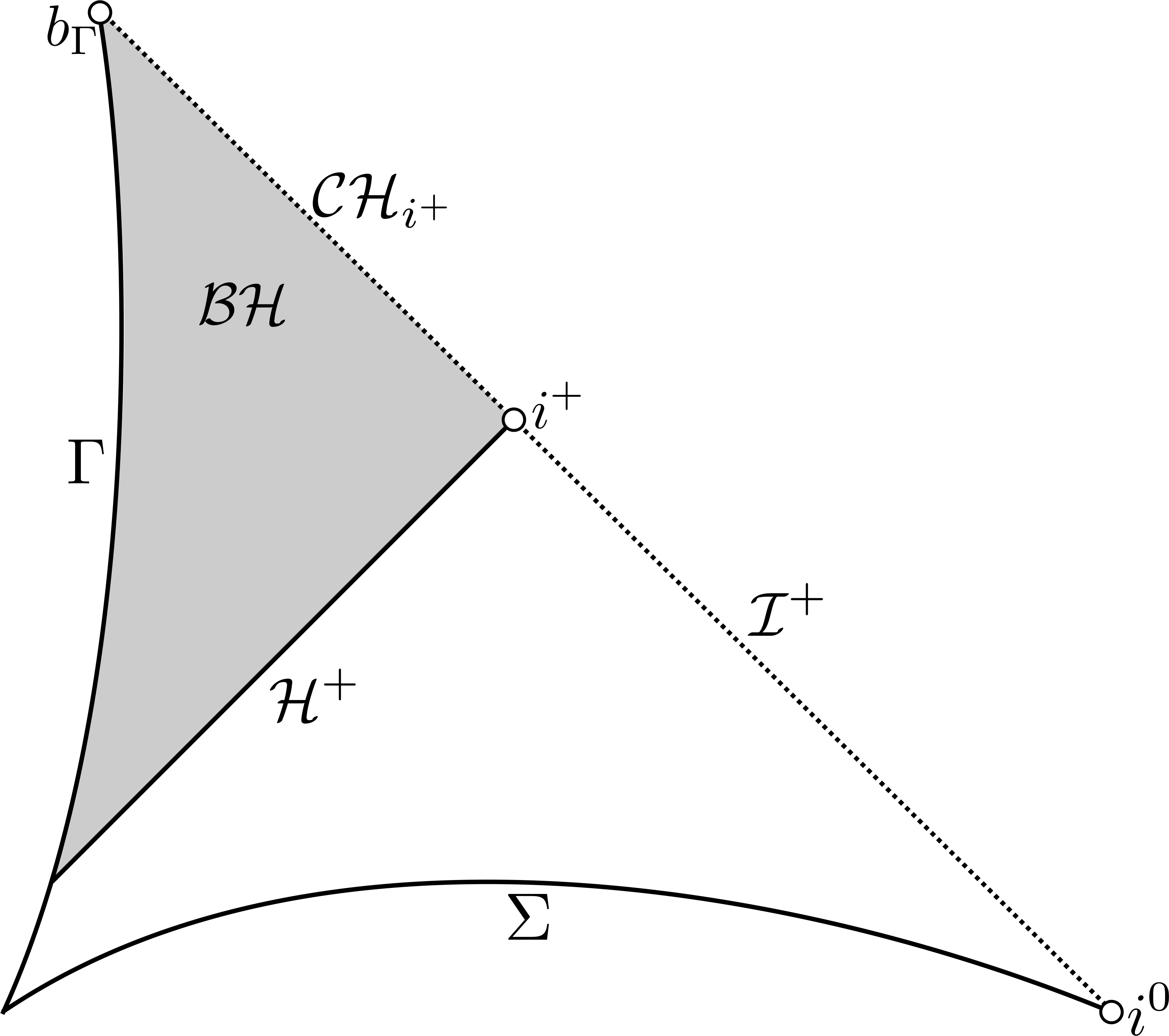}
			
		\end{center}
		
		\caption{Penrose diagram whose existence is disproved in \cite{r=0} if $\CH$ is weakly singular.}
		\label{Fig3}
	\end{figure}
	This systematic break-down is a global phenomenon and involves the centre of symmetry $\Gamma$: for instance, weak null singularities \underline{do not} systematically break-down for two-ended solutions \cite{Mihalisnospacelike}. Note however that the global structure of two-ended solutions is of little significance to the study of gravitational collapse. Since the weakly singular Cauchy horizon breaks down, what does the rest of the interior boundary look like ? It is often conjectured in the literature that the other part of the boundary is a singularity $\mathcal{S}$ on which $r=0$. We state a version of this conjecture present in \cite{Kommemi}:	\begin{conjecture}[$r=0$ singularity conjecture, as formulated in \cite{Kommemi}] \label{spacelikeconj}
		Among all the data admissible from Theorem \ref{oneendedapriori}, there exists a generic sub-class for which if the maximal future development has $\mathcal{Q}^+ \cap J^{-}(\mathcal{I}^+)\neq \emptyset$, then the Penrose diagram is given by Figure \ref{Figconj} i.e.\ $\mathcal{S} \neq \emptyset$, $\CH \neq \emptyset$ and $\mathcal{S}_{\Gamma}^{1}= \mathcal{CH}_{\Gamma} =\mathcal{S}_{\Gamma}^{2} =\emptyset$.
	\end{conjecture}

	Assuming Conjecture \ref{trappedsurfaceconj} -- a slightly stronger result than Weak Cosmic Censorship -- the author has given a proof of this conjecture in \cite{r=0}. This result comes a consequence of break-down of weak null singularities of Theorem \ref{breakdown}:
	\begin{thm}[Generic existence of $r=0$ singularities, \cite{r=0}] \label{spaceliketheorem} Given a one-ended solution $(M,g,F,\phi)$ as in Theorem \ref{oneendedapriori}, we assume that the exterior of the black hole settles down quantitatively towards a sub-extremal Reissner--Nordstr\"{o}m metric and that $\mathcal{S}_{\Gamma}^1=\mathcal{CH}_{\Gamma}=\mathcal{S}_{\Gamma}^2=\emptyset$. Then the Penrose diagram is given by Figure \ref{Figconj}, i.e.\ $ \mathcal{S}  \neq \emptyset$, $\CH \neq \emptyset$, $\mathcal{S}^1_{\Gamma} = \mathcal{CH}_{\Gamma} =  \mathcal{S}^2_{\Gamma}=\emptyset$. \begin{figure}
			
			\begin{center}
				
				\includegraphics[width=107 mm, height=65 mm]{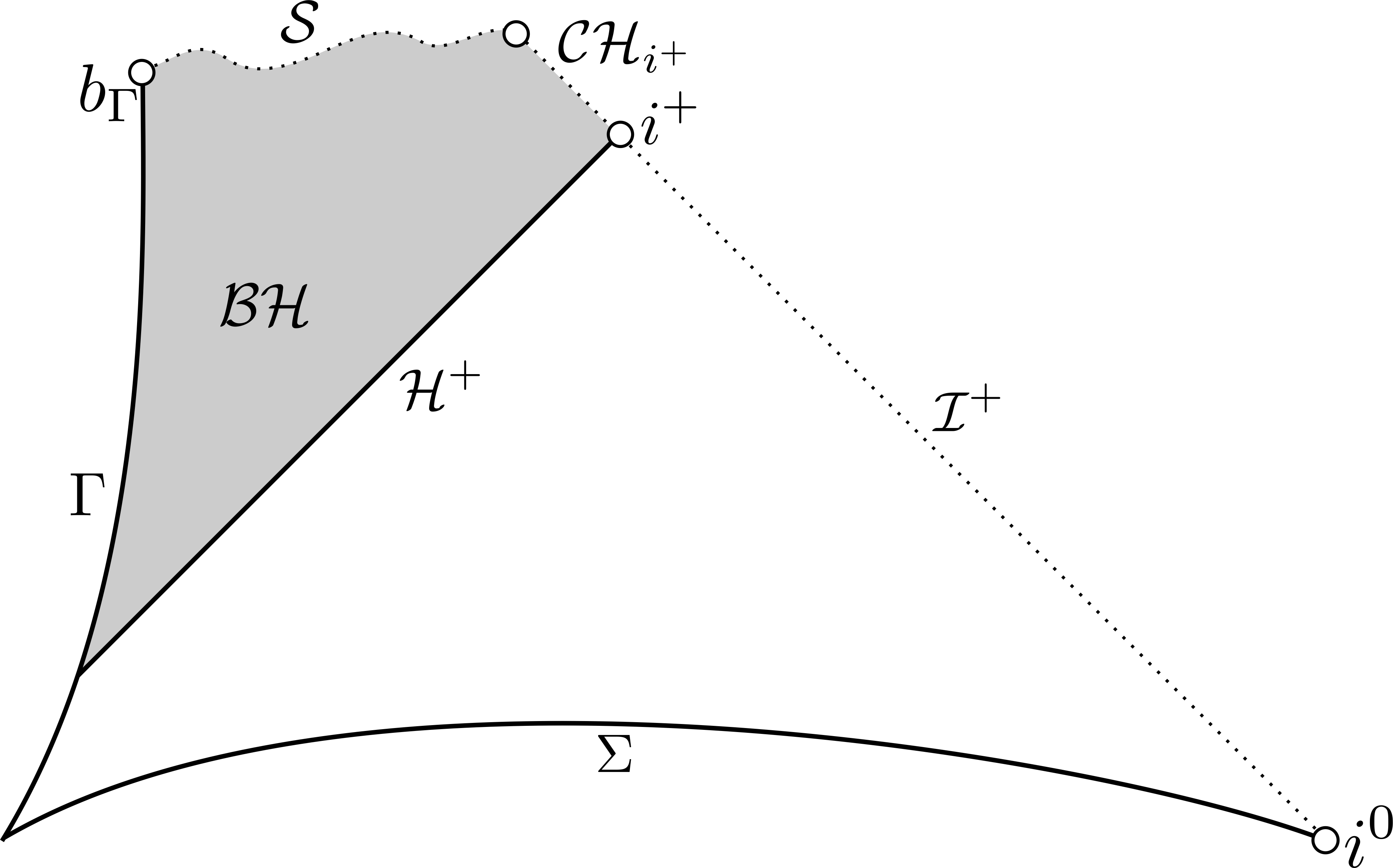}
				
			\end{center}
			
			\caption{\textbf{Generic} Penrose diagram of a one-ended charged black hole under the assumptions of Theorem \ref{spaceliketheorem}, \cite{r=0}.}
			\label{Figconj}
		\end{figure}
	\end{thm}
	
	\subsubsection{Other extendibility/inextendibility results} \label{otherext}
	
	\paragraph{Space-like singularities and $C^0$-inextendibility} For the spherically symmetric model of Christodoulou, i.e. \eqref{1}, \eqref{2}, \eqref{3}, \eqref{4}, \eqref{5} in the special case $F \equiv 0$, $m^2=0$, $\mathcal{S}$ is generically the only non-trivial boundary component in the black hole interior and is ``space-like'' \cite{Christo1}, \cite{Christo2}, \cite{Christo3}. It is conjectured in the literature that Christodoulou's space-times are continuously inextendible, i.e.\ that Conjecture \ref{C0SCC} is true for the Einstein-(uncharged)-scalar field model $(F \equiv 0)$. This conjecture is motivated by the presence of the $r=0$ singularity $\mathcal{S}$ which triggers the blow up of certain tidal deformations of every in-falling observers. The only existing result in that direction is due to Sbierski \cite{JanC0} who proved $C^0$ inextendibility of the Schwarzschild solution, which features the same $r=0$ space-like singularity $\mathcal{S}$ as the Christodoulou black holes.
	
	\paragraph{$C^0$-extendibility of the Cauchy horizon} However, it is well known that the (conjectured) $C^0$-inextendibility of Christodoulou's solutions is an artifact of the model, as black holes arising from gravitational collapse are conjectured to possess a Cauchy horizon, due to the repulsive effect of angular momentum --a feature which is absent in Christodoulou's model. Indeed, Dafermos proved the non-emptiness of a Cauchy horizon and its $C^0$-extendibility \cite{Mihalis1}, \cite{MihalisPHD} for the Einstein--Maxwell-(uncharged)-scalar-field model in spherical symmetry,  i.e. \eqref{1}, \eqref{2}, \eqref{3}, \eqref{4}, \eqref{5} in the special case $q_0 = 0$, $m^2=0$: thus Conjecture \ref{C0SCC} is false, see section \ref{unchargedsection}. Later, the author proved in \cite{Moi} that Conjecture \ref{C0SCC} is also false for the spherical collapse of a charged scalar field, i.e.\ \eqref{1}, \eqref{2}, \eqref{3}, \eqref{4}, \eqref{5} in the special case $m^2=0$, under assumptions on the exterior consistent with Conjecture \ref{chargedconj}. The same result was later reached in the massive case $m^2 \neq 0$ by Kehle and the author \cite{Moi3Christoph}, under assumptions on the exterior consistent with Conjecture \ref{conjecturemassive}. We also mention the monumental work of Dafermos and Luk \cite{KerrStab} in which Conjecture \ref{C0SCC} is falsified, for perturbations of Kerr black holes in vacuum, in the absence of symmetry, and under assumptions that are conjectured to hold in the black hole exterior. \color{black}
	\paragraph{$C^0$ extendible Cauchy horizons with a null contraction singularity}
	
	While $r=0$ singularities are associated with $C^0$ extendibility, it is often conjectured that Cauchy horizons -- i.e.\ null boundaries on which $r$ is bounded away from zero -- are \textit{always} $C^0$-extendible, as there is no obvious mechanism inducing the blow up of tidal deformations if $r>0$. It is possible to prove that this is true for data with ``a reasonable'' decay rate on the event horizon \cite{Moi4Christoph}. However, for a large class\footnote{Essentially, such data decay weakly and are non-oscillating, so do not obey the asymptotics of Conjecture \ref{conjecturemassive} (non-generic behavior).} of data on the event horizon (conjectured to arise from a non-empty, but non-generic set of regular Cauchy data), the author, in \cite{MoiThesis}, and with Kehle in \cite{Moi4Christoph} discovered a new singularity at the Cauchy horizon, for which there exists no ``$C^0$-admissible'' extension, a notion invented by Moschidis \cite{Moschidis}. This instability, which we call \textit{null contraction} (see \cite{Moi4Christoph}), is at the level of metric components, and is triggered by the point-wise blow up of the scalar field at the Cauchy horizon.
	
	\paragraph{Extendibility results for black holes approaching Schwarzschild or extremality}
	The $C^2$-inextendibility results of Theorem \ref{rough2} and Theorem \ref{rough1} only apply when the black hole exterior settles down towards a sub-extremal Reissner--Nordstr\"{o}m space-time, i.e.\ that the black hole charge converges to a non zero and non-extremal value. This situation is conjectured to be generic \cite{Kommemi}. Nevertheless it is interesting to understand what happens both for a black hole converging to Schwarzschild -- i.e.\ when the asymptotic charge is zero -- and for a black hole converging to extremality , as those are limit cases. The author has proved in \cite{MoiThesis} that, if the asymptotic charge is zero then the Cauchy horizon $\CH$ is \underline{empty}, thus $r=0$ on the whole boundary and the space-time is $C^2$-future-inextendible, under the same assumptions as in Theorem \ref{rough2} or Theorem \ref{rough1}. As $r=0$ on the whole boundary, one may even expect that the space-time is also $C^0$-inextendible as in the Schwarzschild case, but this question remains open. In the extremal limit, we mention the result of Gajic and Luk \cite{extremeJonathan} who prove $H^1$ extendibility of the solution, and the absence of a weak null singularity, i.e.\ the finiteness of the Hawking mass. Whether their space-times are inextendible or not in a stronger norm remains an open problem.

	\subsection{Methods and strategy of the proof}  \label{outline}
	
	The main objective of the present paper is to prove that $\CH$, the Cauchy horizon emanating from time-like infinity, is $C^2$-future-inextendible (Theorem \ref{rough2}, Theorem \ref{rough1} and Theorem \ref{rough1cond}). There are two known strategies to obtain $C^2$-inextendibility:\begin{itemize}
		\item by the blow-up of the Hawking mass (triggering the blow up of the Kretschmann scalar);
		\item by the blow-up of $Ric(X,X)$, where $X$ is an null radial geodesic vector field which is transverse to $\CH$.
	\end{itemize} The Hawking mass does not blow up uniformly, due to the existence of Cauchy horizon of static and mixed type so it cannot be used on its own to prove $C^2$-inextendibility. Nevertheless, an alternative strategy would be to prove the blow up of $Ric(X,X)$ over the ``static parts''  of Cauchy horizons of static or mixed types, and use the blow up of the mass for the other part. We make a different choice and rely on the blow up of $Ric(X,X)$ on the entire Cauchy horizon instead to prove $C^2$-inextendibility in all three cases with the same method. While propagating the blow up of $Ric(X,X)$ over the non-static parts is technically more involved, we also obtain other global properties of the Cauchy horizon in this process, and we derive quantitative estimates which are of independent interest \footnote{e.g.\ they are important for the $C^0$-inextendibility result of \cite{Moi3Christoph}.}.
	
	
	
	The $C^2$-inextendibility of $\CH$ results from the classification of the Cauchy horizon into static, mixed or dynamical type and the associated quantitative estimates (Theorem \ref{classificationrough} and its corollary), eventually triggering the blow up of $Ric(X,X)$. 
	
	In turn, the classification relies on the existence of a trapped neighborhood $\mathcal{T}'$ of the Cauchy horizon $\CH$ (Theorem \ref{trappedrough}), as depicted in Figure \ref{Figtrapped1}. Indeed, using the fact that $\mathcal{T}'$ has finite space-time volume (because it is trapped), one can obtain the quantitative estimates responsible for the classification and the blow up of the transverse curvature components.
	
	To prove the existence of a trapped neighborhood of $\CH$, we first establish a breakdown criterion (Theorem \ref{blowuprough}). For this, we define the set of static points $S_0 \subset \CH$ as the set of $u_0 \in \CH$ such that the opposite of \eqref{staticrough} is true i.e.\ \begin{equation} \label{staticconditionintro}
		\int_{C_{u_0}} \frac{dr}{\frac{2\rho}{r}-1} =+\infty,
	\end{equation} where $C_{u_0}$ is a null cone transverse to $\CH$, $r$ is the area-radius and $\rho$ the Hawking mass. We call $\CH-S_0$ the set of Dafermos points, satisfying the Dafermos condition \eqref{staticrough}. In Theorem \ref{blowuprough}, our breakdown condition (triggering \textit{eventually} \color{black} the blow up of the Hawking mass $\rho$ to the future of  $C_{u_0}$\color{black}) is precisely the statement that $u_0$ is a Dafermos point i.e.\ $u_0 \notin S_0$. \begin{rmk}
		Note that on the Reissner--Nordstr\"{o}m Cauchy horizon, all points are static i.e.\ $S_0=\CH$. Nevertheless, in the dynamical case, it is conjectured that, \textit{generically}, every point in the Cauchy horizon is a Dafermos point i.e.\ $S_0=\emptyset$.
	\end{rmk}
	
	Now, we walk the reader through the steps of the paper, starting from the proof of Theorem \ref{blowuprough} to that of Theorem \ref{rough2}.
	
	\begin{enumerate}
		\item \textit{Static points occur only if the radiation is trivial on $\CH$ (section \ref{rigiditysection})} \label{stepstatic}
		
		Let $u_0 \in S_0$, a static point. Then, we prove that there is no radiation on $\CH \cap \{ u\leq u_0\}$ i.e.\ that $\CH \cap \{ u\leq u_0\}$ is isometric to a portion of a Reissner--Nordstr\"{o}m Cauchy horizon. Moreover, we prove that there exists a trapped neighborhood $\mathcal{T}_{u_0+\epsilon}$ of $\CH \cap \{ u\leq u_0+\epsilon\}$, for $\epsilon>0$ and that quantitative estimates hold on $\mathcal{T}_{u_0+\epsilon}$ (Theorem \ref{classificationtheorem} ).
		
		The proof relies on a bootstrap method to extend estimates from $ \{ u \leq u_s\}$, which we have by \cite{Moi} (see section \ref{LB} for a reminder) to a rectangle $[u_s,u_0] \times [v_0,+\infty)$, if $u_0>u_s$ (the easier case $u_0 \leq u_s$ is treated in Lemma \ref{<us}). For this, we define $\mathcal{B}_{v_0} \subset [u_s,u_0]$ as $u \in \mathcal{B}_{v_0}$ if certain (mild) quantitative estimates are valid on $[u_s,u_0] \times [v_0,+\infty)$. For some $v_0$,  $\mathcal{B}_{v_0} \neq \emptyset$ by the estimates in $ \{ u \leq u_s\}$ and we will show that $\mathcal{B}_{v_0}$ is open and closed in $[u_s,u_0]$, hence $\mathcal{B}_{v_0}=[u_s,u_0]$. \begin{enumerate}
			\item We prove that if $u \in \mathcal{B}_{v_0}$, then $\phi \equiv 0$ on $\CH \cap [-\infty,u]$, hence $\CH \cap [-\infty,u]$ is isometric to a portion of Reissner--Nordstr\"{o}m Cauchy horizon; moreover sharper estimates are satisfied on $[u_s,u] \times [v_0,+\infty)$ (Lemma \ref{propagation}).
			\item Then, using these estimates, we show that $[u_s,u] \times [v_0,+\infty)$ is trapped, hence $[u_s,u+\epsilon] \times [v_0,+\infty)$ is also trapped, by openess of the trapped region and the Raychaudhuri equation, for some small $\epsilon>0$ (Lemma \ref{trappedlemma}).
			\item Using that $[u_s,u+\epsilon] \times [v_0,+\infty)$ is trapped, hence has finite and small space-time volume, we prove estimates in this region, which are stronger than the ``original'' mild estimates satisfied on $\mathcal{B}_{v_0}$ (Lemma \ref{finitevolumeestlemma}).
			\item We ``retrieve the bootstrap'': thanks to the estimates, we prove that $\mathcal{B}_{v_0}$ is open and closed, hence $\mathcal{B}_{v_0}=[u_s,u_0]$.
		\end{enumerate}
		We also proved quantitative estimates on the trapped rectangle  $[u_s,u+\epsilon] \times [v_0,+\infty)$; Theorem \ref{classificationtheorem} is then proven.
		\item \textit{A first classification of the Cauchy horizon, by the structure of the static set $\mathcal{S}_0$ (section \ref{classificationsection1} and section \ref{classificationsection2})} \label{stepfirstclassification}
		
		From the Raychaudhuri equation, one can prove immediately that $\mathcal{S}_0$ is a past set: if \eqref{staticconditionintro} holds at $u_0$, then it holds for any $u \leq u_0$. Thus, we introduce the terminology of the classification, with three possible cases (Corollary \ref{threetypes}): \begin{enumerate}
			\item $S_0 =\emptyset$: we then say that $\CH$ is a Cauchy horizon of dynamical type.
			\item  $S_0 = \CH$: we then say that $\CH$ is a Cauchy horizon of static type.
			\item  $S_0 = (-\infty,u_T]$: we then say that $\CH$ is a Cauchy horizon of mixed type and $u_T$ is the transition time.
		\end{enumerate} In the next step, we will relate the dynamical and mixed Cauchy horizons to the blow up of the mass.
		
		\item \textit{``Local'' blow up of the Hawking mass, for dynamical and mixed types (section \ref{localsection})} \label{steplocalblowup}
		
		Using the quantitative estimates of Theorem \ref{classificationtheorem}, we prove that the Hawking mass $\rho$ blows up on \begin{enumerate}
			\item  $\CH \cap \{ u \leq u_s\}$ for some $u_s \in \RR$ if $\CH$ is of dynamical type (Lemma \ref{localmassblowupdynamical}), using estimates from \cite{Moi}.
			\item  $\CH \cap (u_T,u_T+\epsilon]$ for some $\epsilon>0$ if $\CH$ is of mixed type, where $u_T$ is the transition time (Lemma \ref{localmassblowupmixed}). 
		\end{enumerate} 
		
		\item \textit{Propagation of the Hawking mass blow up and proof of Theorem \ref{blowuprough} (section \ref{propagationsection} and section \ref{blowupcriterion})} \label{stepglobalblowup}
		
		We prove that, if the Hawking mass $\rho$ blows up at $u_0 \in \CH$ then $\rho$ blows up for all $u_0 \leq u \in \CH$ (Lemma \ref{massblowuplemma}).
		
		\begin{rmk}
			Lemma \ref{massblowuplemma} is in fact independent of the other results, and can be used alone, as for instance in \cite{r=0}.
		\end{rmk} Thus, using Step \ref{steplocalblowup}, we prove that the Hawking mass $\rho$ blows up on \begin{enumerate}
			\item  $\CH$, if $\CH$ is of dynamical type.
			\item $\CH \cap \{ u > u_T \}$ if $\CH$ is of mixed type, where $u_T$ is the transition time.
		\end{enumerate} 
		Then, invoking the (preliminary)  classification of Step \ref{stepfirstclassification}, we obtain a proof of Theorem \ref{blowuprough}: if $u_0$ is a Dafermos point i.e.\ $u_0 \notin S_0$, then either $\CH$ is of dynamical type, or $\CH$ is of mixed type and $u_0 >u_T$. In any case, the Hawking mass $\rho$ blows up on $\CH$ at any $u \geq u_0$.
		
		\item \textit{Trapped neighborhood of $\CH$ and proof of Theorem \ref{trappedrough} (section \ref{trappedprop})}
		
		Recall (see section \ref{geometricframework}) that $(u,v) \in \T$ if and only if $2\rho(u,v)>r(u,v)$. Hence, since $r$ is bounded inside the black hole, any null cone $C_{u}$ under $\CH$ is eventually trapped, providing the Hawking mass $\rho$ blows up at $u$. Thus, from Step \ref{stepstatic} and Step \ref{stepglobalblowup}, we construct a trapped neighborhood $\mathcal{T}'$ of $\CH$ as depicted in Figure \ref{Figtrapped1} in the following way: \begin{enumerate}
			\item Using the blow up of $\rho$ on the entire $\CH$, if $\CH$ is of dynamical type,
			\item Using the trapped neighborhood $\mathcal{T}_{u_0+\epsilon}$ of Theorem \ref{classificationtheorem} for all $u_0 \in \CH$, if $\CH$ is of static type,
			\item Using the blow up of $\rho$ on $\CH \cap \{ u >u_T \}$ and the trapped neighborhood $\mathcal{T}_{u_0+\epsilon}$ of Theorem \ref{classificationtheorem} for all $u_0 \in \CH \cap \{ u  \leq u_T \}$, if $\CH$ is of mixed type.
		\end{enumerate} Thus, Theorem \ref{trappedrough} is proved.

		\item \textit{Quantitative estimates and final classification of the Cauchy horizon, proof of Theorem \ref{classificationrough} (section \ref{quantitative})}

		At this stage, we already have quantitative estimates, in particular the blow up of $Ric(X,X)$, on a neighborhood of $S_0$, but nothing on $\CH \backslash S_0 \cap \{ u \geq u_s\}$. While this is sufficient to conclude in the static case, we need a new approach in the mixed and dynamical case, to propagate the local estimate of \cite{Moi}, valid only in the region $\{u \leq u_s \}$. 
		
		For this, we prove that $L^1-L^{\infty}$ estimates are true on any region of finite space-time  (Lemma \ref{finitevolumeestlemma2}). Since $\mathcal{T}'$, the trapped neighborhood of $\CH$, has a finite space-time volume, we propagate the desired estimates (Corollary \ref{trappedprop}).
		
		As a result, we obtain the blow up of $Ric(X,X)$ on the \underline{entire} $\CH$ for all three types (Proposition \ref{Ricciblowupprop}).

		Using also Step \ref{stepstatic}, we obtain that $r$ extends continuously to a function $r_{CH}$ on $\CH$ and that \begin{enumerate}
			\item $r_{CH}$ is strictly decreasing, if $\CH$ is of dynamical type,
			\item $r_{CH} \equiv r_->0$ is constant, if $\CH$ is of static type,
			\item $r_{CH}(u) =  r_->0$ for all $u \leq u_T$ and $r_{CH}$ is strictly decreasing on $\CH \cap \{ u >u_T\}$, if $\CH$ is of mixed type.
		\end{enumerate}
		This ends the classification of the Cauchy horizon into dynamical, static or mixed type and the proof of Theorem \ref{classificationrough}.
		
		\item \textit{$C^2$ inextendibility and proof of Theorem \ref{rough2}, Theorem \ref{rough1} and Theorem \ref{rough1cond} (section \ref{inextproofsection})}
		
		To conclude the proof of $C^2$ inextendibility, we work by contradiction, following closely the strategy of \cite{JonathanStab}.
		\begin{enumerate}
			\item If $(M,g)$ is $C^2$-future-extendible,  then no $C^2$ geodesic can cross either of the boundaries $\mathcal{S}$, $\mathcal{S}^1_{\Gamma}$, $\mathcal{S}^2_{\Gamma}$, $\mathcal{S}_{i^+}$ due to the blow up of the Kretschmann scalar when $r=0$ (an argument due to Kommemi \cite{Kommemi}, re-used in \cite{JonathanStab}).
			
			\item Thus, in the two-ended case, one can construct a $C^2$ radial null geodesic of tangent vector $X$ crossing $\CH$ into the extension (Proposition \ref{geomextendJonathan}, originally in \cite{JonathanStab}). This contradicts the blow up of $Ric(X,X)$ (Proposition \ref{Ricciblowupprop}).
			\item One-ended space-times are not necessarily inextendible, due to the presence of $\mathcal{CH}_{\Gamma}$; yet one can still prove that $\CH$ is $C^2$-future-inextendible using the blow up of $Ric(X,X)$ (Proposition \ref{geomextendoneened} and Proposition \ref{inextproponeended}).
			\item Still in the one-ended case, if this obstruction disappears i.e.\ if $\mathcal{CH}_{\Gamma}=\emptyset$, then one can adapt the earlier arguments to prove the $C^2$-future-inextendibility of the space-time (Proposition \ref{oneendedCHgammaprop} and Lemma \ref{lastlemma} originally from \cite{JonathanStab}).
		\end{enumerate}
		
		
	\end{enumerate}

	\subsection{Acknowledgments} 	I am very grateful to Jonathan Luk for suggesting this problem and for fruitful discussions. I would like to thank Sung-Jin Oh for this interest in this problem, and for insightful discussions. I am grateful to Mihalis Dafermos and Jonathan Luk for useful comments on the manuscript. Most of this work was completed while I was a Ph.D.\ student at the University of Cambridge, and a visiting graduate student at Stanford University, whose hospitality I gratefully acknowledge. Finally, I would like to thanks two anonymous reviewers for helpful suggestions which helped clarify some aspects of the paper.

	\section{Geometric framework} \label{geometricframework}

	In this section, we provide the geometric set-up and the definition of various quantities that will be use throughout the paper. We also present the equations and the coordinates that will be used in the proofs.

	\subsection{Spherically symmetric initial data set} \label{sphericallysymdatasection}
	
	To obtain a spherically symmetric space-time, we work with spherically symmetric initial data and this symmetry is then transmitted to the solution (c.f. \cite{Kommemi}). Such a strategy is standard, so we only briefly recall some key definitions. 
	
	\begin{defn}
		An initial data set of the Cauchy problem for \eqref{1}, \eqref{2}, \eqref{3}, \eqref{4}, \eqref{5} is given by a septuplet $(\Sigma^{(3)},h_{i j}, K_{i j}, \phi_0, \phi'_0, E_i, B_i)$, where $(\Sigma^{(3)},h_{i j})$ is a three-dimensional Riemannian manifold, $ K_{i j}$ is a symmetric tensor, $\phi_0$, $\phi'_0$ are complex-valued scalar functions on $\Sigma^{(3)}$, and $E_i$ (the electric field), $B_i$ (the magnetic field) are vectors on $\Sigma^{(3)}$.
	\end{defn}
	\begin{rmk}
		The Cauchy problem for \eqref{1}, \eqref{2}, \eqref{3}, \eqref{4}, \eqref{5} is locally well-posed, see \cite{Kommemi}, \cite{MGHD}: for regular data $(\Sigma^{(3)}, ...)$ satisfying the constraints, there exists a solution $(M,g,\phi,F)$ and a coordinate system $(t,x_1,x_2,x_3)$ such that $$\Sigma^{(3)}= \{t=0\},$$ $$h_{i j}= g_{ | \Sigma^{(3)}},$$ $$K_{i j}=  \mathcal{L}_{\partial_t}  g_{ | \Sigma^{(3)}},$$ $$ \phi_0 = \phi_{| \Sigma^{(3)}}, $$ $$ \phi'_0 = D_t \phi_{| \Sigma^{(3)}}, $$$$E_{i}=  F(\partial_t, \partial_i),$$ $$B_{i}=  *F(\partial_t, \partial_i).$$
	\end{rmk}
	In fact, there exists a unique globally hyperbolic solution $(M,g,\phi,F)$ which is also maximal (c.f.\ \cite{MGHD} for precise definitions). We call $(M,g,\phi,F)$ the maximal (globally hyperbolic) development of the initial data $(\Sigma^{(3)},h_{i j}, K_{i j}, \phi_0, \phi'_0, E_i, B_i)$.
	
	\begin{defn}[\cite{Kommemi}]
		We say that $(M,g,\phi,F)$, the maximal development of $(\Sigma^{(3)},...)$, is spherically symmetric if \begin{enumerate}\item The Lie group $SO(3)$ acts smoothly by isometry on $(\Sigma^{(3)},h_{i j})$.
			\item The action of $SO(3)$ leaves $K$, $\phi_0$, $\phi'_0$, $E$ and $B$ invariant.
			\item $ \Sigma^{(3)}/ SO(3)$ can be equipped with the structure of a one-dimensional Riemmanian manifold-with-boundary.
			
		\end{enumerate}
	\end{defn}
	
	As a result, we obtain a smooth $SO(3)$ action on $(M,g)$ by isometry, with spacelike orbits. Then, one defines the quotient $ M / SO(3) $ which we conformally embed into a bounded subset of $(\RR^{1+1}, m_{\mu, \nu})$. We denote the  embedding as $\mathcal{Q}$ (a bounded subset of $\RR^{1+1}$). We also denote $\mathcal{Q}^+$, the future domain of dependence in $\mathcal{Q}$ of the conformal image of $\Sigma^{(3)}/ SO(3)$ in $\RR^{1+1}$. We also define $\mathcal{B}^+$ as the boundary of $\mathcal{Q}^+$ induced by the ambient $\RR^{1+1}$ (i.e.\ the limit points of $\mathcal{Q}^+$ as a subset of $\RR^{1+1}$)  and subsequently the bounded domain-with-boundary $\overline{\mathcal{Q}^+}:=  \mathcal{Q}^+\cup \mathcal{B}^+ \subset \RR^{1+1}$  (see Proposition 2.1 and its discussion in \cite{Kommemi} for details).

	\color{black}
	\begin{defn}
		We call $\Pi: M \rightarrow \mathcal{Q}$ the natural projection taking a point to its group orbit. Note that for all $p \in \mathcal{Q}$,  $\Pi^{-1}(p)$ is isometric to a sphere. We then define the area-radius function $r$ on $ \mathcal{Q}$ by the formula $$r(p)= \sqrt{\frac{Area(\Pi^{-1}(p))}{4 \pi}}.$$
	\end{defn}
	
	The metric $g$ on $M$ is then given by  \begin{equation} \label{gQ}
		g= g_{\mathcal{Q}}+ r^{2}d \sigma_{\mathbb{S}^{2}},
	\end{equation} where $d \sigma_{\mathbb{S}^{2}}$ is the standard metric on $\mathbb{S}^{2}$ and $g_{\mathcal{Q}}$ is a Lorentzian metric on $\mathcal{Q}$.
	
	We will denote $\Gamma \subset \mathcal{Q}$, that we call the center of symmetry, the set of fixed points under the $SO(3)$ action on $M$ (which we identify with its image under $\Pi$). Notice that, by definition, $r_{|\Gamma}=0$.
	\subsection{Metric in double null coordinates} \label{doublenull}
	$g_{\mathcal{Q}}$ defined in \eqref{gQ} is a $1+1$ Lorentzian metric and, as such, is conformally flat: thus, there exists coordinates $(u,v)$, which we call null coordinates, on $(\mathcal{Q},g_{\mathcal{Q}})$ and a function $\Omega^2(u,v)$ such that $$ g_{\mathcal{Q}}:= - \frac{\Omega^{2}(u,v)}{2} (du  \otimes dv+dv  \otimes du)$$ 
	
	In view of this formalism, we consider (abusing notation) the area-radius $r$ as a function $r(u,v)$ on $\RR^2$. In fact, one can use the coordinate system $(u,v,\theta,\varphi)$ on $M$ where $(\theta,\varphi)$ are the standard coordinates on $\mathbb{S}^2$. Thus, \eqref{gQ} becomes $$ g = - \frac{\Omega^{2}(u,v)}{2} (du  \otimes dv+dv  \otimes du)+ r^2(u,v) (d\theta^2+\sin(\theta)^2 d \varphi ^2).$$
	
	\begin{rmk} \label{doublenullremark}
		The choice of null coordinates is not unique: one can renormalize $(u,v)$ into new null coordinates $(\tilde{u},\tilde{v})$ by the identities $d\tilde{u}= f(u) du$, $d\tilde{v}= g(v) dv$ for \textit{any} strictly positive function $f$ and $g$. Notice that, upon this change of coordinate, $\Omega^2$ is also changed by the formula $\tilde{\Omega}^2 d\tilde{u} d\tilde{v} = \Omega^2 du dv$ i.e.\ $\tilde{\Omega}^2(u,v) = f(u)^{-1}  \cdot g(v)^{-1} \cdot  \Omega^2(u,v)$.
	\end{rmk}
	
	\begin{rmk} \label{penrosediagramremark}
		To draw a Penrose diagram (see Figures \ref{Fig1}, \ref{Fig5}, \ref{Fig2}, \ref{Figtrapped1}, \ref{Fig3}, \ref{Figconj}), we choose $(u,v)$ to be in $\mathcal{P}$, a bounded subset of $\RR^2$ and we draw $\mathcal{P}$. Since $\mathcal{P}$ and $\mathcal{Q}$ are conformally isometric, it is a standard fact that they have the same causal structure.
	\end{rmk}
	
	Now, we will use this coordinate system $(u,v)$ to define important quantities. We abuse notation denote  $F$ the push-forward by $\Pi$ of the original 2-form over $M$, and the same for $\phi$. The spherically symmetric character (c.f.\ \cite{Kommemi}) of $F$ imposes that there exists a scalar function $Q(u,v)$ (independent of the coordinate choice), called the charge, such that  $$ F= \frac{Q}{2r^{2}}  \Omega^{2} du \wedge dv.$$
	\begin{rmk}
		In our setting, the scalar field is charged, hence \eqref{4} is an \textit{inhomogeneous} Maxwell equation, and $Q(u,v)$ is a scalar function. This is in contrast with the uncharged case where $Q \equiv e \in \RR$, c.f. remarks \ref{staticMaxwell} and \ref{staticMaxwell2}.
	\end{rmk}
	
	Subsequently, we define the Lorentzian gradient of $r$, and introduce the mass ratio $\mu$ by the formula $$ 1-\mu:=g_{\mathcal{Q}}(\nabla r,\nabla r)= \frac{-4 \partial_u r \cdot \partial_v r}{\Omega^2}.$$ $\mu$ is \textit{independent} of the coordinate choice. We define the Hawking	
	mass $\rho$ (also independent of the coordinates choice): 
	
	$$ \rho := \frac{\mu \cdot r}{2} =\frac{r}{2} \cdot(1- g_{\mathcal{Q}} (\nabla r, \nabla r )),$$ and the modified mass $\varpi$, the last quantity which is independent of the coordinates choice, also involving the charge $Q$:
	\begin{equation} \label{electromass}
		\varpi := \rho + \frac{Q^2}{2r}= \frac{\mu r}{2} + \frac{Q^2}{2r} .
	\end{equation}	
	
	Now, we introduce notations for coordinate-dependent quantities: the ingoing derivative of $r$ in $(u,v)$ coordinates: \begin{equation} \label{nudef}
		\nu(u,v):= \partial_u r (u,v),
	\end{equation} followed by the outgoing derivative of $r$ in $(u,v)$ coordinates : \begin{equation} \label{lambdadef}
		\lambda(u,v):= \partial_v r (u,v).
	\end{equation}
	Then, we define $\kappa \in \mathbb{R} \cup \{ \pm \infty\} $ by the following formula, also using the previous notations:
	\begin{equation} \label{kappadef}
		\kappa = \frac{\lambda}{1-\frac{2\rho}{r}}\color{black}=\frac{-\Omega^2}{4 \nu },
	\end{equation} and its ``outgoing'' analogue  $\iota \in \mathbb{R} \cup \{ \pm \infty\} $ \begin{equation} \label{iotadef}
		\iota = \frac{\nu}{1-\frac{2\rho}{r}}\color{black}=\frac{-\Omega^2}{4 \lambda }.
	\end{equation}

	We summarize all the relations between the different quantities:	\begin{equation} \label{murelation}
		1-\mu =1-\frac{2\rho}{r} = 1-\frac{2\varpi}{r}+\frac{Q^2}{r^2}=\frac{-4 \nu \lambda}{\Omega^2}= \kappa^{-1} \cdot \lambda = \iota^{-1} \cdot \nu.
	\end{equation}
	
	

	\subsection{The Reissner--Nordstr\"{o}m solution} \label{RNsolution}
	
	The sub-extremal Reissner--Nordstr\"{o}m space-time of mass $M>0$ and charge $ e \in (0,M)$ is a two-ended (see section \ref{topology}) spherically symmetric black hole solving the system \eqref{1}, \eqref{2}, \eqref{3}, \eqref{4}, \eqref{5} for $\phi \equiv 0$, whose metric is given by $$ g_{RN}= -(1-\frac{2M}{r} +\frac{e^2}{r^2})^{-1}  dt^2 + (1-\frac{2M}{r} +\frac{e^2}{r^2})  dr^2 + r^2 (d\theta^2+\sin(\theta)^2 d \varphi ^2),$$ and in the coordinate system $(t,r,\theta,\varphi)$, $\partial_t$ is a time-like Killing vector field. We define null coordinates $u= \frac{r^*-t}{2}$, $v=\frac{r^*+t}{2}$, where $r^*$ is defined by the $\frac{dr^*}{dr}= -(1-\frac{2M}{r} +\frac{e^2}{r^2})^{-1}$. Thus, the metric can be re-written in $(u,v,\theta_,\varphi)$ coordinates as
	$$ g_{RN}= -2(1-\frac{2M}{r} +\frac{e^2}{r^2})  (du  \otimes dv+dv  \otimes du)+ r^2(u,v) (d\theta^2+\sin(\theta)^2 d \varphi ^2).$$
	The polynomial $r^2-2Mr+e^2=r^2(1-\frac{2M}{r} +\frac{e^2}{r^2})$ admits two distinct roots $r_{\pm}(M,e)= M \pm \sqrt{M^2-e^2}$, as $e \in (0,M)$. Note that $r^*(r_{\pm})= \mp \infty$. The larger root corresponds to a (bifurcate) event horizon $\mathcal{H}^+:= \{ r=r_+\} = \{u+v = -\infty\}$, while the smaller root corresponds to a (bifurcate) Cauchy horizon  $\CH:= \{ r=r_-\} = \{u+v = +\infty\}$. \\ We define $K_{+}(M,e)= \frac{M}{r^2_+}- \frac{e^2}{r_+^3}>0$ the surface gravity of $\mathcal{H}^+$ and $K_{-}(M,e)= \frac{M}{r^2_-}- \frac{e^2}{r_-^3}<0$ the surface gravity of $\CH$.
	
	Defining $\Omega^2_{RN} = -4 (1-\frac{2M}{r} +\frac{e^2}{r^2})$, a standard computation in the black hole interior shows that $$   \Omega^2_{RN}(u,v)=\Omega^2_{RN}(r^*) = C_+ \cdot e^{2K_+ \cdot r^*}+o( e^{2K_+ \cdot r^*})= C_+ \cdot e^{2K_+ \cdot (u+v)}+o(  e^{2K_+ \cdot (u+v)}),$$ as $r^* \rightarrow -\infty$ i.e.\ towards the event horizon $\mathcal{H}^+$, for some \textit{explicit} constant $C_+(M,e)>0$\color{black}. Similarly, we have the following asymptotics as $r^* \rightarrow +\infty$ i.e.\ towards the Cauchy horizon $\CH$, for another \textit{explicit} constant $C_-(M,e)>0$: $$   \Omega^2_{RN}(u,v)=\Omega^2_{RN}(r^*) = C_- \cdot e^{2K_- \cdot r^*}+o(e^{2K_- \cdot r^*})= C_- \cdot e^{2K_- \cdot (u+v)}+o( e^{2K_- \cdot (u+v)}).$$
	
	The Penrose diagram of the Reissner--Nordstr\"{o}m black hole is a particular case of the diagram of Figure \ref{Fig2}, where $\mathcal{S}=\mathcal{S}_{i^+}= \emptyset$, i.e.\ the two Cauchy horizons $\CH$ emanating from each end $i^+$ meet at a bifurcation sphere $(u=+\infty,v=+\infty)$.
	Now, we express the quantities defined in section \ref{doublenull} for the Reissner--Nordstr\"{o}m metric. We start with the ingoing and outgoing derivatives of $r$, defined in \eqref{nudef} and \eqref{lambdadef}:
	$$ \nu = \lambda = \frac{dr}{dr^*}= 1-\frac{2M}{r} +\frac{e^2}{r^2} = -\frac{\Omega^2_{RN}}{4}.$$
	Consequently, the quantities $\kappa$ and $\iota$ defined in \eqref{kappadef}, \eqref{iotadef} obey the following relation $$\kappa = \iota = 1.$$
	Moreover, the charge $Q$ and the renormalized mass $\varpi$ are constant (but not the Hawking mass $\rho$): $$ Q \equiv e, \hskip 5 mm \varpi \equiv M,$$ $$ \rho = M -\frac{e^2}{2r}.$$

	\subsection{Equations in double null coordinates in spherical symmetry} 	Now, we formulate the equations \eqref{1}, \eqref{2}, \eqref{3}, \eqref{4}, \eqref{5} in any null coordinate system $(u,v)$ as introduced in section \ref{doublenull} (see also \cite{MihalisPHD}, \cite{JonathanStab}, \cite{Kommemi})\color{black}. We start by the wave equation for $r$, recalling the definitions of $\nu$ and $\lambda$ from \eqref{nudef}, \eqref{lambdadef}:
	\begin{equation}\label{Radius}\partial_{u}\partial_{v}r = \partial_u \lambda = \partial_v \nu =\frac{- \Omega^{2}}{4r}-\frac{\nu \lambda}{r}
		+\frac{ \Omega^{2}}{4r^{3}} Q^2 +  \frac{m^{2}r }{4} \Omega^2 |\phi|^{2} , \end{equation}
	and the following reformulation of \eqref{Radius} will be  useful :		
	\begin{equation} \label{Radius3}
		-\partial_u  \partial_v (\frac{r^2}{2})	=\partial_u (-r  \lambda) =	\partial_v (-r  \nu ) = \frac{\Omega^2}{4}\cdot (1- \frac{Q^2}{r^2}-m^2 r^2 |\phi|^2).
	\end{equation}
	Now we turn to the wave equation for $\log(\Omega^2)$ :\begin{equation}\label{Omega}
		\partial_{u}\partial_{v} \log(\Omega^2)=-2\Re(D_{u} \phi \overline{D_{v}\phi})+\frac{ \Omega^{2}}{2r^{2}}+\frac{2 \nu \lambda}{r^{2}}- \frac{\Omega^{2}}{r^{4}} Q^2,
	\end{equation}
	which can also be written, combining with \eqref{Radius}:
	\begin{equation}\label{Omega3}
		\partial_{u}\partial_{v} \log(r\Omega^2)=  \frac{ \Omega^{2}}{4r^{2}} \cdot \left(1-   \frac{3Q^2}{r^{2}} +m^{2}r^2 |\phi|^{2} - 8 r^2\Re( \frac{D_{u}\phi}{\Omega^2} \cdot  D_{v}\bar{\phi})  \right).
	\end{equation}
	Then we formulate, the ingoing Raychaudhuri equations, recalling the definition of $\kappa$ from \eqref{kappadef}: 
	\begin{equation}\label{RaychU}\partial_u( \kappa^{-1})=\frac {4r}{\Omega^{2}}|  D_{u} \color{black}\phi|^{2}, \end{equation}
	and the outgoing Raychaudhuri equation, recalling the definition of $\iota$ from \eqref{iotadef}:
	\begin{equation} \label{RaychV}\partial_v( \iota^{-1})=\frac {4r}{\Omega^{2}}|  D_{v} \color{black}\phi|^{2},\end{equation}
	
	Now we present the propagation equation for a massive and charged scalar field  (Klein-Gordon wave equation):
	\begin{equation}\label{Field}
		D_{u} D_{v} \phi =-\frac{ \lambda \cdot D_{u}\phi}{r}-\frac{ \nu \cdot  D_{v}\phi}{r} +\frac{ iq_{0} Q \Omega^{2}}{4r^{2}} \phi
		-\frac{ m^{2}\Omega^{2}}{4}\phi,\end{equation}which can also be written in different ways, noticing that $[D_u,D_v]=\frac{ iq_{0} Q \Omega^{2}}{2r^{2}}$:	\begin{equation}\label{Field2}
		D_{u}(rD_{v} \phi) =- \lambda \cdot D_u\phi +\frac{ \Omega^{2} \cdot \phi}{4r} \cdot ( i q_{0} Q-m^2 r^2)
		, \end{equation}\begin{equation}\label{Field3}
		D_{v}(rD_{u} \phi) = - \nu  \cdot D_{v}\phi 
		- \frac{ \Omega^{2} \cdot \phi}{4r} \cdot ( i q_{0} Q+m^2 r^2)
		. \end{equation}
	
	From the Maxwell equation \eqref{4}, we obtain two null transport equations:
	\begin{equation} \label{chargeUEinstein}
		\partial_u Q = -q_0 r^2 \Im( \phi \overline{ D_u \phi}),
	\end{equation}	\begin{equation} \label{ChargeVEinstein}
		\partial_v Q = q_0 r^2 \Im( \phi \overline{D_v \phi}) .
	\end{equation}

	Now we can reformulate the former equations to put them in a form that is more convenient to use. 
	
	It is interesting to use \eqref{Radius}, \eqref{RaychU}, \eqref{RaychV}, \eqref{chargeUEinstein}, \eqref{ChargeVEinstein} to derive an equation for the Hawking mass :

	\begin{equation} \label{massUEinstein}
		\partial_u \rho = \frac{r^2}{2\iota}| D_u \phi |^2 +  \left(\frac{m^2}{2} r^2 |\phi|^2+ \frac{Q^2}{2r^2} \right) \cdot \nu  ,
	\end{equation}
	\begin{equation} \label{massVEinstein} 
		\partial_v \rho = \frac{r^2}{2\kappa}| D_v \phi |^2+ \left(\frac{m^2}{2} r^2 |\phi|^2+ \frac{Q^2}{2r^2} \right) \cdot \lambda.
	\end{equation}

	

	\subsection{Electromagnetic gauge choice and gauge invariant estimates}
	
	Even after we fix a double null coordinate system $(u,v)$, an electromagnetic gauge freedom subsist. 
	
	Indeed, since $F=dA$, $F$ is unchanged by the transformation $ A \rightarrow  A+ d f$. In fact, a solution of the system \eqref{1}, \eqref{2}, \eqref{3}, \eqref{4}, \eqref{5} gives rise to another solution under the following gauge transform  $$ \phi \rightarrow  e^{-i q_0 f } \phi ,$$ $$ A \rightarrow  A+ d f.$$ This is because $D_{\mu} \phi$ is transformed according to the formula (coming from an elementary computation): $$ D_{\mu} \phi \rightarrow e^{-i q_0 f }  D_{\mu}\phi,$$ hence $|D_{\mu} \phi|$ and $|\phi|$ are gauge invariant (but not $\phi$, nor $D_{\mu} \phi$!). 
	
	In view of this fact, one can easily derive the following gauge invariant estimates (see \cite{Moi2}):  for all $u_1<u_2$, $v_1<v_2$: $$ |\phi|(u_2,v) \leq |\phi|(u_1,v)+ \int_{u_1}^{u_2} |D_u \phi|(u',v) du',$$
	or its analogue with $v$ replacing $u$. For simplicity, we will work in this paper in the vicinity of $\CH$ and under the gauge \begin{equation} \label{electrogauge}
		A_v \equiv 0,
	\end{equation}\begin{equation} \label{electrogauge2}
		A_u(u_0,\cdot) = 0,
	\end{equation} for some $u_0 \in \RR$. 
	Although the gauge choice is irrelevant for gauge invariant estimates, $ D_v \phi = \partial_v \phi$ in practice by \eqref{electrogauge}. \color{black}

	
	\subsection{Trapped region and apparent horizon}

	We define the trapped region $\mathcal{T}$, the regular region $\R$ and the apparent horizon $\A$ using \eqref{murelation}, as
	\begin{enumerate}
		\item  \label{characttrapped} $(u,v) \in \T$   if and only if   $1-\frac{2\rho(u,v)}{r(u,v)}<0$ if and only if $\lambda(u,v)<0$ and  $\nu(u,v) <0$,
		\item $(u,v) \in \R$  if and only if   $1-\frac{2\rho(u,v)}{r(u,v)}>0$ if and only if either $\lambda(u,v)>0$ and $\nu(u,v)<0$ or $\lambda(u,v)<0$ and $\nu(u,v)>0$.
		\item $(u,v) \in \A$  if and only if   $1-\frac{2\rho(u,v)}{r(u,v)}=0$ if and only if $\lambda(u,v) \cdot \nu(u,v)=0$.
	\end{enumerate}	
	
	Note that, for admissible space-times (in the sense of Definition \ref{admissibilitydef}), the case  $\lambda(u,v)>0$ and $\nu(u,v)>0$ i.e. that $(u,v)$ is an anti-trapped surface \textit{never} occurs.
	
	Notice also, that, for one-ended admissible space-times, $\nu <0$ everywhere so the occurrence of a trapped surface in $(u,v)$ depends only on the sign of $\lambda(u,v)$. Moreover, for two ended admissible space-times, there exists a neighborhood of $\mathcal{CH}_{i^+_1}$ on which $\nu<0$. Therefore, the following characterisation is pertinent: over $\mathcal{D}:=\{ (u,v)/ \nu(u,v)<0\}$, we have: \begin{enumerate}
		\item  $(u,v) \in \T$  if and only if $\lambda(u,v)<0$.
		\item $(u,v) \in \R$ if and only if $\lambda(u,v)>0$.
		\item $(u,v) \in \A$  if and only if   if and only if $\lambda(u,v)=0$.
	\end{enumerate}	
	
	It is in fact this final characterization that we will use. 
	
	Note also that for one-ended admissible space-time, as $r_{|\Gamma}=0$, we have $\Gamma \subset \R$.

	\subsection{Topology of the initial data: the one-ended case and the two-ended case} \label{topology}
	
	In this section, we define mathematically the notion of one-ended or two-ended space-times, following \cite{Kommemi}.
	
	\begin{defn} \label{oneendedef}
		We say that $(M,g_{\mu \nu}, \phi,F_{\mu \nu})$ is the maximal development of spherically symmetric one-ended \color{black}      initial data if $(M,g_{\mu \nu}, \phi,F_{\mu \nu})$ is the future maximal globally hyperbolic development of $(\Sigma^{(3)},h_{i j}, K_{i j}, \phi_0, \phi'_0, Q_0)$ and $\Sigma^{(3)}$ is diffeomorphic to $\RR^3$.
	\end{defn}

	\begin{defn} \label{twoendedef}
		We say that $(M,g_{\mu \nu}, \phi,F_{\mu \nu})$ is the maximal development of spherically symmetric two-ended initial data if $(M,g_{\mu \nu}, \phi,F_{\mu \nu})$ is the future maximal globally hyperbolic development of $(\Sigma^{(3)},h_{i j}, K_{i j}, \phi_0, \phi'_0, Q_0)$ and $\Sigma^{(3)}$ is diffeomorphic to $ \mathbb{S}^2 \times \RR$.
	\end{defn}
	\begin{rmk}
		The results of Theorem \ref{oneendedapriori} and Theorem \ref{twoendedapriori} prove that in the two-ended case $\Gamma = \emptyset$, while in the one-ended case $\Gamma$ is connected, time-like, and non-empty.
	\end{rmk}
	\begin{defn} \label{admissibilitydef}
		We say that  $(\Sigma^{(3)},h_{i j}, K_{i j}, \phi_0, \phi'_0, Q_0)$ is an admissible data set  \begin{enumerate}
			\item in the one-ended case, if there exists no anti-trapped surface on $\Sigma^{(3)}$ in the sense that $\nu_{|\Sigma_0} <0$.
			
			\item in the two-ended case, if there exists $u_1 < u_2$ such that $\nu_{|\Sigma_0}(u) <0$ for all $u \leq u_2$ 
			and $\lambda_{|\Sigma_0}(u) <0$ for all $u \geq u_1$.	
		\end{enumerate}
	\end{defn}
	\begin{rmk}
		Note that the particularity of two-ended admissible data sets is that they already contain a trapped surface, hence their maximal development feature a black hole. Thus the two-ended case does not allow for trapped surfaces (hence black holes) to form dynamically, in contrast with the one-ended case, suitable to study gravitational collapse.
	\end{rmk}
	
	\subsection{Notions of inextendibility} \label{inextsection}
	
	In this section, we define two notions of $C^2$-inextendibility: the first one is standard and can be found in \cite{JonathanStab}. The second one is the $C^2$-inextendibility across $\CH$, a new (but analogous) notion which we use in the one-ended case, as there exists an additional obstruction to $C^2$-inextendibility in this case, related to Weak Cosmic Censorship (see section \ref{connected}).
	
	\begin{defn} \label{generalinext} We consider $(M,g_{\mu \nu}, \phi,F_{\mu \nu})$ the maximal development of smooth, spherically symmetric and admissible (in the sense of Definition \ref{admissibilitydef}) one-ended or two-ended initial data satisfying the Einstein--Maxwell--Klein--Gordon system. 
		Then we say that $(M,g)$ is $C^2$-future-extendible if there exists a differentiable manifold $\tilde{M}$ equipped with a $C^2$ Lorentzian metric $\tilde{g}$ and a smooth isometric embedding $i: M \rightarrow  \tilde{M}$, such that $i(M)$ is a proper subset of $\tilde{M}$ and moreover the following condition holds true: \begin{enumerate}
			\item For every $p \in \tilde{M}-i(M)$, $I^+(p) \cap i(M) = \emptyset$.
			
		\end{enumerate}

		If no such extension exists, we say that $(M,g)$ is $C^2$-future-inextendible.
	\end{defn}
	
	\begin{defn} \label{CHinext} We consider $(M,g_{\mu \nu}, \phi,F)$ the maximal development of smooth, spherically symmetric and admissible (in the sense of Definition \ref{admissibilitydef}) one-ended initial data satisfying the Einstein--Maxwell--Klein--Gordon system. Following Theorem \ref{oneendedapriori}, we obtain the a priori boundary decomposition of $\mathcal{Q_+}$ induced by the ambient $\RR^{1+1}$ as $$\mathcal{B_+}=b_{\Gamma}\cup \mathcal{S}_{\Gamma}^1 \cup  \mathcal{S}_{\Gamma}^2 \cup \mathcal{S}  \cup \mathcal{S}_{i^+}  \cup \CH \cup i^+ \cup \mathcal{I}^+ \cup i^0 .$$
		
		Then we say that $(\barM,g)$ is $C^2$-future-extendible across $\CH$ if there exists a differentiable manifold $\tilde{M}$ equipped with a $C^2$ Lorentzian metric $\tilde{g}$ and a smooth isometric embedding $i: M \rightarrow  \tilde{M}$, such that $i(M)$ is a proper subset of $\tilde{M}$ and moreover the following conditions hold true: \begin{enumerate}
			\item For every $p \in \tilde{M}-i(M)$, $I^+(p) \cap i(M) = \emptyset$.
			
			\item There exists $p\in \partial M$, where $ \partial M$ is the topological boundary of $i(M)$ in $\tilde{M}$, at which $\partial M$ is differentiable and a continuous curve $\gamma:(-\epsilon,\epsilon) \rightarrow \tilde{M}$ with $\gamma(0)=p$, $\gamma(-\epsilon,0) \subset M$ and moreover $ \overline{\Pi(\gamma(-\epsilon,0))}^{\mathcal{Q}^+} \cap \CH \neq \emptyset.$
		\end{enumerate}

		If no such extension exists, we say that $(M,g)$ is $C^2$-future-inextendible across $\CH$.
	\end{defn}

	In practice, we always identify $M$ with $i(M)$ to lighten the notation when possible, as these two manifolds are isometric. \color{black}
	\section{Precise statement of the main results} \label{results}
	
	In this section, we describe the results of section \ref{firstversion} in a more precise way. We will work with the conventions, notations and definitions of Theorem \ref{oneendedapriori}, Theorem \ref{twoendedapriori} and section \ref{geometricframework}. When there is no risk of confusion, we will use the notation $\CH$ to denote any Cauchy horizon emanating from time-like infinity, in both the one-ended or the two-ended case.
	
	\subsection{Recalling the setting and the previous results of \cite{Moi}, near time-like infinity} 
	
	In this section, we rephrase the results of \cite{Moi} in a way which is convenient to use in our setting. The assumptions of the following theorem are also the basic assumptions we will rely on during the entire paper.
	
	The result of \cite{Moi} is \textit{local} in a neighborhood of $i^+$ and thus, can be applied indifferently to the one-ended or two-ended case. For convenience, we rephrase the statement of the theorem of \cite{Moi} into a one-ended case, and a two-ended case.

	We introduce the notation $A \lesssim B$ if there exists a constant $C(M, e, m^2, q_0)>0$, and $A \sim B$ if  $A \lesssim B$ and  $B \lesssim A$. Sometimes in the course of the paper, we will abuse notation and allow also for $C$ to depend on some fixed $u_0$ (in which case it will be specified, see for instance the proof of Lemma \ref{propagation}).\color{black}
	\begin{thm}[Non-linear stability and instability of the Cauchy horizon, \cite{Moi}] \label{previous}

		We consider the maximal development of smooth, spherically symmetric and admissible (in the sense of Definition \ref{admissibilitydef}) one-ended or two-ended initial data  $(M=\mathcal{Q}^+ \times_r \mathcal{S}^2,g_{\mu \nu}, \phi,F_{\mu \nu})$ satisfying the Einstein--Maxwell--Klein--Gordon system. 
		\begin{enumerate}
			\item One-ended case

			
			\begin{hyp} \label{blackholehyp}
				Assume that $\mathcal{Q}^+-J^{-}(\mathcal{I}^+) \neq \emptyset$.
			\end{hyp}
			
			Then one can define the event horizon $\mathcal{H}^+$ and the black hole region $\mathcal{BH}=\mathcal{Q}^+-J^{-}(\mathcal{I}^+) \subset \mathcal{Q}^+$.
			
			\begin{hyp} \label{affinecomplete}
				Assume that $\mathcal{H}^+$ is null future affine complete.
			\end{hyp}
			
			We introduce a double null coordinate system $(U,v)$ system on $\mathcal{BH}$, in which the metric takes the form $$ g =  - \frac{\Omega^{2}_H}{2} (dU  \otimes dv+dv  \otimes dU)+r^{2}[ d\theta^{2}+\sin(\theta)^{2}d \psi^{2}].$$
			
			In $\mathcal{Q}^+$, we define an ingoing null hypersurface $\mathcal{C}_{in}= \{ v=v_0, \hskip 1 mm 0 \leq U \leq U_0\}$ and an outgoing null hypersurface $ \mathcal{C}_{out}=\mathcal{H}^+ \cap  \{v \geq v_0\}$, where in these coordinates, $\mathcal{H}^+= \{ U=0\}$.
			
			We choose $(U,v)$ to be a regular coordinate  system across $\mathcal{H}^+$ and determined by the following conditions:	\begin{equation} \label{gauge2}
				\kappa_{|\mathcal{H}^+} = (\frac{-\Omega^2_{H}(0,v)}{4\partial_{U}r(0, v)})_{|\mathcal{H}^+}\equiv 1,
			\end{equation}    \begin{equation} \label{gauge1}
				(\partial_{U}r)_{|\mathcal{C}_{in}}(U, v_0) \equiv -1.
			\end{equation}
			
			
			Moreover, we will make the following geometric assumption: 
			
			\begin{hyp} \label{subexthyp}\label{hypchargeevent} We require $\mathcal{H}^+$ to be a sub-extremal  Reissner--Nordstr\"{o}m event horizon in the limit, with non zero asymptotic charge, namely 		
				
				$$ 0<  \limsup_{v \rightarrow +\infty} |Q|_{|\mathcal{H}^+}(v)  \neq  \limsup_{v \rightarrow +\infty} r_{|\mathcal{H}^+}(v) <+\infty.$$

			\end{hyp}

			\begin{hyp} \label{fieldevent}	 We make the following decay assumption for $\phi$ on $\mathcal{H}^+$: there exists $ C>0$ and $s>\frac{3}{4}$ \color{black} such that \footnote{In fact, the original theorem of \cite{Moi} only requires $s>\frac{1}{2}$, but we assume $s>\frac{3}{4}$ to simplify the proofs in the present paper. \color{black}}:
				$$|\phi|_{|\mathcal{H}^+}(v) +| D_{v} \phi|_{|\mathcal{H}^+}(v) \leq  C \cdot v^{-s},$$ 
			\end{hyp}
			
			\begin{hyp} \label{fieldUevent}  We also assume the following red-shift estimate on $C_{in}$: \label{fieldv0}
				$$| D_{U} \phi| (U, v_0) \leq C. $$

			\end{hyp}

			Then $\CH \neq \emptyset $ and we have, in $\mathcal{Q}^+$, $ \{ 0 < U \leq U_s, v=+\infty \} \subset \CH$ for some $U_s>0$\color{black}. Moreover stability estimates \footnote{Additionally, if if $s>1$, we prove in \cite{Moi} that $(M,g,\phi,F)$ admits a continuous extension to $\CH$.} hold in a region $\{ U \leq U_s, v\geq v_0\}$, in particular all the estimates of Proposition \ref{LBprop} (except maybe \eqref{phiVLBlowerbound}).

			If we also make a lower bound assumption on $\phi$, we obtain also an instability result on the Ricci curvature: \begin{hyp} 
				Assume that for some $2s-1 \leq p <  \min \{2s, 4s-2\}$, and some $D>0$, the following lower bound holds \begin{equation} \label{lowerhyp}
					\frac{1}{2}	\int_{v}^{+\infty} |D_v \phi|^2_{|\mathcal{H}^+}(0,v')dv'
					+   q_0 e  [\frac{1}{r_-}-\frac{1}{r_+}]  \int_{v}^{+\infty} \Im(\bar{\phi}_{|\mathcal{H}^+}(0,v') D_v \phi_{|\mathcal{H}^+}(0,v')) dv' \geq D \cdot v^{-p}
				\end{equation} 
				where we denote    
				\[
				M=\lim_{v\to+\infty}\varpi_{|\mathcal H^+}(v),\qquad
				e=\lim_{v\to+\infty}Q_{|\mathcal H^+}(v),\qquad
				r_\pm=M\pm\sqrt{M^2-e^2}.
				\]
				These limits exist under the preceding assumptions.
			\end{hyp}
			
			Then, under this additional assumption, \eqref{phiVLBlowerbound} holds, and moreover, defining the outgoing radial geodesic vector field $V= \Omega^{-2} \partial_v$, we have, for every $0<U \leq U_s $: $$ \limsup_{v \rightarrow+\infty} Ric(V,V)(U,v)=+\infty. $$

			\item  Two-ended case
			
			
			\begin{hyp}
				Assume that $\mathcal{H}^+_1$ and $\mathcal{H}^+_2$ are future null affine complete.
			\end{hyp}
			
			We introduce a double null coordinate system $(U_1,v_1)$ system on $\mathcal{BH}$, in which the metric takes the form $$ g =  - \frac{\Omega^{2}_1}{2} (dU_1  \otimes dv_1+dv_1  \otimes dU_1)+r^{2}[ d\theta^{2}+\sin(\theta)^{2}d \psi^{2}],$$ and a double null coordinate system $(u_2,V_2)$ in which the metric takes the form $$ g =  - \frac{\Omega^{2}_2}{2} (du_2  \otimes dV_2+dV_2  \otimes du_2)+r^{2}[ d\theta^{2}+\sin(\theta)^{2}d \psi^{2}].$$ 
			
			In $\mathcal{Q}^+$, we define two non-intersecting ingoing null hypersurfaces $\mathcal{C}_{in}^1= \{ v_1=v_{0}, \hskip 1 mm 0 \leq U_1 \leq U_0\}$ and $\mathcal{C}_{in}^2= \{ u_2=u_0, \hskip 1 mm 0 \leq V_2 \leq V_0\}$ and two outgoing null hypersurfaces $ \mathcal{C}_{out}^1=\mathcal{H}^+_1 \cap  \{v_1 \geq v_0\}$ and $ \mathcal{C}_{out}^2=\mathcal{H}^+_2 \cap  \{u_2 \geq u_0\}$, where in these coordinates, $\mathcal{H}^+_1= \{ U_1=0\}$, $\mathcal{H}^+_2= \{ V_2=0\}$.
			
			We choose $(U_1,v_1)$ to be a regular coordinate  system across $\mathcal{H}_1^+$ and determined by the following conditions: 	\begin{equation} 
				\kappa_{|\mathcal{H}^+_1} = (\frac{-\Omega^2_{1}(0,v_1)}{4\partial_{U_1}r(0, v_1)})_{|\mathcal{H}^+_1}\equiv 1,
			\end{equation}    \begin{equation}
				(\partial_{U_1}r)_{|\mathcal{C}_{in}^1}(U_1, v_0) \equiv -1.
			\end{equation}   	We choose $(u_2,V_2)$ to be a regular coordinate  system across $\mathcal{H}_1^+$ and determined by the following conditions: 	\begin{equation} 
				\iota_{|\mathcal{H}^+_2} = (\frac{-\Omega^2_{2}(0,v)}{4\partial_{V}r(0, v)})_{|\mathcal{H}^+_2}\equiv 1,
			\end{equation}    \begin{equation}
				(\partial_{V_2}r)_{|\mathcal{C}_{in}^2}(u_0, V_2) \equiv -1.
			\end{equation}

			\begin{hyp}	We require $\mathcal{H}^+_1$ and $\mathcal{H}^+_2$ to be a sub-extremal event horizons in the limit, with non zero asymptotic charge, in the sense of Assumption \ref{subexthyp}. Denote \begin{equation}\begin{split}
						&   \limsup_{v_1\rightarrow +\infty} Q_{|\mathcal{H}^+_1}(v_1)  = e_1,\quad   \limsup_{v_1\rightarrow +\infty} \varpi_{|\mathcal{H}^+_1}(v_1)  = M_1,\quad r_{\pm,1}= M_1 \pm \sqrt{M_1^2-e_1^2},\\ & \limsup_{u_2\rightarrow +\infty} Q_{|\mathcal{H}^+_2}(u_2)  = e_2,\quad   \limsup_{u_2\rightarrow +\infty} \varpi_{|\mathcal{H}^+_2}(u_2)  = M_2, \quad r_{\pm,2}= M_2 \pm \sqrt{M_2^2-e_2^2}.\end{split}
				\end{equation} 	
			\end{hyp}

			\begin{hyp} 	 We also make the following decay assumptions for $\phi$ on $C_{out}^1$ and $C_{out}^2$:  there exists $ C_1>0$, and $s_1>\frac{3}{4}$ such that for all $v_1\geq v_0$
				$$|\phi|_{|\mathcal{H}^+_1}(v_1) +| D_{v_1} \phi|_{|\mathcal{H}^+_1}(v_1) \leq  C_1 \cdot v_1^{-s_1},$$ and there exists $ C_2>0$, and $s_2>\frac{3}{4}$ such that for all $u_2\geq u_0$,
				$$|\phi|_{|\mathcal{H}^+_2}(u_2) +| D_{u_2} \phi|_{|\mathcal{H}^+_2}(u_2) \leq  C_2 \cdot u_2^{-s_2}.$$
			\end{hyp}
			
			\begin{hyp} We also assume the following red-shift estimate on $C_{in}^1$ and $C_{in}^2$: 
				$$| D_{U_1} \phi| (U_1, v_0) \leq C_1,$$  		   	$$| D_{V_2} \phi| (u_0, V_2) \leq C_2.$$

			\end{hyp}

			Then $\CHone \neq \emptyset $ and $\CHtwo \neq \emptyset $ and for some $U_s>0$, $ \{ 0 < U_1 \leq U_s, v_1=+\infty \} \subset \CHone$, and for some $V_s>0$, $ \{ 0 < V_2 \leq V_s, u_2=+\infty \} \subset \CHtwo$. 
			
			Moreover stability estimates hold in a region $\{ 0 \leq U_1 \leq U_s, v_1\geq v_0\} \cup \{ 0 \leq V_2 \leq V_s, u_1\geq u_0\}$.

			If we also make a lower bound assumption on $\phi$, we obtain also an instability result on the Ricci curvature: \begin{hyp}
				Assume that for some $2s_1-1 \leq p_1 <  \min \{2s_1, 4s_1-2\}$, and some $D_1>0$, the following lower bound holds $$\frac{1}{2} \int_{v_1}^{+\infty}|D_{v_1} \phi|_{|\mathcal{H}^+_1}^2(v')dv'+    q_0 e_1  [\frac{1}{r_{-,1}}-\frac{1}{r_{+,1}}]  \int_{v_1}^{+\infty} \Im(\bar{\phi}_{|\mathcal{H}^+_1}(v') D_{v_1}\phi_{|\mathcal{H}^+_1}(v')) dv' \geq D_1 \cdot v_1^{-p_1},$$ 	and assume that for $2s_2-1 \leq p_2 <  \min \{2s_2, 4s_2-2\}$, and some $D_2>0$, the following lower bound holds $$ \frac{1}{2}\int_{u_2}^{+\infty}|D_{u_2} \phi|_{|\mathcal{H}^+_2}^2(u')du' - q_0 e_2  [\frac{1}{r_{-,2}}-\frac{1}{r_{+,2}}]  \int_{u_2}^{+\infty} \Im(\bar{\phi}_{|\mathcal{H}^+_2}(u') D_{u_2} \phi_{|\mathcal{H}^+_2}(u')) du'\geq D_2 \cdot u_2^{-p_2}.$$
			\end{hyp}
			
			Then, under this additional assumption, defining the outgoing radial geodesic vector field $V= \Omega_1^{-2} \partial_{v_1}$ and the ingoing radial geodesic vector field $U= \Omega_2^{-2} \partial_{u_2}$, we have, for every $0<U_1 \leq U_S$: $$ \limsup_{v_1 \rightarrow+\infty} Ric(V,V)(U_1,v_1)=+\infty, $$ and for every $0 < V_2 \leq V_S$: $$ \limsup_{u_2 \rightarrow+\infty} Ric(U,U)(u_2,V_2)=+\infty.$$
			
		\end{enumerate}	
		
	\end{thm}	\begin{rmk}
		We emphasize that the statement of Theorem \ref{previous} was originally formulated in \cite{Moi} as a characteristic initial value problem, with data on $\mathcal{H}^+ \cup \mathcal{C}_{in} $ thus the global topology of the Penrose diagram was \textit{irrelevant} and the statement was identical in both the one-ended and the two-ended case. However, in the present paper, the distinction between one and two ended is somewhat important, which is why we phrased the theorem in this way. 
	\end{rmk}

	\begin{rmk}\label{EHremark}
		Note that assumption \ref{hypchargeevent} consists of four statements \begin{enumerate}[i]
			\item \label{EH1} $\mathcal{H}^+$ is an \underline{event horizon} i.e.\ $$ \limsup_{v\rightarrow+\infty} r_{|\mathcal{H}^+}(v) < +\infty.$$
			\item \label{EH2}  $\mathcal{H}^+$ possesses a finite  \underline{asymptotic charge} i.e.\ $$ \limsup_{v\rightarrow+\infty} |Q|_{|\mathcal{H}^+}(v)<+\infty.$$ \item \label{EH3}  $\mathcal{H}^+$ is asymptotically a  \underline{Reissner--Nordstr\"{o}m} event horizon i.e.\ the asymptotic charge is non-zero i.e.\ $$ \limsup_{v\rightarrow+\infty} |Q|_{|\mathcal{H}^+}(v)>0.$$
			\item \label{EH4}  $\mathcal{H}^+$ is asymptotically a \underline{sub-extremal} (Reissner--Nordstr\"{o}m) event horizon  $$ \limsup_{v\rightarrow+\infty} |Q|_{|\mathcal{H}^+}(v) \neq   \limsup_{v\rightarrow+\infty} r_{|\mathcal{H}^+}(v).$$
		\end{enumerate}
		Note that the sub-assumption \ref{EH2} is in fact unnecessary (although we add it for clarity), as it follows from sub-assumption \ref{EH1} and the decay of the initial data dictated by assumption \ref{fieldevent}, c.f.\ \cite{Moi}.
		
		Assumption \ref{EH3} is important: indeed, if the asymptotic charge \textbf{is} zero instead i.e.\  $\mathcal{H}^+$ is asymptotically a \textbf{Schwarzschild} event horizon, then the author proved that no Cauchy horizon is present: $\CH=\emptyset$, thus $r=0$ on $\mathcal{B}^+$, see \cite{MoiThesis}, section 2.3.4. Note, however, that the asymptotic charge is conjectured to be non-zero generically, see \cite{Kommemi}.

		Notice that sub-assumption \ref{EH4} requiring $\mathcal{H}^+$ to be a sub-extremal is superfluous in the uncharged case $q_0=0$ where the event horizon is necessarily sub-extremal (see Appendix A of \cite{JonathanStab}), a property which does not persist in the charged  $q_0 \neq 0$. The sub-extremality of the event horizon is also conjectured to be generic, c.f.\  \cite{Kommemi}.
		
	\end{rmk}
	
	\color{black}

	We are now ready to phrase our new \textit{global} results, starting from the $C^2$-inextendibility of the Cauchy horizon.
	\subsection{Global inextendibility properties across the Cauchy horizon emanating from $i^+$}
	In this section, we present our $C^2$-inextendibility (defined in section \ref{inextsection}) results, starting with the two-ended case: 
	\begin{thm} \label{conditionnalSCCtheoremtwoended}
		Let $(M,g)$ be the maximal development of admissible two-ended initial data, satisfying the assumptions of Theorem \ref{previous}. Then $(M,g)$ is $C^2$-future-inextendible.

	\end{thm}
	In the one-ended case, additional complications arise because of the potential existence of an ``outgoing Cauchy horizon'' emanating from the center $\mathcal{CH}_{\Gamma}$. Since this $\mathcal{CH}_{\Gamma}$ \textit{could} be $C^2$-extendible, we only prove the $C^2$-inextendibility of $\CH$:
	\begin{thm} \label{CHinexttheorem}
		Let $(M,g)$ be the maximal development of admissible one-ended initial data, satisfying the assumptions of Theorem \ref{previous}. Then $(M,g)$ is $C^2$-future-inextendible across $\CH$, in the sense of Definition \ref{CHinext}.
	\end{thm}
	Nevertheless, $\mathcal{CH}_{\Gamma}=\emptyset$ if one accepts Conjecture \ref{trappedsurfaceconj}. Thus, this last obstruction to $C^2$-inextendibility should disappear with a proof of Conjecture \ref{trappedsurfaceconj}, which would also imply the Weak Cosmic Censorship Conjecture (Conjecture \ref{WCC}). Note, however, that such a proof would require to study the space-time near the center-endpoint $b_{\Gamma}$ and thus would require different techniques than those employed in the present paper. In view of these considerations, it is interesting to use the $C^2$-inextendibility of $\CH$ to obtain the following ``conditional'' $C^2$-future-inextendibility of one-ended space-times:
	\begin{thm} \label{conditionnalSCCtheoremoneended}
		Let $(M,g)$ be the maximal development of admissible one-ended initial data, satisfying the assumptions of Theorem \ref{previous} and assume moreover that $\mathcal{CH}_{\Gamma} =\emptyset$, where $\mathcal{CH}_{\Gamma}$ is a priori boundary component as defined in Theorem \ref{oneendedapriori}. Then $(M,g)$ is $C^2$-future-inextendible.
	\end{thm}

	\begin{rmk}
		In these theorems, we assumed that the asymptotic black hole charge is non-zero, which is the hardest case (and conjecturally the generic one). Recall from Remark \ref{EHremark} that if the charge is zero, then $\CH=\emptyset$. Therefore, the space-time is immediately $C^2$-future-inextendible in the two-ended case (and no more work is needed, see \cite{MoiThesis}). In the one-ended case, the space-time is also $C^2$-future-inextendible, under the same assumption as Theorem \ref{conditionnalSCCtheoremoneended} i.e. if $\mathcal{CH}_{\Gamma}=\emptyset$.
	\end{rmk}
	\subsection{Classification of the Cauchy horizon emanating from $i^+$ and estimates} \label{classificationsection}
	
	
	In this section, we present our classification of the Cauchy horizons emanating from time-like infinity, which is fundamental to our proof of $C^2$-inextendibility. In what follows, we will use the generic notation $\CH$ for any Cauchy horizon emanating from time-like infinity, be it $\CH$ in the one-ended case (see Theorem \ref{oneendedapriori}), or $\CHone$, $\CHtwo$ in the two-ended case (see Theorem \ref{twoendedapriori}). Correspondingly, $\mathcal{H}^+$ will be the generic notation for the corresponding event horizon, e.g. $\mathcal{H}^+_1$ for $\CHone$. In particular, we emphasize that all our subsequent results are valid \textbf{both} in the one-ended case, and the two-ended case. We also take the convention that the end-point of $\CH$ does not belong to $\CH$. 
	
	We will parametrize \footnote{Of course, if $\CH$ is the ingoing Cauchy horizon of the one-ended case, or $\CHone$ in the two-ended case, we can chose $\tau=u$ and $\varsigma=v$.} $\CH = \{ -\infty < \tau < \tau_{\CH}\}$ by $\tau$ and $\mathcal{H}^+:= \{ \varsigma_0 \leq \varsigma \leq +\infty \}$  by $\varsigma$ in this section.
	
	\begin{defn} \label{dynamicaldef}
		We say the Cauchy horizon $\CH = \{ -\infty < \tau < \tau_{\CH}\}$ is a Cauchy Horizon of dynamical type  if there exists $\tau_s \in \RR$ such that the area-radius function $r$ extends continuously to a function $r_{CH}(\tau)$ for $-\infty< \tau \leq \tau_s$ and $\tau \rightarrow r_{CH}(\tau)$ is strictly decreasing on $(-\infty,\tau_s)$.

	\end{defn} 
	\begin{rmk}
		This definition may seem minimalist at first sight, as one would expect $r$ to be strictly decreasing on the \textit{whole} Cauchy horizon in the dynamical type case (not only on $\{\tau \leq\tau_s\}$ as in Definition \ref{dynamicaldef}). While we do not include this latter property in Definition \ref{dynamicaldef}, we prove \textit{a posteriori} that it is satisfied on any Cauchy horizon of dynamical type, see Theorem \ref{classification}. 
	\end{rmk} \color{black}
	\begin{defn}  \label{rigiddef}
		We say the Cauchy horizon $\CH = \{ -\infty < \tau < \tau_{\CH}\}$ is a Cauchy Horizon of static type if the area-radius function $r$ extends continuously to a constant function $r_{CH}(\tau) = r_0$, $r_0 >0$,  for $-\infty< \tau < \tau_{\CH} $.
	\end{defn}
	
	\begin{defn}  \label{mixeddef}
		We say the Cauchy horizon $\CH = \{ -\infty < \tau < \tau_{\CH}\}$ is a Cauchy Horizon of mixed type if there exists a transition time $-\infty < \tau_T <\tau_{\CH}$, $\epsilon>0$ with $\tau_T+\epsilon < \tau_{\CH}$ such that the area-radius function $r$ extends continuously to a function $r_{CH}(\tau)$ for $-\infty< \tau \leq \tau_{T}+\epsilon $ and there exists some constant $r_0 >0$, such that \begin{enumerate}
			\item $r_{CH}(\tau)= r_0$ for all $-\infty< \tau \leq \tau_{T} $.
			\item $\tau \rightarrow r_{CH}(\tau)$ is strictly decreasing on $(\tau_T,\tau_{T}+\epsilon]$.
		\end{enumerate}
	\end{defn}
	
	\begin{thm} \label{classification}
		We work with $\CH$, under the assumptions of Theorem \ref{previous}, and we \footnote{Those limits exist as a soft consequence of the assumptions of Theorem \ref{previous}, see \cite{Moi}.} define $M= \lim_{\varsigma \rightarrow +\infty} \varpi_{|\mathcal{H}^+}(\varsigma)$ and $e= \lim_{\varsigma \rightarrow +\infty} Q_{|\mathcal{H}^+}(\varsigma)$ on the event horizon $\mathcal{H}^+$ corresponding to $\CH$ .

		Then $r$ extends continuously to a function $r_{CH}$ on $\CH$ and there are three possibilities: 
		\begin{enumerate}
			\item $\CH$ is of dynamical type: then $\rho$ and $\varpi$ extend \footnote{By this, we mean that $\rho^{-1}$ and $\varpi^{-1}$ extend continuously to $0$ on $\CH$.} to $+\infty$ on $\CH$ and $r_{CH}$ is strictly decreasing on $(-\infty,\tau_{\CH})$.
			
			\item $\CH$ is of static type: then $r-r_-(e,M)$, $\phi$, $D_u \phi$, $\varpi-M$, $Q-e$ all extend continuously to $0$ on $\CH$.

			\item $\CH$ is of mixed type: then \begin{enumerate}
				\item $r_{CH}$ is strictly decreasing on $(\tau_T,\tau_{\CH})$.
				\item $r-r_-(e,M)$, $\phi$, $D_u \phi$, $\varpi-M$, $Q-e$ all extend continuously to $0$ on $\CH \cap \{ \tau \leq \tau_T\} $.
				
				\item $\rho$ and $\varpi$ extend to $+\infty$ on $\CH  \cap \{ \tau_T< \tau < \tau_{\CH} \}$. 
			\end{enumerate}	 Moreover, in \underline{all three cases}, the following estimates hold: for all $u_1 <u_2<u_{\CH}$, there exists $C(M,e,q_0,m^2,u_1,u_2,s)>0$ such that for all  $u_1  \leq u \leq u_2$ and $v \geq v(u)$, defining $\psi:=r\phi$, recalling \eqref{nudef}, \eqref{lambdadef} and that $K_-(M,e)<0$: \begin{equation} \label{trappedestimate1} 
				\int_{u_1}^{u_2}  |\partial_u \log(\Omega^2)|(u,v)du  \leq C \cdot (v^{2-2s} 1_{ \{s<1\}} +   1+  \log(v)^2 1_{ \{s=1\}}) ,
			\end{equation} 	\begin{equation}  \label{trappedestimate2} 
				\int_{u_1}^{u_2}  |D_u \psi|(u,v)du  \leq C ,	\end{equation} 	\begin{equation}  \label{trappedestimate3} 
				\int_{u_1}^{u_2}   |D_u \phi|(u,v)du  \leq C \cdot (v^{1-s} 1_{ \{s<1\}} +   1 +  \log(v)1_{ \{s=1\}}) ,	\end{equation} 	\begin{equation}  \label{trappedestimate4} 
				\int_{u_1}^{u_2}  |\nu|(u,v)du  \leq C ,
			\end{equation} 	\begin{equation}  \label{trappedestimate5} 
				|\lambda|(u,v)  \leq C \cdot  v^{-2s} ,
			\end{equation} 
			\begin{equation}  \label{trappedestimate6} 
				|\phi|(u,v) \leq C \cdot   (v^{1-s} 1_{ \{s<1\}} +   1 +  \log(v)1_{ \{s=1\}}),
			\end{equation} 
			\begin{equation}  \label{trappedestimate7} 
				|D_v \phi|(u,v) \leq C \cdot v^{-s},
			\end{equation}    
			\begin{equation}  \label{trappedestimate8} 
				|Q|(u,v) \leq C \cdot 	(v^{2-2s} 1_{ \{s<1\}} +   1+  \log(v)^2 1_{ \{s=1\}}),	\end{equation}   \begin{equation} \label{trappedestimate9} 
				C^{-1} \cdot  e^{2.01 K_- \cdot  v}  \leq \Omega^2(u,v) \leq C \cdot  e^{1.99 K_- \cdot  v},	\end{equation} \begin{equation}\label{trappedestimate11}
				|\partial_v\log(\Omega^2)-2K_-|
				\leq C\left(
				v^{1-2s}1_{\{s<1\}}
				+v^{-1}(1+\log v)1_{\{s=1\}}
				+v^{-s}1_{\{s>1\}}
				\right).
			\end{equation} \begin{equation}\label{trappedestimate10} 
				\int_{v(u)}^{v} Ric(V,V)(u,v')dv' \geq C \cdot  e^{1.98 |K_-| v} ,	\end{equation}
			where $V:= \Omega^{-2} \partial_v$ is a null radial geodesic vector field which is transverse to $\CH$. 
		\end{enumerate}
		\begin{rmk} \label{weirdrate}
			The decay rates $e^{2.01 K_- \cdot  v}$, $e^{1.99 K_- \cdot  v}$, $e^{1.98 K_- \cdot  v}$ were chosen as $e^{(2+\eta) K_- \cdot  v}$, $e^{ (2-\eta) K_- \cdot  v}$, $e^{(2-2\eta) K_- \cdot  v}$ where the choice $\eta=0.01$ is arbitrary (and not very important). We will use similar notations in the sequel.
		\end{rmk} \color{black}
		\begin{rmk} \label{s<1remark}
			While we assume $s>\frac{3}{4}$, the statements of our estimates
			distinguish the cases $s<1$, $s=1$ and $s>1$; see
			\eqref{trappedestimate1}, \eqref{trappedestimate3},
			\eqref{trappedestimate6} and \eqref{trappedestimate8}.
			For simplicity of exposition, we give the proofs in detail for
			$\frac{3}{4}<s<1$. For $s=1$ and $s>1$, the same arguments apply
			with the corresponding logarithmic or bounded growth estimates:
			the factors $v^{1-s}$ and $v^{2-2s}$ arising from integration
			are replaced by $1+\log v$ and $(1+\log v)^2$ when $s=1$,
			and by constants when $s>1$.
			Throughout, the exponents $s$ and $p$ in
			Theorem~\ref{previous} retain their original values.
			
		\end{rmk}
		
	\end{thm}
	
	\subsection{A trapped neighborhood of the Cauchy horizon emanating from $i^+$}
	In this section, we state a side result: there exists a trapped neighborhood $\mathcal{T}^{'}$ of $\CH$, as depicted in Figure \ref{Figtrapped1}. This result, which is of independent interest, is also used as a key ingredient in the proof of the classification of Theorm \ref{classification}.
	\begin{thm} \label{trappedtheorem} We work under the assumptions of Theorem \ref{previous} and let $\CH$ be either the Cauchy horizon emanating from time-like infinity in the one-ended case, or $\CHone$, $\CHtwo$ in the two-ended case. \\Then for every $u<u_{\CH}$, there exists $v(u) \in \RR$ such that $\{u\} \times [v(u),+\infty) \subset \mathcal{T}$. \\ In particular, there are no limit points of $\mathcal{A}$ on $\CH$.

	\end{thm}
	\subsection{A breakdown criterion to propagate the weak null singularity} In this last section, we present a breakdown criterion: essentially, if the Dafermos condition \eqref{staticrough} is satisfied on one cone, then the Hawking mass blows up on $\CH$. This result is the very first step towards the proof of the existence of a trapped neighborhood, the classification of the Cauchy horizon and ultimately the inextendibility results.
	\begin{thm} \label{propagationtheorem} Under the assumptions of Theorem \ref{previous}, assume that \underline{one} of the following (non-equivalent) conditions holds on an outgoing cone $C_{u_0}$ under $\CH$, for some $u_0 <u_{\CH}$ (recalling the definition of $\kappa$ from \eqref{kappadef}):
		$$\int_{v_0}^{+\infty} \kappa(u_0,v) dv <+\infty,$$
		$$\limsup_{v \rightarrow+\infty} \rho(u_0,v)=+\infty,$$
		$$\limsup_{v \rightarrow+\infty} \varpi(u_0,v)=+\infty,$$
		$$ \limsup_{v \rightarrow+\infty} |\nu|(u_0,v) >0,$$ 		$$\limsup_{v \rightarrow+\infty} |\phi|^2(u_0,v)+ \limsup_{v \rightarrow+\infty} |Q|(u_0,v)=+\infty,$$
		then for all $u_0 \leq u < u_{\CH}$, \underline{both} the following conditions hold
		$$\lim_{v \rightarrow+\infty} \rho(u,v)=+\infty,$$
		$$\int_{v_0}^{+\infty} \kappa(u,v)dv <+\infty.$$
		
	\end{thm} Once this result is proven, one can establish the classification of Theorem \ref{classification}. \textit{A posteriori}, once Theorem \ref{classification} is also available, we obtain the following result as a corollary of Theorem \ref{propagationtheorem}:
	\begin{cor} We make the same assumptions as Theorem \ref{propagationtheorem}, for some $u_0 <u_{\CH}$.
		
		Then either $\CH$ is of dynamical type, or $\CH$ is of mixed type with $u_0 > u_T$. 
		
		In both cases, for all $u_0 \leq u < u_{\CH}$, $ \lim_{v \rightarrow+\infty} |\partial_u r|(u,v)$ exists and $ \lim_{v \rightarrow+\infty} |\partial_u r|(u,v) >0$. 
	\end{cor}
	
	While we stated our results inductively -- from the most specific (the $C^2$ inextendibility theorems) to the most general (the breakdown criterion) -- we will, understandably, prove them deductively, starting from Theorem \ref{propagationtheorem} and finishing with Theorem \ref{conditionnalSCCtheoremtwoended}, Theorem \ref{CHinexttheorem} and Theorem \ref{conditionnalSCCtheoremoneended}. We refer to section \ref{outline} for the logic of the proof.
	
	\section{Recalling the previous estimates, underlying in Theorem \ref{previous}} \label{LB}
	In this section, we recall the estimates of \cite{Moi} which served in the proof of Theorem \ref{previous}. These estimates will be an important starting point in the proof of our present results. As explained in Remark~\ref{s<1remark}, we display the estimates
	below for $\frac{3}{4}<s<1$; for $s=1$, some factors will have logarithmic growth instead, and if $s>1$, this growth is absent.
	
	We will also renormalize the $U$ coordinate defined by \eqref{gauge1}, defining a new coordinate $u \in \RR$ (singular across $\mathcal{H}^+$) by: \begin{equation} \label{gauge1.5}
		u:= \frac{\log(U)}{2K_+(M,e)}.
	\end{equation}
	
	\begin{prop}[\cite{Moi}]\label{LBprop}
		Under the assumptions of Theorem \ref{previous}, there exists a space-like curve $\gamma$ terminating at $i^+$ and such that, in $\mathcal{LB}:=J^+(\gamma) \cap \{ u \leq u_s \}$ for some $u_s \in \RR$, we have the following estimates, in the gauge \eqref{gauge2}, \eqref{gauge1.5}, \eqref{electrogauge}:
		
		\begin{equation} \label{varpigamma}
			|\varpi(u,v_{\gamma}(u))-M| \lesssim |u|^{1-2s}.
		\end{equation}	
		\begin{equation} \label{Qgamma}
			|Q(u,v_{\gamma}(u))-e| \lesssim |u|^{1-2s}.
		\end{equation}	
		\begin{equation} \label{rgamma}
			|r(u,v_{\gamma}(u))-r_- | \lesssim |u|^{1-2s}.
		\end{equation} \begin{equation} \label{phigamma}
			|\phi|(u,v_{\gamma}(u)) \lesssim |u|^{-s},
		\end{equation}	\begin{equation} \label{Omegaexp}
			e^{2.01 K_- \cdot v}	  \lesssim 		\Omega^2 \lesssim e^{1.99 K_- \cdot v},
		\end{equation}	\begin{equation} \label{phiVLB}
			| \partial_v \phi| \lesssim v^{-s},
		\end{equation}\begin{equation} \label{QphiLB}
			|Q|+ |\phi|^2 \lesssim v^{2-2s},
		\end{equation}		
		\begin{equation} \label{phiVLBlowerbound}
			\int_{v}^{+\infty}	|\partial_v \phi|^2(u,v') dv' \gtrsim v^{-p},
		\end{equation}	\begin{equation} \label{partialvOmegaLB}
			|\partial_v \log(\Omega^2)-2K_-(M,e) |  \lesssim    v^{1-2s},
		\end{equation}					
		\begin{equation} \label{lambdaLB}
			0<	 -\lambda \lesssim  v^{-2s},
		\end{equation}
		\begin{equation} \label{nuLB}
			0< -\nu \lesssim |u|^{-2s}.
		\end{equation}
		
		As a consequence of \eqref{lambdaLB} and \eqref{nuLB}, $r$ extends continuously to $\CH$ to a continuous function $r_{CH}$ on $(-\infty,u_s]$.

	\end{prop}
	
	In the rest of the paper, we will always work with $(u,v)$ defined by the gauge \eqref{gauge2}, \eqref{gauge1.5}, unless specified otherwise.
	
	\section{Classification of Cauchy horizon types} 
	\subsection{Preliminary results} \label{lowerboundslambdasection}
	
	In this section, we provide some preliminary and easy results which will be essential in the rest of the paper. We start by an integrated lower bound on $|\lambda|$ on an outgoing cone transverse to the Cauchy horizon and sufficiently close to time-like infinity (and included in the trapped region, so that $\lambda<0$):
	\begin{lem} For all $u \leq u_s$ and $v \geq v_{\gamma}(u)$, we have the following integrated lower bound on $\lambda$:
		\begin{equation} \label{lambdaus}
			\int_{v}^{+\infty} |\lambda|(u,v') dv' \geq  D \cdot v^{-p},
		\end{equation} 
	\end{lem}
	
	\begin{proof}
		We start with the Raychaudhuri equation \eqref{RaychV} which we write as $$ -\partial_v \lambda  +\lambda \cdot \partial_v \log(\Omega^2)= r |\partial_v \phi|^2.$$
		
		Then, integrating and using \eqref{phiVLBlowerbound}, $$ \lambda(u,v)+ \int_{v}^{+\infty}\lambda(u,v') \cdot \partial_v \log(\Omega^2)(u,v') dv' \gtrsim v^{-p}.$$
		
		Then, using \eqref{lambdaLB} and the fact that $p<2s$ and the fact that $\lambda<0$ we get  $$  \int_{v}^{+\infty}|\lambda|(u,v') \cdot [-\partial_v \log(\Omega^2)(u,v') ]dv' \gtrsim v^{-p}.$$ Now using \eqref{partialvOmegaLB} as $-\partial_v \log(\Omega^2)(u,v') \leq 3 |K_-|(M,e)$ for $v$ large enough  (since $K_-(M,e)<0$) we obtain the desired \eqref{lambdaus}. \color{black}
	\end{proof}
	Next, we prove that a space-time rectangle which is entirely trapped and does not contain $(u_{\infty}(\CH),+\infty)$ (the end-point of $\CH$) has finite space-time volume:
	\begin{lem} \label{trappedvolume}
		Let $-\infty < u_1 <u_2 < u_{\infty}(\CH)$ and $v_0 \in \RR$. Assume that the rectangle $R:=[u_1,u_2] \times [v_0,+\infty]$ is included in the trapped region: $R \subset \mathcal{T}$. Then the space-time volume of $R$ is finite and we have the following estimate: $$ vol(R):=  \int_{u_1}^{u_2} \int_{v_0}^{+\infty} r^2 \Omega^2(u',v') du' dv' \leq (u_2-u_1) \cdot r^3(u_1,v_0) \cdot \sup_{u_1 \leq u' \leq u_2} \frac{\Omega^2(u',v_0)}{|\lambda|(u',v_0)}  <+\infty.$$
		
	\end{lem}
	
	\begin{proof}
		Since $R$ is trapped, for all $u_1 \leq u \leq u_2$, $v \rightarrow r(u,v)$ is decreasing, thus for all $(u,v) \in R$, $r$ is bounded: $$ r(u,v) \leq \sup_{u_1 \leq u' \leq u_2} r(u',v_0) <+\infty.$$

		Since $\lambda(u,v) <0$ for all $(u,v) \in R$, we see from \eqref{RaychV} that $ v \rightarrow \frac{\Omega^2(u,v)}{|\lambda|(u,v)}$ is non-increasing, thus $$ \frac{\Omega^2(u,v)}{|\lambda|(u,v)} \leq \sup_{u_1 \leq u' \leq u_2} \frac{\Omega^2(u',v_0)}{|\lambda|(u',v_0)}<+\infty,$$
		where the last inequality comes from the fact that $|\lambda|$ is bounded away from $0$ on $[u_1,u_2]\times \{v_0\}$. Thus: $$\int_{u_1}^{u_2} \int_{v_0}^{+\infty} r^2 \Omega^2(u',v') du' dv' \leq \left(\sup_{u_1 \leq u' \leq u_2} r(u',v_0) \right)^2 \cdot \sup_{u_1 \leq u' \leq u_2} \frac{\Omega^2(u',v_0)}{|\lambda|(u',v_0)} \int_{u_1}^{u_2} \int_{v_0}^{+\infty} -\lambda(u',v') du' dv',$$
		
		Hence, and since the volume form is $r^2 \Omega^2 du dv$, we finally obtain the finiteness of the space-time volume: $$\int_{u_1}^{u_2} \int_{v_0}^{+\infty} r^2 \Omega^2(u',v') du' dv' \leq (u_2-u_1) \cdot \left(\sup_{u_1 \leq u' \leq u_2} r(u',v_0) \right)^3 \cdot \sup_{u_1 \leq u' \leq u_2} \frac{\Omega^2(u',v_0)}{|\lambda|(u',v_0)}  \leq (u_2-u_1) \cdot r^3(u_1,v_0) \cdot \sup_{u_1 \leq u' \leq u_2} \frac{\Omega^2(u',v_0)}{|\lambda|(u',v_0)} ,$$
		where for the last estimate, we used the fact that $\partial_u r \leq 0$, which come from the admissibility condition (Definition \ref{admissibilitydef}.
	\end{proof}
	
	To finish this section, we prove a small result: in the trapped region and away from the end-point of $\CH$, the area-radius $r$ is upper and lower bounded: \begin{lem} \label{easylemma}
		Let $-\infty < u_1 <u_2 < u_{\infty}(\CH)$ and $v_0 \in \RR$. Assume that the rectangle $R:=[u_1,u_2] \times [v_0,+\infty]$ is included in the marginally trapped region: $R \subset \mathcal{T} \cup \mathcal{A} $. Then there exists $r_{inf}>0$  such that for all $(u,v) \in R$: $$  r_{inf}\leq  r \leq r_+(M,e).$$
		
		\begin{proof}
			By definition of $\CH$ and since $-\infty < u_1 <u_2 < u_{\infty}(\CH)$, it is clear that $r_{\CH}(u)$ is lower bounded on $[u_1,u_2]$ so $  r_{\CH}^{-1}(u_2) =\| r^{-1}_{\CH} \|_{ L^{\infty}([u_1,u_2])} < +\infty$, where the first equality is due to $\nu \leq 0$.
			But for all $(u,v) \in R$, since $\lambda \leq 0$, $r(u,v) \geq r_{\CH}(u) \geq r_{\CH}(u_2)=r_{inf}>0$. The upper bound is trivial: $r(u,v) \leq r_{|\mathcal{H}^+}(v) \leq r_+(M,e)$, using $\nu \leq 0$.
		\end{proof}
		
	\end{lem}
	Lemma \ref{easylemma} will be used implicitly everywhere throughout the proof, and we will very frequently omit to refer to it in the course of the argument. Additionally, in terms of notations in all that follows, we are going to assume that $r_{CH}$ is a given function on $(-\infty,u_{\infty}(\CH))$ and thus, whenever a quantity depends on $r_{inf}(u_0):=r_{CH}(u_0)$ for $u_0 < u_{\infty}(\CH)$, we are just going to write that this quantity depends on $u_0$.
	\subsection{A rigidity theorem} \label{rigiditysection}
	In this section, we start effectively the proof of our main results. Our first theorem is a ``rigidity result'': if $u_0$ is a ``static point'', then some rigidity estimates hold in the past of $u_0$, in particular the radiation is \textit{trivial} i.e. $\CH \cap \{ u \leq u_0 \}$ is an isometric copy of its Reissner--Nordstr\"{o}m counterpart. This result is one of the key ingredients in the proof:

	\begin{thm} \label{classificationtheorem} If there exists $u_0 <u_{\infty}(\CH)$ and $v \in \RR$ such that the following condition is satisfied \begin{equation} \label{MihalisPHDcondition}
			\int_{v}^{+\infty} \kappa(u_0,v')dv' = +\infty,\end{equation} then there exists $v_0=v_0(e,M,u_0) \in \RR$ such that \begin{enumerate}
			\item \label{class1} If $u_0 \geq u_s$, then \footnote{Otherwise, we already know that $\mathcal{LB} \cap \{u \leq u_s\} \subset \mathcal{T}$ (recall that $\mathcal{LB}$ is defined as the future of a space-like trapped curve $\gamma$, c.f.\ \cite{Moi}).\color{black}} we have the inclusion $[u_s,u_0] \times [v_0,+\infty) \subset \mathcal{T}$.
			\item \label{secondeasy} For all $u \leq u_0$, $$\lim_{v \rightarrow +\infty}|\nu|(u,v)=0.$$
			\item $r-r_-(e,M)$, $\phi$, $D_u \phi$, $\varpi-M$, $Q-e$ extend continuously to $0$ on $CH^{\leq u_0}_{i^+}:=\CH \cap \{ u \leq u_0\}$.
			\item The following estimates are true for all $(u',v) \in [u_s,u_0] \times [v_0,+\infty)$, for some $C=C(M,e,q_0,m^2,u_0)>1$: \begin{equation}\label{rigidity1}
				C^{-1} \cdot e^{2.01 K_- v}\leq 	\Omega^2(u',v) \leq  C \cdot e^{1.99K_- v}, 	\end{equation}  \begin{equation} \label{rigidity2}
				| \phi|(u',v)+	\color{black}	|\partial_v \phi|(u',v) \leq  C \cdot v^{-s},
			\end{equation}  \begin{equation} \label{rigidity3}
				|D_u \phi|(u',v)\leq  C \cdot \Omega^2(u',v) \cdot v^{-s} \leq  C^2 \cdot e^{1.99K_- v} \cdot v^{-s},
			\end{equation}  	\begin{equation} \label{rigidity4}
				\int_{v}^{+\infty} |\partial_v \phi|^2(u,v') dv' \geq  C \cdot v^{-p},			\end{equation}\begin{equation} \label{rigidity5}
				|Q-e|(u',v) \leq  C \cdot v^{1-2s},
			\end{equation}\begin{equation} \label{rigidity6}
				|r(u',v)-r_-(M,e)| \leq  C \cdot v^{1-2s},
			\end{equation}
			\begin{equation} \label{rigidity7}
				|\partial_v \log(\Omega^2)(u',v)-2K_-|\leq  C \cdot v^{1-2s},
			\end{equation} \begin{equation} \label{rigidity9}
				|\kappa^{-1}(u,v)-1| \leq  C \cdot v^{1-2s},
			\end{equation} \begin{equation} \label{rigidity10}
				|\varpi(u,v)-M| \leq  C \cdot v^{1-2s},
			\end{equation}
			\begin{equation} \label{rigidity11}
				|\lambda|(u',v) \leq  C \cdot v^{-2s},	\end{equation}
			\begin{equation} \label{rigidity12}
				\int_{v}^{+\infty}	|\lambda|(u',v') dv' \geq  C \cdot  v^{-p}\end{equation}
			\begin{equation} \label{rigidity13}
				|\nu|(u',v)\leq  C \cdot \Omega^2(u',v) \leq C^2 \cdot e^{1.99K_- v}. 			\end{equation}
			
			\item There exists $\epsilon_0=\epsilon_0(M,e,q_0,m^2,u_0)>0$ such that for all $0<\epsilon<\epsilon_0$, the following estimates are true for all $(u',v) \in [u_0,u_0+\epsilon] \times [v_0,+\infty)$, for some $D=D(M,e,q_0,m^2,u_0)>1$: \begin{equation}\label{rigidity14}
				D^{-2} \cdot e^{2.01K_- v} \leq  D^{-1} \cdot \Omega^2(u_0,v)\leq 	\Omega^2(u',v) \leq D \cdot \Omega^2(u_0,v)  \leq   D^2\cdot e^{1.99K_- v},  \end{equation} \begin{equation}\label{rigidity15}
				|\phi|+|Q-e| \leq D \cdot \epsilon ,  \end{equation}
			\begin{equation} \label{rigidity16}
				\int_{v}^{+\infty}	|\lambda|(u',v') dv' \geq  D \cdot  v^{-p},\end{equation}  \begin{equation} \label{rigidity17}
				|\lambda|(u',v) \leq   D \cdot  v^{-2s}.\end{equation}	\end{enumerate}
	\end{thm}
	
	We divide the proof of Theorem \ref{classificationtheorem} in several lemmata. We start to prove the result for $u_0 \leq u_s$, slightly simpler than the general case, as we already have estimates from section \ref{LB} at disposal. This is the object of the following lemma:
	
	\begin{lem} \label{<us}We define $u'_s:= \min \{ u_0, u_s\}$ and $v'_s:=v_{\gamma}(u'_s)$. Under the assumption of Theorem \ref{classificationtheorem}: \begin{enumerate}
			\item For all $u \leq u'_s$, $$\lim_{v \rightarrow +\infty}|\nu|(u,v)=0.$$
			\item $r-r_-(e,M)$, $\phi$, $D_u \phi$, $\varpi-M$, $Q-e$ extend continuously to $0$ on  $CH^{\leq u'_s}_{i^+}:=\CH \cap \{ u \leq u'_s\}$.
			\item The following additional estimates are true on $\LB \cap \{ u \leq u'_s\}$:   \begin{equation} \label{ropt}
				|r-r_-(M,e)| \lesssim v^{1-2s},
			\end{equation} 
			\begin{equation} \label{phiopt}
				|\phi|+ |\partial_v \phi| \lesssim v^{-s},
			\end{equation}  	\begin{equation} \label{Duphi2}
				|D_u \phi| \lesssim \Omega^2 \cdot v^{-s},
			\end{equation} 
			\begin{equation} \label{kappaopt}
				|\kappa^{-1}-1| \lesssim v^{1-2s},
			\end{equation} \begin{equation} \label{massopt}
				|\varpi-M| \lesssim v^{1-2s},
			\end{equation} \begin{equation} \label{Qopt}
				|Q-e| \lesssim v^{1-2s}.
			\end{equation}
		\end{enumerate}
		
	\end{lem}
	\begin{proof}
		From the Raychaudhuri equation \eqref{RaychU}, we see that $u \rightarrow \kappa(u,v)$ is a non-increasing function of $u$. Therefore, from \eqref{MihalisPHDcondition}, we see that for all $u \leq u_0$, \begin{equation} 	\int_{v}^{+\infty} \kappa(u,v')dv' = +\infty.\end{equation} 
		
		Using the estimates \eqref{QphiLB} of \cite{Moi}, we see that on $\LB=\{ (u,v), u \leq u_s, v \geq v_{\gamma}(u) \}$ we have, using also \eqref{Radius3}: \begin{equation} \label{dvnu}
			|\partial_v(r \nu)| \lesssim \Omega^2 \cdot v^{2-2s}  \lesssim e^{1.99 K_- v} \cdot v^{2-2s} \color{black}
		\end{equation} where in the last inequality, we used \eqref{Omegaexp}.
		Then, for all $u \leq u_s$, $ v \rightarrow |\partial_v(r \nu)|(u,v)$ is clearly integrable as $v \rightarrow +\infty$. Thus for some $l(u) \geq 0$, $r \nu(u,v) \rightarrow l(u)$ as $v \rightarrow +\infty$, for all $u \leq u_s$. We want to show that for all $u \leq u'_s$, $l(u)=0$. Suppose not; if $l(u_1)>0$ for some $u_1 \leq u'_s$, then there exists $v_1$ such that for all $v \geq v_1$, $r |\nu|(u_1,v)>\frac{l(u_1)}{2}$. Then, it means that, since $r$ is bounded: $$ \int_{v_1}^{+\infty}\kappa(u_1,v)dv \lesssim  \int_{v_1}^{+\infty}\frac{-\partial_v \Omega^2(u_1,v)}{l(u_1)} dv  \leq \frac{\Omega^2(u_1,v_1)}{l(u_1)}<+\infty,$$ which is a contradiction (note that we used \eqref{partialvOmegaLB} and $K_-(M,e)<0$)\color{black}. Thus, as $r$ is lower bounded, $|\nu|(u,v) \rightarrow 0$ as $v\rightarrow+\infty$ for all $u \leq u'_s$.
		
		Then, we can integrate \eqref{dvnu}, also using \eqref{Omegaexp} to obtain \begin{equation} \label{nu}
			|\nu| \lesssim e^{1.99 K_- v} \cdot v^{2-2s}  \lesssim  e^{1.98 K_- v}
		\end{equation}
		
		Recall from Proposition \ref{LBprop} that $r$ extends continuously to $r_{CH}$ on $\CH$. We write for $u_1<u_2 \leq u'_s$ and using \eqref{nu}: 
		
		$$ |r_{CH}(u_2)-r_{CH}(u_1)|= \lim_{v \rightarrow +\infty} |r(u_2,v)-r(u_1,v)| \lesssim (u_2-u_1)  \lim_{v \rightarrow +\infty}e^{1.98 K_- \cdot v}=0.$$ This means that $r_{CH}(u)$ is a constant function on $(-\infty,u'_s]$. Using \eqref{lambdaLB}, we also get $$ |r_{CH}(u)-r(u,v_{\gamma}(u))| \lesssim |u|^{1-2s},$$
		
		and taking the limit $u \rightarrow -\infty$, also using \eqref{rgamma}, we proved that $r_{CH}(u)=r_-(M,e)>0$ for all $u \leq u'_s$ and also \eqref{ropt}.
		
		Then, from  \eqref{phiVLB}, \eqref{QphiLB}, \eqref{Omegaexp} and \eqref{Field3}, we get $|\partial_v(r D_u \phi)| \lesssim \Omega^2 \cdot v^{2-2s} \lesssim e^{1.99 K_- v} \cdot v^{2-2s}$  \color{black}which is integrable. This means that for all $u\leq u'_s$, $r D_u \phi(u,v) \rightarrow l'(u)$ for some $l'(u) \in \CC$ as $v \rightarrow +\infty$. Moreover, by \eqref{Omegaexp}, for all $u\leq u'_s$ and $v$ large enough \begin{equation} \label{Duphi}
			|r D_u\phi(u,v)-l'(u)| \lesssim \Omega^2(u,v) \cdot v^{-s} \lesssim  e^{1.98 K_- \cdot v}.
		\end{equation}
		
		Now, using \eqref{Omegaexp} and \eqref{nu}, we also have the (very sub-optimal but ultimately sufficient) bound \begin{equation} \label{kappasubopt}
			\kappa^{-1}(u,v) \lesssim e^{0.03 |K_-| \cdot v}.
		\end{equation}
		
		Now, integrating \eqref{RaychU} and using the fact that $r$ is lower bounded, we see that for all $u_1 < u_2 \leq u'_s$ and for $v$ large enough we get $$ \int_{u_1}^{u_2} \frac{r^2|D_u \phi|^2}{\Omega^2}(u',v) du' \lesssim e^{0.03 |K_-| \cdot v},$$ which also implies, using \eqref{Duphi}, \eqref{Omegaexp} and Cauchy-Schwarz: $$ \int_{u_1}^{u_2} \frac{|l'(u)|^2}{\Omega^2}(u',v) du' \lesssim \int_{u_1}^{u_2} \frac{r^2|D_u \phi|^2}{\Omega^2}(u',v) du'+ \int_{u_1}^{u_2} \frac{|l'(u)-rD_u \phi|^2}{\Omega^2}(u',v) du' \lesssim  e^{0.03 |K_-| \cdot v}+(u_2-u_1)e^{-1.95 |K_-| \cdot v},$$ which then implies, using \eqref{Omegaexp} again $$ \int_{u_1}^{u_2} |l'(u)|^2 du' \lesssim e^{-1.95 |K_-| \cdot v},$$ which proves immediately that $l'(u)=0$ for all $u \leq u'_s$. As a consequence, since $r$ is lower bounded, we also obtain \eqref{Duphi2}.

		We can integrate this estimate, noticing that $\partial_u (e^{iq_0 \int_{u_{\gamma}(v)}^{u} A_u(u',v) du'} \phi)(u,v) = e^{iq_0 \int_{u_{\gamma}(v)}^{u} A_u(u',v) du'} D_u\phi(u',v)$:$$ | \phi(u,v)- e^{-iq_0 \int_{u_{\gamma}(v)}^{u} A_u(u',v) du'}\phi(u_{\gamma}(v),v)| \lesssim (u-u_{\gamma}(v)) \cdot e^{1.98 K_- \cdot v} \lesssim v \cdot e^{1.98 K \cdot v}.$$ 
		
		Using also \eqref{phiVLB}, \eqref{phigamma}, we obtain \eqref{phiopt} and that $\phi$ extends continuously to $0$ as $v\rightarrow +\infty$ on $ \{u \leq u'_s \}$.
		
		Then, we integrate \eqref{chargeUEinstein} using \eqref{Qgamma} and \eqref{phiopt}, \eqref{Duphi2} to obtain \eqref{Qopt}. Then we return to \eqref{Radius3},  now written as $$ \partial_v(4r |\nu|)= -\partial_v \Omega^2 \cdot \frac{1- \frac{e^2}{r_-(M,e)^2}}{2K_-(M,e)} + \partial_v \Omega^2  \cdot \frac{\frac{Q^2}{r^2}- \frac{e^2}{r_-(M,e)^2}+m^2r^2 |\phi|^2}{2K_-(M,e)} + \Omega^2 \cdot (1-\frac{\partial_v \log(\Omega^2)}{2K_-}) \cdot (1 -\frac{Q^2}{r^2}+m^2r^2 |\phi|^2 ).$$
		Using \eqref{partialvOmegaLB}, \eqref{phiopt} and \eqref{Qopt}, it is clear that we have in fact, for some constant $C>0$:  $$ -\partial_v \Omega^2 \cdot \frac{1- \frac{e^2}{r_-(M,e)^2} +C \cdot v^{1-2s}}{2K_-(M,e)}  \leq \partial_v(4r |\nu|) \leq -\partial_v \Omega^2 \cdot \frac{1- \frac{e^2}{r_-(M,e)^2} - C \cdot v^{1-2s}}{2K_-(M,e)}.$$ We integrate this expression on $[v,+\infty)$: first notice that integration by parts of the term $(-\partial_v \Omega^2) \cdot v^{1-2s}$ term gives $$ \int_v^{+\infty}-\partial_v \Omega^2(u,v') \cdot (v')^{1-2s} dv'= \Omega^2(u,v) \cdot v^{1-2s}-(2s-1)   \int_v^{+\infty}\frac{-\partial_v \Omega^2(u,v')}{-\partial_v \log(\Omega^2) (u,v')} \cdot (v')^{-2s} dv'= \Omega^2(u,v) \cdot [v^{1-2s}+O(v^{-2s})],$$ 
		
		where we used  \eqref{partialvOmegaLB} and a second integration by parts to obtain the last estimate. Therefore, for some $\check{C}>0$: 	\color{black}
		\begin{equation*}
			\frac{ \frac{e^2}{r_-(M,e)^2} -1 -\check{C} \cdot v^{1-2s}}{2|K_-|(M,e)} \leq  r \kappa^{-1}(u,v) \leq \frac{ \frac{e^2}{r_-(M,e)^2} -1 +\check{C} \cdot v^{1-2s}}{2|K_-|(M,e)},
		\end{equation*} which implies, again using \eqref{ropt} that for some other constant $C'>0$: \begin{equation*}
			\frac{ \frac{e^2}{r_-(M,e)^2} -1 -C' \cdot v^{1-2s}}{2|K_-|(M,e) \cdot r_-(M,e)} \leq  \kappa^{-1}(u,v) \leq \frac{ \frac{e^2}{r_-(M,e)^2} -1 +C' \cdot v^{1-2s}}{2|K_-|(M,e) \cdot r_-(M,e)}.
		\end{equation*}
		Now, we use a computation from section \ref{geometricframework} which states that $2K_-(M,e) \cdot r_-(M,e) = 1-\frac{e^2}{r_-(M,e)}$. Hence we proved \eqref{kappaopt}. 
		
		Now, using \eqref{massUEinstein}, \eqref{massVEinstein} together with \eqref{kappaopt}, \eqref{phiVLB}, \eqref{Qopt}, \eqref{phiopt}, it is easy to see that $$ |\partial_u \varpi| \lesssim \Omega^2,$$  $$ |\partial_v \varpi| \lesssim v^{-2s}.$$
		From these two estimates, we conclude that $\varpi$ extends continuously to a constant $\varpi_0 \in \RR$ on $\CH$ and that $$ |\varpi(u,v) - \varpi_0| \lesssim v^{1-2s}.$$ From \eqref{varpigamma}, we get that $\varpi_0=M$, which finally gives \eqref{massopt}.

	\end{proof}
	Lemma \ref{<us} concludes the proof of Theorem \ref{classificationtheorem} in the case $u_0 \leq u_s$. The harder case $u_0 >u_s$ remains. We will address it in the next two lemma, using Lemma \ref{<us} as a building block. 
	
	The objective is to use  bootstrap method, which we write in detail. We introduce a set $\mathcal{B}_{v_0}$ over which certain estimates are satisfied. We prove that \begin{enumerate}
		\item \label{first} The set $ \{ (u,\infty), u \in \mathcal{B}_{v_0}  \}\subset \CH$ is ``static'' i.e.\ all quantities coincide with their Reissner--Nordstr\"{o}m counterparts, like in  statement \ref{secondeasy} from Lemma \ref{<us}. Moreover, estimates analogous to those  of Lemma \ref{<us} are true $\mathcal{B}_{v_0}$.

		\item \label{third} $\mathcal{B}_{v_0}$ is entirely included in the trapped region, using the lower bound \eqref{lambdaus} of section \ref{lowerboundslambdasection}.
		
		\item \label{fourth} There exists a strictly bigger open set $\mathcal{B}'_{v_0}\supset \mathcal{B}_{v_0}$, which is trapped, hence has finite volume.
		\item \label{fifth} $L^1$-type estimates are true on $\mathcal{B}'_{v_0}$ and they imply \eqref{B1} and \eqref{B2} on $\mathcal{B}'_{v_0}$.
		
		As a consequence of these four facts, we will show that $\mathcal{B}_{v_0}$ is both open and closed, concluding the bootstrap-type argument; therefore $\mathcal{B}_{v_0}=[u_s,u_0]$, which concludes the proof of Theorem \ref{classificationtheorem}.
	\end{enumerate}
	
	The proof of step \ref{first}, together with the definition of the set $\mathcal{B}_{v_0}$, is the object of the following lemma:

	\begin{lem} \label{propagation} Under the assumptions of Theorem \ref{classificationtheorem}, we additionally assume that $u_0 >u_s$. Consider, for $\Delta>0$, the following estimates: \begin{equation}\label{B1}
			\Omega^2(u',v) \leq  \Delta \cdot e^{K_- v}, \end{equation}  \begin{equation} \label{B2}
			| \phi|(u',v) 	+	|\partial_v \phi|(u',v) \leq  \Delta \cdot v^{-s}.
		\end{equation} and $\mathcal{B}_{v_0}$, the set of $u \in [u_s,u_0]$ such that the estimates \eqref{B1}, \eqref{B2} hold on the rectangle $[u_s,u] \times [v_0,+\infty)$. 
		
		Then, there exists $\Delta=\Delta(M,e,m^2,q_0,u_0)>0$, $\tilde{v}_0=\tilde{v}_0(M,e,m^2,q_0,u_0)$ such that, for all $v_0 \geq \tilde{v}_0$: \begin{enumerate}
			\item \label{firstlemma} $\mathcal{B}_{v_0}$ is non-empty.
			\item \label{secondlemma}  If $u \in \mathcal{B}_{v_0}$, then for all $u' \leq u$
			$$\lim_{v \rightarrow +\infty}|\nu|(u',v)=0.$$
			\item  \label{thirdlemma}  If $u \in \mathcal{B}_{v_0}$, then $r-r_-(e,M)$, $\phi$, $D_u \phi$, $\varpi-M$, $Q-e$ extend continuously to $0$ on $CH^{\leq u}_{i^+}:=\CH \cap \{ u' \leq u\}$.
			
			\item \label{fourthlemma} There exists $C=C(M,e,m^2,q_0,u_0)>0$ such that for all $u\in  \mathcal{B}_{v_0}$ and for all $(u',v) \in [u_s,u] \times [v_0,+\infty)$: \begin{equation}\label{B1'}
				C^{-1} \cdot 	e^{2.01 K_- v} \leq 	\Omega^2(u',v) \leq  C \cdot e^{1.99K_- v}, 	\end{equation}  \begin{equation} \label{B2'}
				|\phi|(u',v)+	|\partial_v \phi|(u',v) \leq  C \cdot v^{-s},
			\end{equation}  \begin{equation} \label{B2.5'}
				|D_u \phi|(u',v)\leq  C \cdot \Omega^2(u',v) \cdot v^{-s}  \leq \color{black}  C^2 \cdot e^{1.99K_- v} \cdot v^{-s},
			\end{equation}  	\begin{equation} \label{B.275'}
				\int_{v}^{+\infty} |\partial_v \phi|^2(u,v') dv' \geq  C \cdot v^{-p},			\end{equation}\begin{equation} \label{B3'}
				|Q-e|(u',v) \leq  C \cdot v^{1-2s},
			\end{equation}\begin{equation} \label{B4'}
				|r(u',v)-r_-(M,e)| \leq  C \cdot v^{1-2s},
			\end{equation}
			\begin{equation} \label{B5'}
				|\partial_v \log(\Omega^2)(u',v)-2K_-|\leq  C \cdot v^{1-2s},
			\end{equation}  	\begin{equation} \label{B6'}
				|\kappa^{-1}(u,v)-1| \leq  C \cdot v^{1-2s},
			\end{equation} \begin{equation} \label{B7'}
				|\varpi(u,v)-M| \leq  C \cdot v^{1-2s},
			\end{equation}
			\begin{equation} \label{lambda}
				|r\lambda(u',v)-r\lambda(u_s,v)| \leq  C \cdot e^{1.99K_- v},	\end{equation}
			\begin{equation} \label{lambda2}
				|\lambda|(u',v) \leq  C \cdot v^{-2s},	\end{equation}
			\begin{equation} \label{nu2}
				|\nu|(u',v)\leq  C \cdot \Omega^2(u',v) \leq \color{black} C^2 \cdot e^{1.99K_- v}. 			\end{equation}

		\end{enumerate}
		
	\end{lem}
	\begin{proof}

		From Lemma \ref{<us}, if $\Delta$ is large enough, and for $v_0 \geq v_s$, we see that $u_s \in \mathcal{B}_{v_0}$, thus $\mathcal{B}_{v_0} \neq \emptyset$, thus statement \ref{firstlemma} is proven. We will chose $\Delta$ to be a large constant depending only on $M$, $e$, $m^2$, $q_0$ and $u_0$. We recall the notation $A \lesssim B$ if there exists a constant $C(M, e, m^2, q_0,u_0)>0$ such that $A \leq C \cdot B$. In this notation, $\Delta \lesssim 1$.
		
		If $u\in \mathcal{B}_{v_0}$, then, using the same method as for Lemma \ref{<us}, we can show that $r-r_-(e,M)$, $\phi$, $D_u \phi$, $\varpi-M$, $Q-e$ extend continuously to $0$ on $CH^{\leq u}_{i^+}$. This is because \eqref{B1}, \eqref{B2} imply, using \eqref{ChargeVEinstein}: $$ |\partial_v Q| \lesssim  v^{-2s},$$ which is integrable, therefore there exists $Q_+>0$ such that for all $u_s \leq u' \leq u$, $v \geq v_0$: $$|Q|(u',v) \leq Q_+.$$
		
		From this, we can prove an estimate similar to \eqref{dvnu} and following the same argument as in Lemma \ref{<us}, we prove statement \ref{secondlemma} and statement \ref{thirdlemma}, together with the following estimates:  \begin{equation*} 
			|\nu|(u',v) \lesssim  e^{ K_- v},
		\end{equation*} \begin{equation*} \label{est2}
			| \phi|(u',v)+ | \partial_v \phi|(u',v)  \lesssim  v^{-s}
		\end{equation*} \begin{equation*} 
			|D_u \phi|(u',v) \lesssim  e^{ K_- v},
		\end{equation*} \begin{equation*} \label{est4}
			|r(u',v)-r_-(M,e)|+ |Q-e| + |\varpi-M|  \lesssim  v^{1-2s},
		\end{equation*}  \begin{equation*} \label{est7}
			|r\lambda(u',v)-r\lambda(u_s,v)| \lesssim  e^{K_- v}.
		\end{equation*}
		\begin{equation*} \label{est9}
			e^{2.01 K_- v} \lesssim 	\Omega^2(u',v)\lesssim  e^{1.99 K_- v}.
		\end{equation*}  
		Then, using \eqref{Omega} with all these estimates, together with \eqref{partialvOmegaLB}, one can prove that \begin{equation*}
			|\partial_v \log(\Omega^2)-2K_-(M,e)| \lesssim  v^{1-2s},
		\end{equation*} thus for $v_0$ large enough $\partial_v \log(\Omega^2) <K_-(M,e)$ and one can repeat the argument of Lemma \ref{<us} and improve the estimates:  \begin{equation*} \label{est1}
			|\nu|(u',v) \lesssim \Omega^2(u',v) \lesssim e^{1.99 K_- v },
		\end{equation*}  \begin{equation*} \label{est3}
			|D_u \phi|(u',v) \lesssim  \Omega^2(u',v) \cdot v^{-s} \lesssim e^{1.99 K_-v} \cdot v^{-s}  ,
		\end{equation*} \begin{equation*} \label{est6}
			|\kappa^{-1}(u',v)-1| \lesssim  v^{1-2s}.
		\end{equation*} 
		The last estimate \eqref{B.275'} is then obtained with no further difficulty, using all the other estimates.

	\end{proof}
	
	In the next lemma, we prove step \ref{third} and step \ref{fourth} of the bootstrap argument. \eqref{lambda} is the crucial estimate: combining with \eqref{lambdaus} from section \ref{lowerboundslambdasection}, we can prove that for all for all $u\in \mathcal{B}_{v_0}$, $C_u \cap \{v\geq v_0\}$ is included in the trapped region.	Step \ref{fourth} is then achieved using the openness of the trapped region: 
	\begin{lem} \label{trappedlemma}
		
		\begin{enumerate} We choose $\Delta=\Delta(M,e,m^2,q_0,u_0)>0$ and  $\tilde{v}_0=\tilde{v}_0(M,e,m^2,q_0,u_0)$ as in the statement of Lemma \ref{propagation}. Then, there exists  $\breve{v}_0(M,e,m^2,q_0,u_0,s,p) \geq \tilde{v}_0$ such that for all $u\in \mathcal{B}_{\breve{v}_0}$, the following statements are true:
			\item \label{first10}  $[u_s,u] \times [\breve{v}_0,+\infty] \subset \mathcal{T}$.
			
			\item \label{second10} There exists $\epsilon_0(M,e,q_0,m^2,u_0)>0$ \underline{independent} of $u$, such that $[u_s,u+\epsilon_0] \times [\breve{v}_0,+\infty] \subset \mathcal{T}$.
			
			\item \label{third10} For all $0 <\epsilon \leq \epsilon_0$, $[u,u+\epsilon] \times [\breve{v}_0,+\infty]$ has finite space-time volume and we have the estimate $$ \int_{u}^{u+\epsilon} \int_{\breve{v}_0}^{+\infty} r^2 \Omega^2(u',v') du' dv'\leq 8 \epsilon  \cdot r_-^3 \sup_{u \leq u' \leq u+\epsilon} \frac{\Omega^2(u',\breve{v}_0)}{|\lambda|(u',\breve{v}_0)} <+\infty.$$
			
		\end{enumerate}
	\end{lem}\begin{proof} First, we start to \textit{choose} $\breve{v}_0$. \eqref{lambdaus} shows that for all $v_1 \geq \tilde{v}_0>0$, there exists $v_2 \geq v_1$ such that \begin{equation} \label{lambdalowerpointwise}
			|\lambda|(u_s,v_2) \geq D \cdot p \cdot v_2^{-p}.
		\end{equation} 
		Indeed, it is easy to show \eqref{lambdalowerpointwise} by a small contradiction argument. We will choose  $v_1 \geq \tilde{v}_0$ large enough so that \begin{equation} \label{v_1choice}
			e^{1.99 K_- v} \cdot v^{p} < \frac{ D \cdot p}{100C \cdot r_-(M,e)}
		\end{equation} 
		for all $v\geq v_1$, where $C>0$ is the constant present on the right-hand-side of \eqref{lambda} (recall that $K_-<0$). Then we choose $v_2 \geq v_1$ so that \eqref{lambdalowerpointwise} is satisfied. Note that, since $v_2 \geq v_1$, \eqref{v_1choice} is also satisfied for $v=v_2$. We will now choose $\breve{v_0}=v_2$.
		
		\color{black}
		For statement \ref{first10}, note that it is enough to prove that $[u_s,u] \times \{\breve{v}_0\} \subset \mathcal{T}$, by the monotonicity induced by \eqref{RaychV}. To obtain this statement, we work by contradiction: assume that there exists $u_{\mathcal{R}}(\breve{v}_0) \in [u_s,u]$ with $\lambda(u_{\mathcal{R}}(\breve{v}_0),\breve{v}_0) \geq 0$. 
		
		\color{black} Then, since for any fixed $\breve{v}_0$, $u \rightarrow \lambda(u,\breve{v}_0)$ is a continuous\footnote{We are using, implicitly, a standard local well-posedness argument for characteristic initial data \underline{entirely} inside the space-time, i.e.\ whose closure is disjoint from the boundary of the Penrose diagram.} function, there exists $u_{\mathcal{A}}(\breve{v}_0) \in [u_s,u]$ such that $\lambda(u_{\mathcal{A}}(\breve{v}_0),\breve{v}_0) = 0$. Taking $u=u_{\mathcal{A}}(\breve{v}_0)$ in \eqref{lambda} we get also using \eqref{B4'} as $r(u_s,v)>\frac{r_-}{2}$: $$ |\lambda|(u_s,\breve{v}_0) \leq \frac{C \cdot e^{1.99 K_- \breve{v}_0 }}{r(u_s,\breve{v}_0)} \leq  \frac{2C \cdot e^{1.99 K_- \breve{v}_0 }}{r_-} \leq  \frac{ D \cdot p \cdot \breve{v}_0^{-p} }{50} ,$$ where for the last inequality we used \eqref{v_1choice}. This last estimate clearly contradicts \eqref{lambdalowerpointwise}.	
		\color{black} Thus, $[u_s,u] \times [\breve{v}_0,+\infty] \subset \mathcal{T}$.
		
		In particular, $(u,\breve{v}_0) \in \mathcal{T}$. Since $\mathcal{T}$ is an open set in the $\RR^{1+1}$ topology of the Penrose diagram, there exists, in particular,  $\epsilon_0=C(M,e,q_0,m^2,u_0) \cdot e^{1.99 K_- \breve{v}_0} \breve{v}_0^{-2s}>0$ such that $[u,u+\epsilon_0] \times \{\breve{v}_0\} \subset \mathcal{T}$.

		Then, using the monotonicity from \eqref{RaychV}, we see that a rectangle is trapped: $[u,u+\epsilon_0] \times [\breve{v}_0,+\infty) \subset \mathcal{T}$. This provides a proof of statement \ref{second10}.
		
		Then, using Lemma \ref{trappedvolume}, we get that the rectangle $[u,u+\epsilon_0] \times [\breve{v}_0,+\infty)$ has finite space-time volume and we also obtain the claimed estimate, also taking advantage of $r(u,v) \leq 2r_-$ for $v \geq \breve{v}_0$ and $\breve{v}_0$ large enough, using \eqref{B4'}. This concludes the proof of statement \ref{third10}.
		
	\end{proof}
	In the last part of the proof of Theorem \ref{classificationtheorem}, we prove step \ref{fifth} and conclude the bootstrap argument. The $L^1$-type estimates, using finiteness of the space-time volume of $ [u_s,u+\epsilon_0] \times [v_0,+\infty]$ proven in step \ref{fourth}, are inspired by the work \footnote{Precisely Lemma 10.3, in which it is assumed that the rectangle has finite volume. The finiteness of the volume is later obtained by a different argument, using a monotonicity property specific to the uncharged and massless considered in \cite{JonathanStab}.} of Luk and Oh \cite{JonathanStab} on uncharged and massless scalar fields. 
	
	Some important modifications are, however, carried out to accommodate the case of a variable charge . First, we also couple the $L^1$ estimates with bootstraped $L^{\infty}$ estimates on $Q$ (because the field is charged) and $L^{\infty}$ estimates on $\phi$ (because the field is massive). In view of the slow decay assumed on the event horizon, consistent with the expected decay of massive and/or charged fields c.f. section \ref{decayconj}, the analogue of estimates (10.6) and (10.7) of Lemma 10.3 of \cite{JonathanStab} do not hold, because $\partial_v \log(\Omega^2)-2K-$ and $\partial_v \phi$ are no longer integrable in the case $s<1$ (this is expected for massive and charged scalar fields). Yet, we can still prove the analogue of (10.5) and the ``ingoing'' parts of (10.6) and (10.7):

	\begin{lem} \label{finitevolumeestlemma} We choose $\Delta=\Delta(M,e,m^2,q_0,u_0)>0$ as in the statement of Lemma \ref{propagation}, and $\epsilon_0=\epsilon_0(M,e,q_0,m^2,u_0)>0$, and  $\breve{v}_0=\breve{v}_0(M,e,m^2,q_0,u_0,s,p)$ as in the statement of Lemma \ref{trappedlemma}. For $u\in \mathcal{B}_{\breve{v}_0}$, assume without loss of generality that $u+\epsilon_0 < u_{\infty}(\CH)$. 
		Then there exists $v_0(M,e,m^2,q_0,u_0,s,p)\geq \breve{v}_0$  and $0<\epsilon \leq \epsilon_0$ such that the following estimates are true, for some $D(M,e,q_0,m^2,u_0,s,p)>0$ \underline{independent} of $u$ and of $\epsilon$: \begin{equation} \label{dulogomegaL1}
			\int_{u}^{u+\epsilon} \sup_{v \in [v_0,+\infty)} |\partial_u \log( \frac{\Omega^2(u,v)}{\Omega^2(u,v_0)})|du  \leq D.
		\end{equation} 	\begin{equation} \label{duphiL1}
			\int_{u}^{u+\epsilon} \sup_{v \in [v_0,+\infty)} |D_u \phi|(u,v)du  \leq D.
		\end{equation} 	\begin{equation} \label{durL1}
			\int_{u}^{u+\epsilon} \sup_{v \in [v_0,+\infty)} |\nu|(u,v)du  \leq D.
		\end{equation} 	\begin{equation} \label{dvrL1}
			\int_{v_0}^{+\infty} \sup_{u' \in [u,u+\epsilon]} |\lambda|(u',v') dv' \leq D.
		\end{equation} 	\begin{equation} \label{phiQrLinfty}
			\sup_{(u',v) \in [u,u+\epsilon] \times [v_0,+\infty]} |Q|(u',v)+ |\phi|(u',v) \leq D.
		\end{equation}
		
		As a consequence, $\mathcal{B}_{v_0}$ is a closed and open set in the topology of $[u_s,u_0]$, therefore  $\mathcal{B}_{v_0}=[u_s,u_0]$. 
	\end{lem}
	
	\begin{rmk} \label{estimates>u0} Notice that we assumed $u+\epsilon < u_{\infty}(\CH)$, in order to work with a lower bound on $r$, c.f. Lemma \ref{easylemma}, but we are allowed to have $u+\epsilon > u_0$: in particular, since $u_0 \in \mathcal{B}_{v_0}$ (a posteriori), this gives estimates \eqref{dulogomegaL1}, \eqref{duphiL1}, \eqref{dvrL1}, \eqref{phiQrLinfty} on $[u_0,u_0+\epsilon] \times [v_0,+\infty]$, with $\epsilon>0$ depending, however, on $u_0$.
	\end{rmk} \begin{rmk}
		In fact, for every $\eta>0$ small enough, there exists $v_0(\eta)$ large enough and $\epsilon(\eta)$ small enough such that the estimates of Lemma \ref{finitevolumeestlemma} hold with $D$ replaced by $\eta$. We will not, however, require such ``smallness estimates''.
	\end{rmk}
	\begin{proof} Take $v_0 \geq \breve{v_0}$, to be chosen more precisely later to satisfy certain inequalities. First, by local well-posedness, notice that $u' \rightarrow \frac{\Omega^2(u',v_0)}{|\lambda|(u',v_0)}$ is continuous thus there exists $\eta(v_0)>0$ such that, if $0 \leq \epsilon \leq \eta$ then for all $u \leq u' \leq u+\epsilon$: $$\frac{\Omega^2(u,v_0)}{2|\lambda|(u,v_0)} \leq \frac{\Omega^2(u',v_0)}{|\lambda|(u',v_0)} \leq \frac{2\Omega^2(u,v_0)}{|\lambda|(u,v_0)}. $$
		
		Define the rectangle $R= [u,u+\epsilon] \times [v_0,+\infty)$ and its volume $vol(R)$. Lemma \ref{propagation} and estimate \eqref{B4'} imply that $vol(R)$ is finite and we have the estimate $$\int_{u}^{u+\epsilon} \int_{v_0}^{+\infty} \Omega^2 du dv \lesssim vol(R) \lesssim \epsilon \cdot \frac{\Omega^2(u,v_0)}{|\lambda|(u,v_0)} \lesssim \epsilon \cdot \frac{ e^{1.99 K_- v_0}}{r|\lambda|(u_s,v_0) -C \cdot r(u_s,v_0) \cdot e^{1.99 K_-  v_0} }  \lesssim   \epsilon \cdot  e^{1.99 K_- v_0} \cdot v_0^{p+1}\color{black} \lesssim \epsilon,$$ where we used \eqref{B1'}, and we picked $v_0$ large enough and such that $|\lambda|(u_s,v_0) \geq C \cdot v_0^{-p-1}$ by \eqref{lambda} (see \eqref{lambdalowerpointwise} and the proof of Lemma \ref{trappedlemma}).\color{black} \\	Then, we make the following bootstrap assumptions, for $v\geq v_0$ and $u \leq u' \leq u+ \epsilon$:  \begin{equation} \label{B1''}
			|Q|(u',v) \leq 2|e|,
		\end{equation}
		\begin{equation} \label{B2''}
			|\phi|(u',v) \leq 2|e|.
		\end{equation}
		
		More precisely, we consider the set $\mathcal{B}(e) \subset R$ of space-time spheres $(u'',v'') \in R$, for which \eqref{B1''} and \eqref{B2''} are satisfied for all $u \leq u' \leq u''$ and for all $v_0 \leq v \leq v''$. Notice that $ \{u\} \times [v_0,+\infty) \subset \mathcal{B}(e)$, for $v_0$ large enough, hence $\mathcal{B}(e) \neq \emptyset$.
		
		From now on, we choose $(u'',v'') \in \mathcal{B}(e) $ and we will make $L^1$-based estimates on the rectangle $R(u'',v'')= [u,u''] \times [v_0,v''] \subset R$. We start integrating \eqref{Radius3} in the $u$ direction: we get, for all $(u',v') \in R(u'',v'')$ for some $C(M,e,q_0,m^2)>0$
		
		$$ r|\lambda|(u',v') \leq r|\lambda|(u,v') +C \cdot \int_{u}^{u'}  \Omega^2(u''',v') du''',$$ where we used bootstraps \eqref{B1''} and \eqref{B2''}.
		Now, taking a $\sup$ and then integrating in $v$, we get that for some $C'(M,e,q_0,m^2,u_0)>0$: \begin{equation} \label{estim4}
			\int_{v_0}^{v''}\sup_{ u \leq u' \leq u'' }r|\lambda|(u',v') dv' \leq C' \cdot( \epsilon+v_0^{1-2s}),
		\end{equation} where in the last line, we used the space-time volume estimate, and \eqref{lambda2}. Similarly, we can integrate \eqref{Radius3} in the $v$ direction and obtain \begin{equation} \label{estim3}
			\int_{u}^{u''}\sup_{ v_0 \leq v' \leq v'' }r|\nu|(u',v') du' \leq C' \cdot \epsilon + \frac{ r^2(u,v_0)- r^2(u'',v_0)}{2} \leq C' \cdot \epsilon +\delta(\epsilon,v_0),  \end{equation} where $\delta(\epsilon,v_0) \rightarrow 0$ as $\epsilon \rightarrow 0$, using the continuity of $u' \rightarrow r^2(u',v_0)$. Then, we integrate \eqref{Field3}, we get 
		$$ r|D_u \phi|(u',v') \leq r|D_u \phi|(u',v_0) +C \cdot \int_{v_0}^{v'}  \Omega^2(u',v''') dv'''+ |\int_{v_0}^{v'} \nu \cdot\partial_v \phi(u',v')  dv'|.$$
		
		For the last term of this estimate, we must integrate by parts, using \eqref{Radius}\color{black}, as $$ \int_{v_0}^{v'}  \nu \cdot \color{black} \partial_v \phi(u',v')  dv' =  \nu \cdot \color{black} \phi(u',v')- \nu \cdot \color{black} \phi(u',v_0) +  \int_{v_0}^{v'} \frac{ \nu \cdot \lambda}{r}  \color{black}\phi(u',v''')  dv'''+  \int_{v_0}^{v'}\frac{\Omega^2 \phi}{4r}\cdot (1- \frac{Q^2}{r^2}-m^2 r^2 |\phi|^2)(u',v''') dv''',$$
		which we can estimate, using bootstraps \eqref{B1''} and \eqref{B2''} together with $ \int_{v_0}^{v'}  |\lambda|(u',v''')dv'''  \leq C' \cdot( \epsilon+v_0^{1-2s})\lesssim 1$: $$   | \int_{v_0}^{v'}  \nu \cdot \color{black} \partial_v \phi(u',v')  dv'| \lesssim  \sup_{ v_0 \leq v' \leq v'' }r|\nu|(u',v')  + \int_{v_0}^{v'} \Omega^2(u',v''')dv''',$$  where we also used the fact that $r$ is lower bounded \color{black}. Thus, we also get  $$ \sup_{ v_0 \leq v' \leq v''}r|D_u \phi|(u',v') \leq r|D_u \phi|(u',v_0) +C' \cdot \left( \int_{v_0}^{v''}  \Omega^2(u',v''') dv'''+  \sup_{ v_0 \leq v' \leq v'' }r|\nu|(u',v') \right).$$
		
		Now, $u' \rightarrow r|D_u \phi|(u',v_0)$ is continuous thus there exists $\eta'_0(v_0)>0$ such that, if $0 \leq \epsilon \leq \eta'$ then for all $u \leq u' \leq u+\epsilon$: $$\frac{r|D_u \phi|(u,v_0)}{2}\leq r|D_u \phi|(u',v_0)\leq 2r|D_u \phi|(u,v_0) \lesssim e^{1.99 K_- v_0},$$ where we also used \eqref{B2.5'}. Combining these estimates with \eqref{estim3} we also get: \begin{equation}  \label{estim2}
			\int_{u}^{u+\epsilon}\sup_{ v_0 \leq v' \leq v''}|D_u \phi|(u',v') du' \lesssim \epsilon +\delta(\epsilon,v_0)+ e^{1.99 K_- v_0}.
		\end{equation}

		Integrating \eqref{estim2}, also using \eqref{B2.5'}, we get, for some $C''=C''(M,e,q_0,m^2,u_0)>0$: \begin{equation} \label{estim5}
			|\phi|(u'',v'') \leq C''  \cdot (\epsilon +\delta(\epsilon,v_0)+ v_0^{-s}).
		\end{equation}
		
		Now we can choose $v_0(M,e,q_0,m^2,u_0,s)$ large enough such that $C'' \cdot v_0^{-s} \leq \frac{|e|}{4}$. Then $v_0$ is fixed for the rest of the proof. Then, we can chose $\epsilon$ small enough so that $C''  \cdot (\epsilon +\delta(\epsilon,v_0)) \leq  \frac{|e|}{2}$; thus bootstrap \eqref{B2''} is retrieved.

		Now, integrating \eqref{chargeUEinstein} using \eqref{B3'} and \eqref{estim5}, we also retrieve \eqref{B1''}. Therefore, we proved that for all $(u',v) \in R$: \begin{equation} \label{estim7}
			|Q|(u',v) \leq 2|e|,
		\end{equation}
		\begin{equation} \label{estim8}
			|\phi|(u',v) \leq 2|e|.
		\end{equation}
		
		We obtain that the $L^1$ estimates \eqref{estim4}, \eqref{estim3}, \eqref{estim2}, \eqref{estim5} are valid on $R$. This proves \eqref{duphiL1}, \eqref{durL1}, \eqref{dvrL1}, \eqref{phiQrLinfty}.
		
		Finally, we must also prove \eqref{dulogomegaL1}. We integrate \eqref{Omega} in the $v$ direction and use similar estimates as before: $$ |\partial_u \log(\Omega^2)(u',v')-\partial_u \log(\Omega^2)(u',v_0)| \lesssim   | \int_{v_0}^{v'} \Re(\overline{D_u \phi} \partial_v \phi)(u',v''')dv'''|+\sup_{ v_0 \leq v' \leq v'' }r|\nu|(u',v')  + \int_{v_0}^{v'} \Omega^2(u',v''')dv'''$$
		
		We make use of an integration by parts for the first term, using \eqref{Field3}: \begin{equation*}\begin{split}
				\int_{v_0}^{v'} \Re(\overline{D_u \phi} \partial_v \phi)(u',v''')dv'''=  \Re( \phi\overline{D_u \phi})(u',v')-\Re(\phi\overline{D_u \phi})(u',v_0)+\\  \int_{v_0}^{v'} \phi\left(\frac{\partial_{u}r \partial_{v}\phi }{r}
				+\frac{ m^{2}  \Omega^{2}}{4} \phi+\frac{ q_{0}i \Omega^{2}}{4r^2}Q \phi \right)
				(u',v''')dv''' + \int_{v_0}^{v'} \frac{\lambda}{r^2}\Re(\overline{D_u \phi} \ \phi)(u',v''')dv'''.\end{split} 	\end{equation*} 
		
		Now, using \eqref{estim7}, \eqref{estim8} we get $$ |\int_{v_0}^{v'} \Re(\overline{D_u \phi} \partial_v \phi)(u',v''')dv'''| \lesssim \sup_{ v_0 \leq v' \leq v''}|D_u \phi|(u',v')+ \sup_{ v_0 \leq v' \leq v'' }r|\nu|(u',v')  + \int_{v_0}^{v'} \Omega^2(u',v''')dv''',$$ where to obtain this inequality, \color{black} we estimated, in a similar way to what was done before (involving an integration by parts): $$   | \int_{v_0}^{v'} \frac{\nu}{r} \phi \partial_v \phi(u',v')  dv'| \lesssim  \sup_{ v_0 \leq v' \leq v'' }r|\nu|(u',v')  + \int_{v_0}^{v'} \Omega^2(u',v''')dv''' .$$
		
		Thus, integrating in $u$ and combining with \eqref{estim2} we obtain $$ |\partial_u \log(\Omega^2)(u',v')-\partial_u \log(\Omega^2)(u',v_0)| \lesssim \sup_{ v_0 \leq v' \leq v'' }r|\nu|(u',v')  + \int_{v_0}^{v'} \Omega^2(u',v''')dv'''.$$
		
		Now we can take a $\sup$, integrate in $u$ and use \eqref{estim3} with the volume estimate to finally obtain \eqref{dulogomegaL1}.

		As a consequence of \eqref{dulogomegaL1}, we see that for all $u \leq u' \leq u+\epsilon$, $v \geq v_0$: $$ |\log(\frac{\Omega^2(u',v)}{\Omega^2(u',v_0)})- \log(\frac{\Omega^2(u,v)}{\Omega^2(u,v_0)})| \leq D, $$ hence using \eqref{B1'} we have $$ \Omega^2(u',v) \lesssim e^D \cdot e^{1.99 K_- v} \lesssim  e^{1.99 K_- v}.$$ 
		
		Thus, combined with \eqref{phiQrLinfty}, it means that for some $\Delta$ large enough, $[u,u+\epsilon] \in \mathcal{B}_{v_0}$ so  $\mathcal{B}_{v_0}$ is open.
		
		Now we want to show that $\mathcal{B}_{v_0}$ is closed: let $u_n \in  \mathcal{B}_{v_0}$ to be a sequence of points converging to some $u_s \leq u_{lim} \leq u_0$ as $n \rightarrow+\infty$. Then, using Lemma \ref{trappedlemma}, we see that there exists $\epsilon(M,e,q_0,m^2,u_0)>0$ \underline{independent} of $n$  such that $[u_s,u_n +\epsilon] \times [v_0,+\infty] \subset \mathcal{T}$. Take $n$ large enough so that $u_n \geq u_{lim}-\frac{\epsilon}{2}$: then we have that $[u_s,u'_{lim} ] \times [v_0,+\infty] \subset \mathcal{T}$, where $u'_{lim}= \min \{ u_{lim} +\frac{\epsilon}{2}, u_0\}$. In any case, this means that $[u_s,u_{lim} ] \times [v_0,+\infty] \subset \mathcal{T}$, which implies, using the same argument as developed in Lemma \ref{trappedlemma} and the present lemma, that $u_{lim} \in  \mathcal{B}_{v_0}$. Thus $\mathcal{B}_{v_0}$ is closed.
		
		Since $[u_s,u_0]$ is a connected interval and that $\mathcal{B}_{v_0}$ is non-empty, this implies that $\mathcal{B}_{v_0}=[u_s,u_0]$.
		
	\end{proof}

	\subsection{Static points and Dafermos points} \label{classificationsection1}
	
	Recall that the ``staticity condition'' \eqref{MihalisPHDcondition}, which is equivalent to \eqref{staticconditionintro} (see section \ref{geometricframework}), is gauge-independent (see the discussion in section \ref{outline}). It was first introduced by Dafermos in \cite{MihalisPHD}, in the context of his proof of mass inflation in the interior of dynamical black holes, for the Einstein--Maxwell-(uncharged)-scalar-field model. While Dafermos does not use of this staticity condition \eqref{MihalisPHDcondition} in his proof, he effectively produces \footnote{Making point-wise lower bounds assumptions on the scalar field on the event horizon, and exploiting a special monotonicity property which is specific to the uncharged and massless case.} a connected portion of the Cauchy horizon on which condition \eqref{MihalisPHDcondition} is \underline{violated}, for all $u_0 \leq u_s$, for some $u_s \in \RR$. Dafermos notices the difference between the space-time he constructs, for which \eqref{MihalisPHDcondition} is never \footnote{As noticed in \cite{Mihalisnospacelike}, the monotonicity of the Hawking mass in the uncharged and massless case allows to propagate the mass blow-up to the entire Cauchy horizon, which implies that condition \eqref{MihalisPHDcondition} is violated \textbf{everywhere} on $\CH$ for the space-times under consideration. This technique does not, however, survive when the field is massive and/or charged.} satisfied, and the Reissner--Nordstr\"{o}m space-time,  which satisfies condition \eqref{MihalisPHDcondition} for \textbf{all} $u_0 \in \RR$. We now introduce the set of ``static points on $\CH$'' for which \eqref{MihalisPHDcondition} is true:
	
	\begin{defn} For $u_0 \in \RR$ and $(u_0,v=+\infty) \in \CH$, we say that $(u_0,+\infty)$ is a static point of $\CH$ if the condition \eqref{MihalisPHDcondition} is true at $u_0$, for some $v \in \RR$. We define the static set $\mathcal{S}_0 \subset \CH$ as the collection of static points of $\CH$. By abuse of notation, we also denote $\mathcal{S}_0$, the projection of $\mathcal{S}_0$ on its first component:  $\{ u_0, \hskip 1 mm(u_0,+\infty) \in \mathcal{S}_0 \}$.
	\end{defn}
	
	If, on the contrary, \eqref{MihalisPHDcondition} is violated at $u$, i.e.\ $u\in \CH-\mathcal{S}_0$, i.e\ \begin{equation}\label{Dafermoscondition}
		\int_v^{+\infty} \kappa(u,v')dv'<+\infty;
	\end{equation}  $u$ is called a Dafermos point and \eqref{Dafermoscondition} the Dafermos condition, which is equivalent to \eqref{staticrough} (see section \ref{geometricframework}).
	
	Note that the Dafermos set $\CH-\mathcal{S}_0$ is an ``increasing set'' i.e. for all $u_1 \leq u_2$, $u_1 \in \CH-\mathcal{S}_0$ implies $u_2 \in \CH-\mathcal{S}_0$: we obtain this property immediately from the Raychaudhuri equation \eqref{RaychU} and the null energy condition $\mathbb{T}_{u u}=2|D_u \phi|^2 \geq 0$. Equivalently, the static set $\mathcal{S}_0$ is a ``decreasing set'': for all $u_1 \leq u_2$, $u_2 \in \mathcal{S}_0$ implies $u_1 \in \mathcal{S}_0$.

	The rigidity theorem \ref{classificationtheorem} imposes additional constraints on the static set $\mathcal{S}_0$. In the next result, which follows almost immediately from Theorem \ref{classificationtheorem}, we show that the static set, at least away\footnote{This additional assumption is required to obtain a lower bound on $r$, which is not necessarily valid as one approaches $\{(u_{\infty}(\CH),v=+\infty)\}$ the end-point towards which $r$ may (or may not) approach $0$, c.f. \cite{Kommemi}.} from $(u_{\infty}(\CH),v=+\infty)$ the future end point of $\CH$, must be an (possibly empty) interval and a neighborhood of $(-\infty,v=+\infty)$:

	\begin{cor} \label{supcorollary}
		For all $u_1 <u_{\infty}(\CH)$, the static set $ \mathcal{S}_0^{\leq u_1}:=\mathcal{S}_0 \cap \{u \leq u_1\}$ is a closed set in the topology of $(-\infty,u_1]$.
		
		If $\mathcal{S}_0 \neq \emptyset$, we define $u_D(u_1):= \sup \mathcal{S}_0^{\leq u_1}= \max \mathcal{S}_0^{\leq u_1}$. Then $\mathcal{S}_0^{\leq u_1}:= (-\infty,u_D(u_1)]$.
		
	\end{cor} 
	
	\begin{proof} Let $u_1 <u_{\infty}(\CH)$. If $ \mathcal{S}_0^{\leq u_1} = \emptyset$ there is nothing to prove. If not, there exists $u_0 \leq u_1$ such that \eqref{MihalisPHDcondition} is true. 
		
		The proof that $\mathcal{S}_0^{\leq u_1}$ is closed is roughly similar to the proof that $\mathcal{B}_{v_0}$ is closed in Lemma \ref{finitevolumeestlemma}: if we have a sequence $u_n \in \mathcal{S}^{\leq u_1} \rightarrow u_{lim} \leq u_1$ as $n \rightarrow +\infty$, then there exists $\epsilon(u_1)>0$ independent of $n$ such that estimates \eqref{dulogomegaL1}, \eqref{duphiL1},\eqref{durL1}, \eqref{dvrL1}, \eqref{phiQrLinfty} are true on a rectangle $[u_{lim}-\frac{\epsilon}{2} , u_{lim}+\frac{\epsilon}{2} ] \times [v_0,+\infty]$, see Remark \ref{estimates>u0}. 
		
		With those estimates, one can re-do the proof of Lemma \ref{propagation}: we start from the following estimate, obtain from \eqref{dulogomegaL1}: $$\Omega^2 \lesssim e^{1.99K_-v},$$
		and then all the estimate of Lemma \ref{propagation} follow on $[u_{lim}-\frac{\epsilon}{2} , \min\{u_1, u_{lim}+\frac{\epsilon}{2}\}] \times [v_0,+\infty]$, in particular \eqref{B6'}, which shows that the staticity condition \eqref{MihalisPHDcondition} is satisfied at $u_{lim}$: $u_{lim} \in \mathcal{S}^{\leq u_1}$. \\	
		Since $\mathcal{S}^{\leq u_1}$ is non-empty and compact, we can define $u_D(u_1)=\sup \mathcal{S}^{\leq u_1}= \max \mathcal{S}^{\leq u_1}$. \\ Then, the monotonicity of \eqref{RaychU} implies that \eqref{MihalisPHDcondition} is satisfied for all $u \leq u_D(u_1)$ so $\mathcal{S}^{\leq u_1}=(-\infty,u_D(u_1)]$.
		
	\end{proof}
	
	\subsection{Three types of Cauchy horizons emanating from time-like infinity}  \label{classificationsection2}
	We now obtain a first version of the classification of Theorem \ref{classification}: we can assert that $\CH$ is either of dynamical type, or static type, or mixed type (following the definitions of section \ref{classificationsection}).
	
	\begin{cor} \label{threetypes}
		$\CH$ is either of dynamical, static of mixed type. More precisely, we have the following possibilities:
		\begin{enumerate}
			\item The static set is empty $\mathcal{S}_0 = \emptyset$: then $\CH$ is of dynamical type.
			
			\item $\mathcal{S}_0 \neq  \emptyset$ and $\sup_{u <u_{\infty}(\CH)} u_D(u) <u_{\infty}(\CH)$: then $\CH$ is of mixed type. We then define the transition retarded time as $u_T:=\sup_{u <u_{\infty}(\CH)} u_D(u)$. Then $\mathcal{S}_0=(-\infty,u_T]$ and the following properties are true:  \begin{enumerate}
				\item \label{a}If $u_T \geq u_s$, then for some $\epsilon>0$, we have the inclusion $[u_s,u_T+\epsilon] \times [v_{\gamma}(u_s),+\infty) \subset \mathcal{T}$.
				\item \label{b}For all $u \leq u_T$, $$\lim_{v \rightarrow +\infty}|\nu|(u,v)=0.$$
				\item \label{c} $r-r_-(e,M)$, $\phi$, $D_u \phi$, $\varpi-M$, $Q-e$ extend continuously to $0$ on  $CH^+_{\leq u_T}=\{ u \leq u_T, v=+\infty\}$.
				\item \label{d} There exists $v_T=v_T(M,e,u_T)$ such that the estimates are true \eqref{rigidity1}, \eqref{rigidity2}, \eqref{rigidity3}, \eqref{rigidity4}, \eqref{rigidity5}, \eqref{rigidity6}, \eqref{rigidity7}, \eqref{rigidity9}, \eqref{rigidity10}, \eqref{rigidity11}, \eqref{rigidity12} on $[u_s,u_T] \times [v_T,+\infty)$ for $C=C(M,e,q_0,m^2,u_T)>1$.
				\item \label{e} There exists $\epsilon_T=\epsilon_T(M,e,q_0,m^2,u_T)>0$ such that for all $0<\epsilon<\epsilon_T$, the estimates \eqref{rigidity14}, \eqref{rigidity15}, \eqref{rigidity16}, \eqref{rigidity17} are true for all $(u',v) \in [u_T,u_T+\epsilon_T] \times [v_T,+\infty)$, for some $D=D(M,e,q_0,m^2,u_T)>1$.
				
			\end{enumerate}
			\item \label{rigid} $\mathcal{S}_0 \neq  \emptyset$ and $\sup_{u <u_{\infty}(\CH)} u_D(u) =u_{\infty}(\CH)$: then $\mathcal{S}=\CH$ and $\CH$ is of static type and the following properties are true:  \begin{enumerate}
				
				\item For all $u < u_{\CH}$, $$\lim_{v \rightarrow +\infty}|\nu|(u,v)=0.$$
				\item  $r-r_-(e,M)$, $\phi$, $D_u \phi$, $\varpi-M$, $Q-e$ extend continuously to $0$ on $\CH=\{ u < u_{\CH}, v=+\infty\}$.
				\item There exists $v_0=v_0(M,e,u_0)$ such that the estimates are true \eqref{rigidity1}, \eqref{rigidity2}, \eqref{rigidity3}, \eqref{rigidity4}, \eqref{rigidity5}, \eqref{rigidity6}, \eqref{rigidity7}, \eqref{rigidity9}, \eqref{rigidity10}, \eqref{rigidity11}, \eqref{rigidity12} on $[u_s,u_T] \times [v_T,+\infty)$ for $C=C(M,e,q_0,m^2,u_T)>1$. \end{enumerate}
		\end{enumerate}
		
	\end{cor}
	
	\begin{proof} We start by the case $\mathcal{S}_0 =  \emptyset$: then the staticity condition is violated everywhere: in particular, for all $u \leq u_s$: $$ \int_{v}^{+\infty} \kappa(u,v') dv' < +\infty.$$
		
		Then, quite similarly to what was done in the proof of Lemma \ref{<us}, we can use the estimates of Proposition \ref{LBprop} to obtain \eqref{dvnu}, which implies that $r \nu(u,v)$ has a limit $r \nu_{\CH}(u)$ as  $v\rightarrow+\infty$ for all $u \leq u_s$ and $r(u,v)$ has a limit $r_{\CH}(u)>0$ as  $v\rightarrow+\infty$ for all $u \leq u_s$ so, following the logic of the proof of Lemma \ref{<us}, since now the Dafermos condition is satisfied, it must be that for all $u \leq u_s$, $ \nu_{\CH}(u)<0$ thus $\CH$ is of dynamical type, following Definition \ref{dynamicaldef}.
		
		Now, if $\mathcal{S}_0 \neq  \emptyset$ and $u_D^{max}:=\sup_{u <u_{\infty}(\CH)} u_D(u) <u_{\infty}(\CH)$: then, quite similarly to what was done in the proof of Corollary \ref{supcorollary}, we find that there exists $\epsilon=\epsilon(u_D^{max})>0$ such that for all $u < u_D^{max}$, the estimates of Lemma \ref{propagation} are true on $[u-\frac{\epsilon}{2} , \min\{u_D^{max}, u+\frac{\epsilon}{2}\}]$ and those of Lemma \ref{finitevolumeestlemma} are true on  on $[u-\frac{\epsilon}{2} , , u+\frac{\epsilon}{2}]$, as a consequence of Theorem \ref{classificationtheorem}. This provides a proof of statements \ref{a}, \ref{b}, \ref{c}, \ref{d} and \ref{e}. This also shows that $\CH$ is a Cauchy horizon of mixed type, following Definition \ref{mixeddef}.
		
		Lastly, if  $\mathcal{S}_0 \neq  \emptyset$ and $u_D^{max}:=\sup_{u <u_{\infty}(\CH)} u_D(u) =u_{\infty}(\CH)$ then the assumptions of Theorem \ref{classificationtheorem} are satisfied for all $u_0 <u_{\infty}(\CH)$ so, clearly, $\CH$ is a Cauchy horizon of static type, following Definition \ref{rigiddef} and $\mathcal{S}_0 = \CH$.

	\end{proof}

	\section{Local blow-up of the mass for dynamical and mixed type $\CH$} \label{localsection}
	
	In this section, we prove that the Hawking mass blows up for sufficiently late retarded-time $u$ on Cauchy horizon of dynamical type (section \ref{sub1}) or mixed type (section \ref{sub2}), exploiting the classification of Corollary \ref{threetypes}.

	\subsection{Local blow-up of the mass for dynamical type Cauchy horizons} \label{sub1}
	
	We start with the dynamical case $\mathcal{S}_0=\emptyset$:
	
	\begin{lem} \label{localmassblowupdynamical}
		Assume that $\mathcal{S}_0=\emptyset$, thus $\CH$ is of dynamical type by Corollary \ref{threetypes}. Then, for all $u\leq u_s$ we have $$ \lim_{v \rightarrow +\infty} \rho(u,v)=+\infty.$$
		Moreover for all $u_1<u_s$, there exists $C(M,e,q_0,m^2,u_1,u_s)>0$ such that the following lower bound is true in for all $  u_1\leq u\leq u_s$, $v \geq v_{\gamma}(u)$: \begin{equation} \label{masslowerbound}
			\rho(u,v) \geq C \cdot  e^{ 1.98 |K_-|v}
		\end{equation} 
	\end{lem}
	
	\begin{proof} Recall that the entire region is trapped so for $u\leq u_s$, $v \geq v_{\gamma}(u)$, $\lambda(u,v)<0$ and that both $r$ and $\nu$ extends continuously to $\CH \cap \{ u \leq u_s \}$.

		Since $\CH$ is of dynamical type, there exists $\eta(u_1,u_s)>0$ such that for all $u_1\leq u\leq u'_s$, $v \geq v_{\gamma}(u)$, $|\nu|(u,v)>\eta$. There exists also $r_0>0$ such that for all $u\leq u_s$, $v \geq v_{\gamma}(u)$, $r_0<r(u,v)<r_+(M,e)$. Since $\rho=\frac{r}{2}(1+\frac{|\lambda| |\nu|}{\Omega^2})$, it implies that for all $u_1 \leq u\leq u_s$, $v \geq v_{\gamma}(u)$ $$ \rho(u,v) \geq \frac{  \eta \cdot r_0}{2} \cdot \frac{|\lambda|(u,v) }{\Omega^2(u,v)}.$$
		
		Now we can integrate \eqref{partialvOmegaLB} on $u_1 \leq u\leq u_s$, $v \geq v_{\gamma}(u)$ to obtain 
		$$	v^{-1}|\log(\Omega^2)-2K_- \cdot v |  \leq D \cdot  v^{1-2s}.$$
		From this estimate, we get that for some $D'>0$ $$ \rho(u,v) \geq  D' \cdot   e^{ 2|K_-|v \cdot (1-D \cdot v^{1-2s})}  \cdot |\lambda|(u,v).$$
		
		Recall that from \eqref{lambdaus} we get for some $D''>0$:  \begin{equation} \label{lambdaL1lowerbound}
			\int_v^{+\infty}|\lambda|(u,v) \geq D'' \cdot  v^{-p}.
		\end{equation}This implies that there exists a $\alpha$-adic sequence $v_{n}= \alpha^{n-1} v_1$ for $\alpha=1.0001$ such that for all $\max\{u_{\gamma}(v_n),u_1\} \leq  u\leq u_s$: \begin{equation} \label{lambdadyadiclowerbound}
			|\lambda|(u,v_n) \gtrsim v_n^{-p-1}.
		\end{equation} Combining this with the previous lower bound on $\rho$ we get that for all $\max\{u_{\gamma}(v_n),u_1\}\leq  u\leq u_s$ \begin{equation} \label{rhodyadiclowerbound}
			\rho(u,v_n) \geq D' \cdot D''  \cdot  e^{ 2|K_-|v_n \cdot (1-C \cdot v_n^{1-2s})}  \cdot v_n^{-p-1}.
		\end{equation} Now, we use \eqref{massVEinstein} together with \eqref{QphiLB} to get $$ \partial_v \varpi \gtrsim -v^{-2s} \cdot ( 1+ v^{2-2s}) -v^{-s} \cdot ( 1+ v^{3-3s}),$$ which we integrate on $[v_n,v]$ for $v \in [v_n,v_{n+1}]$, using the lower bound on $\rho(u,v_n)$, the formula $\varpi=\rho+\frac{Q^2}{2r}$ and \eqref{QphiLB}:  $$ \varpi(u,v) \gtrsim \varpi(u,v_n) -1- v_n^{3-4s} -v_n^{1-s} -v_n^{4-4s} \gtrsim  \rho(u,v_n)- v_n^{4-4s} -1- v_n^{3-4s} -v_n^{1-s} -v_n^{4-4s}  \gtrsim  e^{ 2|K_-|v_n \cdot (1-C \cdot v_n^{1-2s})}  \cdot v_n^{-p},$$ and now we use that for $v_1$ large enough, $2|K_-|v_n \cdot (1-C \cdot v_n^{1-2s}) \geq 1.999 |K_-|v$ and since $e^{ 2|K_-| x \cdot (1-C \cdot x^{1-2s})}  \cdot x^{-p-1}$ is increasing for $x$ large enough, we get $$ \varpi(u,v) \gtrsim  e^{ 1.998|K_-|v}.$$
		
		Now, using again \eqref{QphiLB}, we get $ \rho(u,v) \gtrsim  e^{ 1.998|K_-|v}-1-v^{2-2s}$ which implies \eqref{masslowerbound} for $v_1$ large enough.
	\end{proof}
	\subsection{Local blow-up of the mass for mixed type Cauchy horizons} \label{sub2}
	Now we turn to the mixed case $\mathcal{S}_0 \neq \emptyset$, $\mathcal{S}_0 \neq \CH$. The proof is similar to the dynamical case.
	\begin{lem} \label{localmassblowupmixed}
		Assume that $\mathcal{S}_0 \neq  \emptyset$ and $u_T:=\sup_{u <u_{\infty}(\CH)} u_D(u) <u_{\infty}(\CH)$; thus $\CH$ is of mixed type, by Corollary \ref{threetypes}. Then, for all $u_T< u\leq u_T+\epsilon_T$, where $\epsilon_T$ was introduced in the statement of Corollary \ref{threetypes}, we have $$ \lim_{v \rightarrow +\infty} \rho(u,v)=+\infty.$$
		Moreover, for all $u_T<u_1<  u_T+\epsilon_T$, there exists $C(M,e,q_0,m^2,u_1)>0$ such that the following lower bound is true in for all $u_1 \leq u\leq u_T+\epsilon_T$, $v \geq v_{T}$: \begin{equation} \label{masslowerboundmixed}
			\rho(u,v) \geq C \cdot e^{ 1.98 |K_-|v}.
		\end{equation}
	\end{lem}
	
	\begin{proof}
		Since $\mathcal{S}_0=(-\infty,u_T]$, this implies, by Corollary \ref{threetypes} that for all $u_T<u_1<u_2<u_{\CH}$ there exists $\eta(u_1,u_2,v_T)>0$ such that for all $v \geq v_T$, $u_1 \leq u\leq u_2$, $|\nu|(u,v) \geq \eta$. 
		
		From there, it is easy to reproduce the proof of Lemma \ref{localmassblowupdynamical}, with \eqref{rigidity14}, \eqref{rigidity15}, \eqref{rigidity16}, \eqref{rigidity17} playing the role of \eqref{partialvOmegaLB}, \eqref{QphiLB} and \eqref{lambdaus}. \eqref{masslowerboundmixed} follows immediately.
	\end{proof}
	\section{Propagation of the mass blow-up, proof of Theorems \ref{trappedtheorem} and \ref{propagationtheorem}}
	
	Once we know that the Hawking mass blows up (locally) for Cauchy horizon of dynamical or mixed type (section \ref{localsection}), we have to prove that this blow up is propagated (section \ref{propagationsection}). We will then use this result to prove Theorems \ref{trappedtheorem} (section \ref{blowupcriterion}) and Theorem \ref{propagationtheorem} (section \ref{trappedprop}).
	
	\subsection{The blow-up of the Hawking mass}  \label{propagationsection}
	We start by Lemma \ref{massblowuplemma}, our result proving the propagation of the blow up on the Hawking mass $\rho$ on $\CH$. This statement is quite general: we do not, in fact, require the assumption of Theorem \ref{previous} to obtain its conclusions, nor the formalism of the classification of the Cauchy horizon of Corollary \ref{threetypes}. This is why Lemma \ref{massblowuplemma} is also used independently in \cite{r=0}.
	
	\begin{lem} \label{massblowuplemma}
		If there exists $u_1 < u_{\infty}(\CH)$ such that $$ \lim_{v \rightarrow +\infty} \rho(u_1,v)=+\infty,$$ and $v_1 \in \RR$ large enough, a constant $D>0$ such that for all $v \geq v_1$: \begin{equation} \label{phiinitblowup}
			|\phi|^2(u_1,v)+ |Q|(u_1,v) \leq D \cdot |\log(\rho)|(u_1,v).
		\end{equation} 
		then, for all $ u_1 \leq u_2 < u_{\infty}(\CH)$, $$ \lim_{v \rightarrow +\infty} \rho(u_2,v)=+\infty.$$
		
		Moreover, there exists $v_1'=v_1'(u_1,u_2) \geq v_1$ such that $[u_1,u_2] \times [v_1',+\infty) \subset \mathcal{T}$ and we have the following estimates for all $(u,v) \in [u_1,u_2] \times [v_1',+\infty)$ for some $C'(D,M,e,q_0,m^2,u_1,u_2)>0$, some  $0<\alpha<\frac{1}{2}$:
		
		\begin{equation} \label{phiQblowupimproved}
			|\phi|^2(u,v)+ |Q|(u,v) \leq C' \cdot |\log(\rho)|(u,v), 
		\end{equation} 
		\begin{equation} \label{massblowupest}
			\rho(u,v) \geq \rho(u_1,v) \cdot (1-C' \cdot \rho^{2\alpha-1}(u_1,v)).
		\end{equation}

	\end{lem}
	\begin{rmk}
		Note that the assumption \eqref{phiinitblowup} is satisfied under the assumptions of Corollary \ref{threetypes} , as $|\phi|^2(u,v)+ |Q| \lesssim v^{2-2s} \lesssim v$. In fact, assumption \eqref{phiQblowupimproved} can be considerably weakened to  $|\phi|^2(u,v)+ |Q| = O( \rho \cdot \log(\rho)^{-1})$.
	\end{rmk}
	
	\begin{proof} 
		
		First, we recall that we have the bounds $r_0(u_1,u_2,v_1) \leq r(u,v) \leq r_+(M,e)$ for all $(u,v) \in [u_1,u_2] \times [v_1,+\infty)$, for some $r_0(u_1,u_2,v_1)>0$. 
		Then, since $\rho(u_1,v) \rightarrow +\infty$ as $v \rightarrow +\infty$, $2\rho(u_1,v)> 2r_+$ for $v$ large enough hence there exists $v_1$ such that $\{u_1\} \times [v_1,+\infty) \subset \T$. By \eqref{RaychV}, this implies that $\iota^{-1}(u_1,v) \geq \eta_0>0$ for all $v\geq v_1$, defining $\eta_0:=\iota^{-1}(u_1,v_1)$.

		For some $0<\alpha<\frac{1}{2}$ and $\eta_0>\eta>0$, we bootstrap for some $C>0$ to be chosen later: \begin{equation} \label{B1mblowup}
			|\phi|^2+ Q^2 \leq C \cdot \rho^{2\alpha},
		\end{equation} 
		\begin{equation} \label{B2mblowup}
			\iota^{-1}  \geq \eta.
		\end{equation}
		
		For $C>0$ large enough, it is clear that  \eqref{B1mblowup}  and \eqref{B2mblowup} are satisfied already on $\{u_1\} \times [v_1,+\infty]$. Then, using \eqref{massUEinstein} together with bootstrap \eqref{B1mblowup}, we have for some $C'(C,M,e)>0$, 
		$$ \partial_u \rho \geq \frac{r^2}{2} \cdot \iota^{-1} |D_u \phi|^2 -C' \cdot \rho^{2\alpha} \cdot |\nu|  \geq  -C' \cdot \rho^{2\alpha}  \cdot |\nu|,$$ where for the last lower bound, we just used $\iota^{-1} \geq0$, as a soft consequence of \eqref{B2mblowup}. Since $0<\alpha<\frac{1}{2}$, it is clear that $$ \partial_u( \rho^{1-2\alpha})(u,v) \geq -(1-2\alpha) \cdot C' \cdot |\nu|(u,v).$$
		
		Thus, integrating, it is clear  that for all $u_1<u_2$ such that the bootstraps are satisfied on $[u_1,u_2] \times [v_1,+\infty)$: $$ \rho^{1-2\alpha}(u_2,v) \geq  \rho^{1-2\alpha}(u_1,v)+C'  \cdot (1-2\alpha)\cdot (r(u_2,v)-r(u_1,v)).$$
		
		From this we obtain the blow up of the mass and the estimate \eqref{massblowupest}. This estimate implies that there exists $v'_1>v_1$ such that for all $u_1 \leq u\leq u_2$, $v\geq v'_1$, $2\rho(u,v)>r_+$ thus $(u,v) \in \T$. Therefore, by \eqref{RaychV}, that $\iota^{-1}(u,v) \geq \eta_1>0$ for all $v\geq v_1'$ and $u_1 \leq u\leq u_2$, defining $\eta_1:=\sup_{u \in [u_1,u_2]}\iota^{-1}(u,v_1')$. Thus we retrieve bootstrap \eqref{B2mblowup} if $0<\eta<\eta_1$.

		Now we need to retrieve bootstrap \eqref{B1mblowup}. For this, consider \eqref{massUEinstein} and write, under bootstrap \eqref{B1mblowup} and \eqref{B2mblowup} $$   \frac{r^2 \cdot \iota^{-1}}{2} \cdot |D_u \phi|^2(u,v) \leq C' \cdot \rho^{2 \alpha}(u,v)+ \partial_u \rho(u,v),$$ which is also equivalent, using \eqref{murelation} to $$   \frac{ r|D_u \phi|^2}{2|\nu|}(u,v) \leq \frac{ C'  \cdot \rho^{2 \alpha}(u,v)}{2\rho(u,v)-r(u,v)}+ \frac{\partial_u \rho(u,v)}{2\rho(u,v)-r(u,v)} \leq C'  \cdot \rho^{-1+2 \alpha}(u,v)+ \partial_u \log(\rho)(u,v),$$ where we have used $2\rho(u,v) -r(u,v) \geq \rho(u,v)$ on $[u_1,u_2] \times [v_1,+\infty]$ for $v_1$ large enough, since $\rho$ tends to $+\infty$ by \eqref{massblowupest}.
		
		Thus, we get, integrating, also using \eqref{massblowupest}: $$ \int_{u_1}^{u_2} \frac{ r|D_u \phi|^2}{2|\nu|}(u,v) du \lesssim   \rho^{-1+2 \alpha}(u_1,v)+  \log( \frac{\rho(u_2,v)}{\rho(u_1,v)}).$$
		
		We can integrate in $u$ this estimate, using Cauchy-Schwarz as $$ |\phi(u_2,v)- e^{-iq_0 \int_{u_1}^{u_2} A_u(u',v) du'}\phi(u_1,v)| \leq \int_{u_1}^{u_2} |D_u \phi|(u',v) du' \leq (\int_{u_1}^{u_2} \frac{ r|D_u \phi|^2}{2|\nu|}(u,v) du)^{\frac{1}{2}} (\int_{u_1}^{u_2} \frac{ 2|\nu|}{r}(u,v) du)^{\frac{1}{2}},$$ which gives, using the former estimate $$ |\phi|(u_2,v)\lesssim |\phi|(u_1,v)+   (\log( \frac{\rho(u_2,v)}{\rho(u_1,v)}) )^{\frac{1}{2}} \lesssim \log( \rho(u_2,v))^{\frac{1}{2}},$$ where we used \eqref{phiinitblowup} and \eqref{massblowupest}, which is already sufficient to retrieve the $|\phi|^2$ part of the bootstrap \eqref{B1mblowup}.
		
		Then, notice using \eqref{chargeUEinstein} that $$ |\partial_u Q|(u,v) \lesssim |\phi|(u,v)|D_u \phi|(u,v),$$ so we can integrate, use Cauchy-Schwarz and the previous bounds to obtain for all $u_1 \leq u \leq u_2$:
		$$ |Q(u,v)| \lesssim |Q(u_1,v)|+ \log( \rho(u_2,v) \lesssim \log( \rho(u_2,v)) $$ where we used \eqref{phiinitblowup}: this retrieves bootstrap \eqref{B1mblowup}.


	\end{proof}

	\subsection{Proof of Theorem \ref{propagationtheorem}} \label{blowupcriterion}
	
	In this section, we return to the main proof and we work again under the assumptions of Corollary \ref{threetypes} (i.e.\ the assumptions of Theorem \ref{previous}). We will use Lemma \ref{propagation} and Corollary \ref{threetypes} to obtain a proof a Theorem \ref{propagationtheorem} in the following proposition:
	
	\begin{prop} \label{propprop}
		Assume that \underline{one} of the following conditions holds for some $u_0 <u_{\CH}$: \begin{equation} \label{condition1}
			\limsup_{v \rightarrow+\infty} \rho(u_0,v)=+\infty,
		\end{equation}  \begin{equation} \label{condition2}
			\limsup_{v \rightarrow+\infty} \varpi(u_0,v)=+\infty,
		\end{equation} \begin{equation} \label{condition3}
			\int_{v_0}^{+\infty} \kappa(u_0,v) <+\infty,
		\end{equation} \begin{equation} \label{condition4}
			\limsup_{v \rightarrow+\infty} |\nu|(u_0,v) >0,
		\end{equation}\begin{equation} \label{condition5}
			\limsup_{v \rightarrow+\infty} |\phi|^2(u_0,v)+ \limsup_{v \rightarrow+\infty} |Q|(u_0,v)=+\infty.
		\end{equation}
		Then $\CH$ is either of dynamical type, or of mixed type, with $u_T<u_0$. In any case, we have for all $u \geq u_0$ : \begin{equation} \label{result1}
			\lim_{v \rightarrow+\infty} \rho(u,v)=+\infty,
		\end{equation}  \begin{equation} \label{result2}
			\lim_{v \rightarrow+\infty} \varpi(u,v)=+\infty,
		\end{equation} \begin{equation} \label{result3}
			\int_v^{+\infty} \kappa(u,v')dv'<+\infty.
		\end{equation} 
	\end{prop}
	\begin{proof}
		Clearly, if $\CH$ was of static type, then by Corollary \ref{threetypes}, none of condition \eqref{condition1}, \eqref{condition2}, \eqref{condition3}, \eqref{condition4}, \eqref{condition5} are possible. If $\CH$ is of mixed type, then again by Corollary \ref{threetypes}, it means that $u_0>u_T$, otherwise none of condition \eqref{condition1}, \eqref{condition2}, \eqref{condition3}, \eqref{condition4}, \eqref{condition5} are possible. Thus, $[u_0,u_{\CH}] \subset \CH-\mathcal{S}_0$ hence \eqref{result3} is satisfied for all $u \geq u_0$.

		To prove \eqref{result1}, we start by the case that $\CH$ is of dynamical type: then, using Lemma \ref{localmassblowupdynamical}, we see that we have the blow up $\lim_{v \rightarrow+\infty} \rho(u_s,v)=+\infty$ and by \eqref{QphiLB}, \eqref{phiinitblowup} is satisfied: thus we can apply Lemma \ref{massblowuplemma} thus \eqref{result1} and a forciori \eqref{result2} follow.	If now $\CH$ is of mixed type: then, using Lemma \ref{localmassblowupmixed}, then we obtain the same result.

	\end{proof}
	\subsection{Proof of Theorem \ref{trappedtheorem}} 
	
	Now, as a second consequence of Lemma \ref{propagation}, we obtain the existence of a trapped neighbohood of $\CH$:
	\begin{prop} \label{trappedprop}
		For all $u < u_{\CH}$, there exists $v(u) \in \RR$ such that $\{u\} \times [v(u),+\infty) \subset \T$.

	\end{prop}
	
	\begin{proof}
		From the estimates of \cite{Moi}, we know that  $\{ (u,v), v \geq v_{\gamma}(u),-\infty <u \leq u_{s}   \} \subset \mathcal{T}$. Now there are three cases: 
		
		\begin{enumerate}[a]
			\item $\CH$ is of static type. \label{trappedrigid}
			
			Then $\mathcal{S}=(-\infty,u_{\CH})$ and the claim is given statement \ref{class1} by Theorem \ref{classificationtheorem}.

			\item $\CH$ is of dynamical type. \label{trappeddyn}
			
			Then $\mathcal{S}=\emptyset$. By Lemma \ref{localmassblowupdynamical}, for all $u\leq u_s$, $v\geq v_{\gamma}(u)$, \eqref{masslowerbound} is true. Thus, also using \eqref{QphiLB}, the assumptions of Lemma \ref{massblowuplemma} are satisfied for $u_1=u_s$. Recall that $(u,v) \in \T$ if and only if $2\rho(u,v)>r(u,v)$.  Since $r$ is bounded, then for all $u<u_{\CH}$, there exists $v(u)$ such that $[u_s,u] \times [v(u),+\infty) \subset \mathcal{T}$.
			
			\item $\CH$ is of mixed type.
			
			Then  $\mathcal{S}=(-\infty,u_T]$ and by Lemma \eqref{localmassblowupmixed} there exists $\epsilon>0$ such that \eqref{masslowerboundmixed} is true at $u=u_T+\epsilon$, together with \eqref{rigidity15}. Thus, the assumptions of Lemma \ref{massblowuplemma} are satisfied for $u_1=u_T+\epsilon$, which gives the result, as in case \ref{trappeddyn}.

		\end{enumerate} 
		
	\end{proof}

	\section{Quantitative estimates and proof of Theorem \ref{classification}} \label{quantitative}
	
	Now we establish quantitative estimates on an arbitrarily late portion of $\CH$. The key ingredient in our proof is the presence of a trapped neighborhood of $\CH$ obtained in Theorem \ref{trappedtheorem}, whose volume is thus finite. While such estimates are not strictly speaking necessary to prove $C^2$-inextendibility (but we \textit{choose} to prove inextendibility using them), they bridge the gap between the (preliminary) classification of Theorem \ref{classificationtheorem} and the final classification of Theorem \ref{classification}.
	\subsection{$L^1$ and $L^{\infty}$ estimates for a trapped rectangle}
	
	In this section, we prove $L^1$ estimates on a trapped \footnote{In particular $R$ is of finite space-time volume, but as we shall see, the stronger assumption $R \subset \mathcal{T}$ turns out to be necessary to our method.} rectangle $R$  which includes a portion of $\CH$. These estimates are in the vein of those of Lemma \ref{finitevolumeestlemma} but we introduce a considerable difficulty, to accommodate the case of a Cauchy horizon of dynamical type \footnote{Recall that, in retrospect, Lemma \ref{finitevolumeestlemma} was concerned by Cauchy horizon of mixed type: in this case, $\phi$ and $Q$ were necessarily bounded.}: both $\phi$ and $Q$ are allowed to blow up, at a polynomial rate. While this seems anecdotal, we cannot close the estimates of Lemma \ref{finitevolumeestlemma} because we start with weaker assumptions \eqref{phiLinfty2hyp}, \eqref{QLinfty2hyp} and we must get contend \footnote{The lack of integrability when $s<1$, even if  $\phi$ is bounded, was a new difficulty in Lemma \ref{finitevolumeestlemma} already, compared to Lemma 10.3 of \cite{JonathanStab}.} with even weaker estimates. Allowing for this possibility is not necessary to prove Theorem \ref{conditionnalSCCtheoremtwoended} and Theorem \ref{CHinexttheorem}, as we already know that the mass blows up by Lemma \ref{massblowuplemma} but is crucial \footnote{This is because we require estimates, even weak ones, just to prove that $\nu$ has a non-zero limit on $\CH$.} to proving Theorem \ref{classification}. 
	\begin{lem} \label{finitevolumeestlemma2} Let $u_1 < \uu < u_{\infty}(\CH)$,  in particular $\uu <+\infty$ and $v_0\geq 1$. Assume that the rectangle $R:= [u_1,\uu] \times [v_0,+\infty]$ is trapped: $R \subset \mathcal{T}$. By Lemma \ref{trappedvolume}, $vol(R)<\infty$. We also assume that for some $D>0$ and some $\frac{3}{4}<s<1$, $\alpha(M,e)>0$, the following estimates are true on the past boundary of $R$:  \begin{equation} \label{dulogomegaL12hyp}
			\int_{u_1}^{u_{max}}|\partial_u \log(\Omega^2)|(u,v_0)du  \leq D,
		\end{equation} 	\begin{equation} \label{duphiL12hyp}
			\int_{u_1}^{u_{max}} |D_u \phi|(u,v_0)du \leq D,
		\end{equation} 	\begin{equation} \label{durL12hyp}
			\int_{u_1}^{u_{max}} |\nu|(u,v_0)du  \leq D,
		\end{equation} 	\begin{equation} \label{dvrL12hyp}
			\sup_{v \in  [v_0,+\infty)} v^{2s} \cdot |\lambda|(u_1,v) dv \leq D,
		\end{equation} 	\begin{equation} \label{phiLinfty2hyp}
			\sup_{v \in  [v_0,+\infty)} v^{-1+s} \cdot|\phi|(u_1,v) \leq D,
		\end{equation}  	\begin{equation} \label{QLinfty2hyp}
			\sup_{v \in  [v_0,+\infty)} v^{-2+2s} \cdot |Q|(u_1,v) \leq D,
		\end{equation} \begin{equation} \label{OmegaLinfty2hyp}
			\sup_{v \in  [v_0,+\infty)} e^{\alpha \cdot v} \cdot \Omega^2(u_1,v) \leq D.
		\end{equation} 
		Then, for some $C(M,e,q_0,m^2,u_1,\uu,v_0,vol(R),D,s)>0$, $v_1=v_1(M,e,q_0,m^2,u_1,\uu,vol(R),D,s)>v_0$, the following estimates are true on the smaller rectangle $[u_1,\uu] \times [v_1,+\infty)$, defining $\psi:=r\phi$: \begin{equation} \label{dulogomegaL12}
			\int_{u_1}^{\uu}\sup_{v \in [v_1,+\infty)} v^{-2+2s} \cdot |\partial_u \log(\Omega^2)|(u,v)du  \leq C,
		\end{equation} 	\begin{equation} \label{dupsiL12}
			\int_{u_1}^{\uu} \sup_{v \in [v_1,+\infty)} |D_u \psi|(u,v)du  \leq C ,	\end{equation} 	\begin{equation} \label{duphiL12}
			\int_{u_1}^{\uu} \sup_{v \in [v_1,+\infty)} v^{-1+s} \color{black} \cdot |D_u \phi|(u,v)du  \leq C ,	\end{equation} 	\begin{equation} \label{durL12}
			\int_{u_1}^{\uu} \sup_{v \in [v_1,+\infty)} |\nu|(u,v)du  \leq C ,
		\end{equation} 	\begin{equation} \label{dvrL12}
			\sup_{(u,v) \in [u_1,\uu] \times [v_1,+\infty)} v^{2s} \cdot |\lambda|(u,v) dv \leq C,
		\end{equation} 
		\begin{equation} \label{phiLinfty2}
			\sup_{(u,v) \in [u_1,\uu] \times [v_1,+\infty)} v^{-1+s} \cdot|\phi|(u,v) \leq C,
		\end{equation} 
		\begin{equation} \label{dvphiLinfty2}
			\sup_{(u,v) \in [u_1,\uu] \times [v_1,+\infty)} v^{s} \cdot|\partial_v \phi|(u,v) \leq C,
		\end{equation}    
		\begin{equation} \label{QLinfty2}
			\sup_{(u,v) \in [u_1,\uu] \times  [v_1,+\infty)} v^{-2+2s} \cdot|Q|(u,v), \leq C.	\end{equation}   \begin{equation} \label{OmegaLinfty2}
			\sup_{(u,v) \in [u_1,\uu] \times  [v_1,+\infty)} e^{0.99 \alpha \cdot  v}\cdot \Omega^2(u,v) \leq C.	\end{equation}  
	\end{lem}
	\begin{proof} First, we introduce the notation $A \lesssim B$, which means that there exists a constant $C(M,e,q_0,m^2,u_1,\uu,vol(R),D,s)>0$ such that $A \leq C \cdot B$. Before anything else, we start with an a priori estimate on $\lambda$: taking advantage of the fact that we are in the trapped region, we use \eqref{Radius3} together with the monotonicity induced by \eqref{RaychV}: there exists a constant $C_0(D,M,e,u_1,\uu)>0$ such that $$ \partial_u(r|\lambda|)(u,v)\leq \frac{\Omega^2(u,v)}{4} \leq \frac{\Omega^2}{4|\lambda|}(u,v_0) \cdot |\lambda|(u,v) \leq C_0 \cdot r |\lambda|(u,v) ,$$
		which allows us to apply a Gr\"{o}mwall estimate and obtain, using \eqref{dvrL12hyp}, that for all $(u,v) \in R$: \begin{equation} \label{lambdaglobal}
			|\lambda|(u,v) \lesssim v^{-2s}.
		\end{equation} This gives directly \eqref{dvrL12}. Using again the estimate $\Omega^2(u,v) \leq \frac{\Omega^2}{|\lambda|}(u,v_0) \cdot |\lambda|(u,v) $, we also get \begin{equation} \label{Omegaglobal}
			\Omega^2(u,v) \lesssim v^{-2s}.
		\end{equation} 
		
		Then, we reduce the size of the rectangle: we define $R'=[u_1,\uu] \times [v_1,+\infty)$ for some $v_1>v_0$. By local well-possessedness on the rectangle $[u_1,\uu] \times [v_0,v_1]$ , $u \rightarrow \partial_u \log(\Omega^2)(u,v_1)$, $u \rightarrow \nu(u,v_1)$ and $u \rightarrow D_u \phi(u,v_1)$ are bounded functions so there exists $C_1=C_1(u_1,\uu,v_0,v_1,D)>0$ such that $$ \sup_{u_1 \leq u \leq \uu}|\partial_u \log(\Omega^2)|(u,v_1)+ |\nu|(u,v_1)+|D_u \phi|(u,v_1) \leq C_1.$$
		
		Then, we cut $R'$ into small rectangles $R_i=[u_i,u_{i+1}] \times [v_1,+\infty]$, $u_N=\uu$, $u_{i+1}-u_i=\epsilon$, $\epsilon>0$ and $N=\frac{\uu-u_1}{\epsilon}$. Using Lemma \ref{trappedlemma}, we see immediately that $vol(R_i) \lesssim \epsilon$.

		Thus, we have the following initial smallness estimates for the $L^1$ norms on $[u_1,\uu] \times \{v_1\}$: \begin{equation} \label{dulogomegaL12hypsmall}
			\int_{u_i}^{u_{i+1}}|\partial_u \log(\Omega^2)|(u,v_1)du  \lesssim \epsilon,
		\end{equation} 	\begin{equation} \label{duphiL12hypsmall}
			\int_{u_i}^{u_{i+1}} |D_u \phi|(u,v_1)du \lesssim \epsilon,
		\end{equation} 	\begin{equation} \label{durL12hypsmall}
			\int_{u_i}^{u_{i+1}} |\nu|(u,v_1)du  \lesssim \epsilon.
		\end{equation}

		The proof of the result is a finite induction: we will prove the induction hypothesis: there exists 
		
		$C=C(M,e,q_0,m^2,u_1,\uu,vol(R),D,s)>0$, $C'=C'(M,e,q_0,m^2,u_1,\uu,vol(R),D,s)>0$,

		$C''=C''(M,e,q_0,m^2,u_1,\uu,vol(R),D,s)>0$ such that for all $(u,v) \in R_i$ \begin{equation} \label{hyp1}
			|\phi|(u,v) \leq \phi_i \cdot v^{1-s},
		\end{equation} 
		\begin{equation} \label{hyp2}
			|Q|(u',v) \leq  Q_i \cdot v^{2-2s},
		\end{equation}  \begin{equation} \label{hyp3}
			|\partial_v \phi|(u',v) \leq  \Phi_i \cdot v^{-s},
		\end{equation} where we defined $\phi_i:= D +C  \cdot \color{black}\epsilon \cdot i $, $\Phi_i:=  D +C'  \cdot \color{black}\epsilon \cdot i $ and $Q_i:= D +C''  \cdot \color{black} \epsilon \cdot i $.
		
		Using \eqref{QLinfty2hyp}, \eqref{phiLinfty2hyp}, \eqref{dvphiLinfty2}, we see those estimates are initially satisfied on $\{u_1\} \times [v_1,+\infty]$, so the initialization is true. Notice that $1 \leq i \leq N:=\frac{u_{max}-u_1}{\epsilon}$, an estimate we will often use.
		
		Notice that for any $0<C_0 \lesssim 1$, we have the estimate, which we use implicitly several times in the argument: $$ D+ C_0 \epsilon  \cdot i  \leq D+ C_0 \cdot  \frac{\uu-u_1}{\epsilon} \cdot  \epsilon  \lesssim 1.$$

		Assume that \eqref{hyp1}, \eqref{hyp2}, \eqref{hyp3} are true on $R_i$ and we will prove them on $R_{i+1}$.

		We bootstrap the following estimates on $R_{i+1}$: \begin{equation} \label{boostrap1}
			|\phi|(u,v) \leq 4 \phi_{i+1} \cdot  v^{1-s},
		\end{equation} \begin{equation} \label{boostrap2}
			|Q|(u,v) \leq 4 Q_{i+1} \cdot  v^{2-2s},
		\end{equation} \begin{equation} \label{boostrap3}
			|\partial_v \phi|(u,v) \leq 4 \Phi_{i+1} \cdot  v^{-s}.
		\end{equation} 
		By \eqref{hyp1}, \eqref{hyp2}, \eqref{hyp3}, we see that the bootstraps are initially satisfied on $\{u_{i+1}\} \times [v_1,+\infty]$.

		Then, taking advantage of the fact that we are in the trapped region, we use \eqref{Radius3} as $$\partial_v(r|\nu|) \leq \frac{\Omega^2}{4}.$$
		
		Then, integrating on $R_i$ and using \eqref{durL12hypsmall}, we get, very similarly \footnote{In particular we do not need to use the bootstraps \eqref{boostrap1}, \eqref{boostrap2} or \eqref{boostrap3}, which makes the estimates relatively easy.} to the proof of Lemma \ref{finitevolumeestlemma}: 	\begin{equation} \label{estimate1}
			\int_{u_i}^{u_{i+1}}\color{black} \sup_{v \in [v_1,+\infty)} |\nu|(u',v)du' dv \lesssim \epsilon.
		\end{equation} 	
		
		Then, using \eqref{Field3}\color{black} with bootstraps \eqref{boostrap1}, \eqref{boostrap2} and \eqref{boostrap3} gives, also using \eqref{Omegaglobal}, for all $u\in[u_i,u_{i+1}]$: $$ |\partial_v(rD_u \phi)|(u,v) \lesssim  |\nu|(u,v) \cdot v^{-s} +  v^{3-5s},$$ which then implies, integrating first in $v$ and using that $3-4s<0$ (recall that $s<1$, see Remark \ref{s<1remark}): $$ |D_u \phi|(u,v) \lesssim |D_u \phi|(u,v_1)  + v^{1-s} \cdot (	1+\sup_{ v' \in [v_1,+\infty)}|\nu|(u,v')),$$
		an estimate we can integrate in $u$, using \eqref{duphiL12hypsmall} and \eqref{estimate1}: for $C=C(M,e,q_0,m^2,u_1,\uu,vol(R),D,s)>0$ independent of $i$ we get \begin{equation} \label{Duphiest}
			\int_{u_i}^{u_{i+1}}	\sup_{ v \in [v_1,+\infty)} v^{-1+s} \color{black}|D_u \phi|(u,v) du \leq C \cdot \epsilon,
		\end{equation} and this gives us an estimate for $\phi$, using \eqref{hyp1}: \begin{equation}  \label{phi}
			|\phi|(u,v) \leq (\phi_i + C\cdot \epsilon) \cdot v^{1-s} \leq \phi_{i+1} \cdot v^{1-s},
		\end{equation} which closes bootstrap \eqref{boostrap1} and also proves the first induction hypothesis \eqref{hyp1} on $R_{i+1}$.
		
		Then, we can integrate \eqref{chargeUEinstein}, using \eqref{Duphiest} and \eqref{phi} and obtain, using also \eqref{hyp2}: for all $(u,v) \in R_{i+1}$: \begin{equation} \label{Q}
			|Q|(u,v) \leq (Q_i+C''\cdot \epsilon) \cdot v^{2-2s} = Q_{i+1} \cdot v^{2-2s},
		\end{equation} for $C''=C''(M,e,q_0,m^2,u_1,\uu,vol(R),D,s)>0$ independent of $i$. 
		
		This closes bootstrap \eqref{boostrap2} and also proves the second induction hypothesis \eqref{hyp2} on $R_{i+1}$.
		
		Now we turn to \eqref{Field}\color{black}, which we write, using \eqref{lambdaglobal}, \eqref{Omegaglobal}, \eqref{phi}, \eqref{Q} and bootstrap \eqref{boostrap3} as: $$ |D_u(\partial_v \phi)| \lesssim v^{-s} \cdot \sup_{ v\in [v_1,+\infty)} |\nu|(u,v)+ v^{1-3s} \cdot \sup_{ v\in [v_1,+\infty)}  v^{-1+s} \color{black} |D_u \phi|(u,v) + v^{3-5s}.$$ We can integrate this equation in $u$, using the fact that $v^{1-3s} \leq v^{-s}$ and $v^{3-5s} \leq v^{-s}$, with \eqref{Duphiest}, \eqref{estimate1} to get \begin{equation}
			|\partial_v \phi|(u,v)\leq \Phi_i \cdot v^{-s} + C' \cdot v^{-s}=\Phi_{i+1} \cdot v^{-s},
		\end{equation} for $C'=C'(M,e,q_0,m^2,u_1,\uu,vol(R),D,s)>0$ independent of $i$ and for all $(u,v) \in R_{i+1}$. This closes bootstrap \eqref{boostrap3} and also proves the third induction hypothesis \eqref{hyp3}.

		Therefore, we finished proving the induction and we immediately obtain  \eqref{dupsiL12}, \eqref{duphiL12}, \eqref{durL12}, \eqref{dvrL12}, \eqref{phiLinfty2}, \eqref{dvphiLinfty2}, \eqref{QLinfty2} on the entire rectangle $R'$.
		
		Now we turn to the proof of \eqref{dulogomegaL12} and \eqref{OmegaLinfty2}. First, we are going to bootstrap the following estimate: for some $\Delta>0$ to be determined, \begin{equation} \label{bootstrap4}
			\Omega^2 \leq \Delta \cdot  e^{-\frac{\alpha}{2} v}.
		\end{equation}
		
		Now, we use \eqref{Omega} together with \eqref{lambdaglobal}, \eqref{dvphiLinfty2} to get $$ |\partial_u \partial_v \log(\Omega^2)|(u,v) \lesssim v^{1-2s} \cdot \sup_{ v \in [v_1,+\infty)} v^{-1+s} |D_u \phi|(u,v) + \Delta \cdot v^{2-2s}  \cdot  e^{-\frac{\alpha}{2} v}+v^{-2s} \cdot \sup_{ v \in [v_1,+\infty)}|\nu|(u,v),$$ and we can integrate in $v$ to obtain  $$ |\partial_u \log(\Omega^2)|(u,v) \lesssim |\partial_u \log(\Omega^2)|(u,v_1)+ v^{2-2s} \cdot \sup_{ v \in [v_1,+\infty)} v^{-1+s} |D_u \phi|(u,v) + \Delta +  \sup_{ v \in [v_1,+\infty)}|\nu|(u,v).$$
		
		Now we can integrate in $u$ on $[u_1,\uu]$ using \eqref{durL12}, \eqref{dulogomegaL12hyp}, \eqref{duphiL12} to get \begin{equation} \int_{u_1}^{\uu} \sup_{ v \in [v_1,+\infty)}
			|\partial_u  \log(\Omega^2)|(u,v) \lesssim v^{2-2s} + \Delta.
		\end{equation}
		Integrating and using \eqref{OmegaLinfty2hyp}, this implies that, for a constant $E=E(M,e,q_0,m^2,u_1,\uu,v_0,vol(R),D,s)>0$ and for all $(u,v) \in R'$: $$ \Omega^2(u,v) \leq D \cdot e^{-\alpha v + E \cdot ( v^{2-2s} + \Delta)}.$$
		
		Then, chose $\Delta=2D$ and $v_1$ large enough so that $E \cdot ( v_1^{1-2s} + 2D \cdot v_1^{-1}) < \frac{\alpha}{100}$. 
		
		This closes bootstrap \eqref{bootstrap4} and also proves  \eqref{dulogomegaL12} and \eqref{OmegaLinfty2}.
		
		Now we want to derive an estimate for the ingoing derivative of $\psi$: we write, using \eqref{Field3}, and also \eqref{Radius3}: \begin{equation} \label{Field4}
			\partial_v( D_u \psi )= -\frac{\Omega^2}{4r} \phi \cdot \left[1-\frac{Q^2}{r^2} +iq_0 Q \phi\right] -\frac{\lambda \nu}{r} \phi,
		\end{equation} and now, we can write the estimate, using \eqref{OmegaLinfty2}, \eqref{phiLinfty2hyp}, \eqref{QLinfty2}, \eqref{lambdaglobal}: $$ |\partial_v( D_u \psi )| \lesssim e^{-0.99 \alpha \cdot v} v^{3-3s}+ v^{1-3s} \cdot \sup_{ v \in [v_1,+\infty)}|\nu|(u,v),$$ which eventually gives \eqref{dupsiL12}, after integrating in $u$ and $v$, using \eqref{durL12} and the fact that $s>\frac{3}{4}>\frac{2}{3}$.
		
	\end{proof}
	
	\begin{rmk}
		We claim that the proof of Lemma \ref{finitevolumeestlemma2}, together with a bit of algebra (see \cite{Moi3Christoph} for the full proof ), implies that for all $u_1<u_2<u_{\CH}$ such that $[u_1,u_2]\times [v_1,+\infty)$ is trapped, then \begin{equation} \label{MoiChristoph2}
			\int_{u_1}^{u_{2}} \sup_{v\in[v_1,+\infty)}  |\partial_u \left(\log(\Omega^2)+ |\phi|^2(u,v) +\int_{u_1}^{u}\frac{\nu}{r} |\phi|^2(u',v) du' \right)|du \leq C,
		\end{equation} for some $C=C(M,e,q_0,m^2,u_1,\uu,v_1,D)>0$. This estimate plays a crucial role in the inextendibility argument of \cite{Moi3Christoph}.
	\end{rmk}

	\begin{rmk}
		A slight modification of the argument allows to prove Lemma \ref{finitevolumeestlemma2} in the case $\frac{1}{2} < s \leq  \frac{3}{4}$. We shall not pursue this route, for simplicity and because in practice we expect $s >  \frac{3}{4}$ in any cases of physical interest.  \color{black}
	\end{rmk}

	\subsection{Concluding the proof of Theorem \ref{classification}}

	Now, we are ready to establish the quantitative estimates of Theorem \ref{classification}: \begin{prop}
		For all $u_1 <u_2<u_{\CH}$, there exists $C(M,e,q_0,m^2,u_1,u_2,s)>0$ such that for all  $u_1  \leq u \leq u_2$ and $v \geq v(u)$, the estimates \eqref{trappedestimate1}, \eqref{trappedestimate2}, \eqref{trappedestimate3}, \eqref{trappedestimate4}, \eqref{trappedestimate5}, \eqref{trappedestimate6}, \eqref{trappedestimate7}, \eqref{trappedestimate8}, \eqref{trappedestimate9}, \eqref{trappedestimate11}, \eqref{trappedestimate10} hold.
	\end{prop} \begin{proof}
		For some $u_1 < u_2<u_{\CH}$, and defining $v_s:= v_{\gamma}(u_s)$, we start with data on $[u_1,u_2] \times \{v_s\} \cup \{u_1\} \times [v_s,+\infty)$. Using the estimates of section \ref{LB} and a standard well-posedness result in the interior of the space-time, we get that  \eqref{dulogomegaL12hyp}, \eqref{duphiL12hyp}, \eqref{durL12hyp}, \eqref{dvrL12hyp}, \eqref{phiLinfty2hyp}, \eqref{QLinfty2hyp}, \eqref{OmegaLinfty2hyp} are satisfied for some $D=D(M,e,q_0,m^2,u_1,u_2)>0$ and we apply Lemma \ref{finitevolumeestlemma2} on the corresponding rectangle. This gives immediately \eqref{trappedestimate1}, \eqref{trappedestimate2}, \eqref{trappedestimate3}, \eqref{trappedestimate4}, \eqref{trappedestimate5}, \eqref{trappedestimate6}, \eqref{trappedestimate7}, \eqref{trappedestimate8}.
		
		Additionally, we also have, for all $v \geq v_s$:  $$ |\partial_v \log(\Omega^2)(u_s,v) -2K_-(M,e) | \leq D \cdot v^{1-2s}.$$ For $s=1$ and $s>1$, the same argument gives respectively
		$O(v^{-1}(1+\log v))$ and $O(v^{-s})$, using the corresponding
		bounds for $\int |D_u\phi|\,du$ and the estimates of
		Section~\ref{LB}. Hence, integrating \eqref{Omega} in $u$, we get, for all $u_1 \leq u \leq u_2$: \begin{equation*} \begin{split}
				|\partial_v \log(\Omega^2)(u,v) -2K_-(M,e) | \leq D \cdot v^{1-2s} + 2 \sup_{u' \in [u_s,u]} |\partial_v \phi|(u',v) \cdot \int_{u_s}^u |D_u \phi|(u',v)du'+ (u-u_{s}) \cdot \sup_{u' \in [u_s,u]} \Omega^2 \cdot ( \frac{1}{2r^2}+  \frac{Q^2}{r^4})(u',v) \\+ 2 \sup_{u' \in [u_s,u]} \frac{|\lambda|}{r^2}(u',v) \cdot \int_{u_s}^u |\nu|(u',v) du' \lesssim    D \cdot v^{1-2s}+  C^2 \cdot v^{1-2s}+ C^2 \cdot (u-u_s) \cdot e^{1.99 K_- v}+ C^2 \cdot  v^{-2s} \lesssim v^{1-2s}.
		\end{split}\end{equation*} Hence, possibly for a different constant $D$ still depending on $u_1$ and $u_2$, \eqref{trappedestimate9} and \eqref{trappedestimate11}  hold.
		
		Then, recall that $V = \Omega^{-2} \partial_v$, $Ric(V,V) = \Omega^{-4} |\partial_v \phi|^2$. Integrating \eqref{Radius3} with the help estimate \eqref{dulogomegaL12}, \eqref{phiLinfty2}, \eqref{QLinfty2}, we obtain for some $C>0$: $$ |r\lambda(u,v)-r\lambda(u_s,v)| \lesssim e^{ C \cdot v^{2-2s}}  \Omega^2(u_s,v),$$ which also implies, making use of the fact that $r$ is lower bounded: $$| \frac{|\lambda|(u,v)}{\Omega^2(u,v)}- \frac{r|\lambda|(u_s,v)}{r\Omega^2(u,v)}| \lesssim e^{2C \cdot v^{2-2s}}.$$ Then, we use \eqref{lambdadyadiclowerbound} for a $\alpha$-adic sequence $v_{n}= \alpha^{n-1} v_1$, $\alpha=1.0001$, together with \eqref{OmegaLinfty2} to get  $$ \frac{|\lambda|(u,v_n)}{\Omega^2(u,v_n)}\gtrsim v_n^{-1-p} \cdot e^{1.99 |K_-| v_n} -e^{2C \cdot v_n^{2-2s}} \gtrsim e^{1.98 |K_-| v_n}.$$
		
		Then, using the monotonicity of \eqref{RaychV}, we can immediately say that for all $v\geq v_0$, \begin{equation}
			\frac{|\lambda|(u,v)}{\Omega^2(u,v)} \gtrsim e^{1.97\color{black} |K_-| v }.
		\end{equation}
		Then, using \eqref{RaychV}, we see that this implies $$ \int_{v(u)}^{v} \frac{ |\partial_v \phi|^2}{\Omega^2}(u_p,v')dv' \gtrsim e^{1.97 |K_-| v }.$$
		
		Then, using the lower bound $\Omega^{-2}(u,v) \gtrsim \Omega^{-2}(u,v(u)) \gtrsim 1$, we immediately get \eqref{trappedestimate10} as $$  \int_{v(u)}^{v}Ric(V,V)(u,v')dv' = \int_{v(u)}^{v} \frac{ |\partial_v \phi|^2}{\Omega^4}(u,v')dv' \gtrsim  \int_{v(u)}^{v}\frac{ |\partial_v \phi|^2}{\Omega^2}(u,v')dv' \gtrsim e^{1.97\color{black} |K_-| v }.$$

	\end{proof}
	
	As a corollary of the estimates, we obtain a reinforcement of Theorem \ref{propagationtheorem}, which completes the proof of Theorem \ref{classification}:
	
	\begin{cor}
		Assume that \underline{one} of the following conditions holds for some $u_0 <u_{\CH}$: \begin{equation} \label{condition1a}
			\limsup_{v \rightarrow+\infty} \rho(u_0,v)=+\infty,
		\end{equation}  \begin{equation} \label{condition2a}
			\limsup_{v \rightarrow+\infty} \varpi(u_0,v)=+\infty,
		\end{equation} \begin{equation} \label{conditiona3}
			\int_{v_0}^{+\infty} \kappa(u_0,v) <+\infty,
		\end{equation} \begin{equation} \label{condition4a}
			\limsup_{v \rightarrow+\infty} |\nu|(u_0,v) >0,
		\end{equation}\begin{equation} \label{condition5a}
			\limsup_{v \rightarrow+\infty} |\phi|^2(u_0,v)+ \limsup_{v \rightarrow+\infty} |Q|(u_0,v)=+\infty.
		\end{equation}
		Then, in addition to the results of Proposition \ref{propprop}, we also have that for all $u \geq u_0$:  \begin{equation} \label{result3a}
			\lim_{v \rightarrow+\infty} |\nu|(u,v) >0.
		\end{equation}  
	\end{cor}
	
	\begin{proof} By contradiction, assume that there exists $u \geq u_0$ such that \begin{equation} \label{contrad}
			\liminf_{v \rightarrow+\infty} |\nu|(u,v) =0.
		\end{equation}  In view of \eqref{Radius3}, \eqref{trappedestimate9}, \eqref{trappedestimate8}, \eqref{trappedestimate6}, $ v \rightarrow r |\nu|(u,v)$ is integrable and moreover,  integrating as we did before, using \eqref{trappedestimate11}:  $$ |\nu|(u,v) \lesssim \Omega^2(u,v),$$
		hence $\kappa^{-1} \lesssim 1$, hence $v \rightarrow \rho(u,v)$ is bounded, also using \eqref{trappedestimate4}. This contradicts the mass blow up of Proposition \ref{propprop}.
	\end{proof}
	
	To conclude the proof of Theorem \ref{classification}, notice that \eqref{conditiona3} is satisfied for all $u < u_{\CH}$ if $\CH$ is of dynamical type, hence \eqref{result3a} is true for all $u < u_{\CH}$. If $\CH$ is of mixed type, then \eqref{conditiona3} is satisfied for all $u>u_T$ then \eqref{result3a} is true for all $u > u_T$.
	
	\section{Global inextendibility across the Cauchy horizon and proof of \\ Theorem \ref{conditionnalSCCtheoremtwoended}, Theorem \ref{CHinexttheorem} and Theorem \ref{conditionnalSCCtheoremoneended}} \label{inextproofsection}
	
	In this section, we present a geometric proof of $C^2$-future-inextendibility under various assumptions. We only require very soft geometric arguments, as the quantitative estimates were already obtained in section \ref{quantitative}. While these quantitative estimates are very different from their uncharged counterparts, in the absence of certain simplifying mechanisms such as monotonicity, the geometric argument is extremely similar to what was used in \cite{JonathanStab} to prove $C^2$-future-inextendibility in the uncharged case. In this section, we rely on the blow up of the Ricci curvature given by \eqref{trappedestimate10} to obtain $C^2$-future-inextendibility and for this purpose, we adapt marginally the argument of Luk and Oh from \cite{JonathanStab} to our setting.

	\subsection{$C^2$-inextendibility in the two-ended case}
	
	We start with the two-ended case. First, we need a result from \cite{JonathanStab}: \begin{prop}[\cite{JonathanStab}, section 11] \label{geomextendJonathan}
		
		We consider $(M,g_{\mu \nu}, \phi,F)$ the maximal development of smooth, spherically symmetric and admissible \textbf{two-ended} initial data.	Assume that $(M,g)$ is $C^2$-future-extendible in the sense of Definition \ref{generalinext}.
		
		Then \footnote{Strictly speaking, the result in \cite{JonathanStab} is Lemma 11.5, which presents an alternative between two possibilities, one of them of the statement of Proposition \ref{geomextendJonathan}, and the other one is proven to be impossible in Lemma 11.6, using soft estimates present in \cite{Kommemi}.}, there exists $p \in \mathcal{CH}_{i^+_1} \cup \mathcal{CH}_{i^+_2}$ and a future-directed null geodesic $\gamma:(-\epsilon,\epsilon) \rightarrow \tilde{M}$ such that  \begin{enumerate}
			\item $\gamma(-\epsilon,0) \subset M$.
			\item $\gamma_{|(-\epsilon,0)}$ is radial.
			\item  $\Pi \circ \gamma_{|(-\epsilon,0)} \subset \{ (u_p,v), v \geq v_0\}$ for $u_p<u_{\CH}$, such that $p=(u_p,+\infty)$ in $(u,v)$ coordinates and $v_0 \in \RR$.
			\item For all $-\epsilon < t <0$, we have $\frac{d}{dt}(\Pi \circ \gamma(t))=c  \cdot \Omega^{-2}(\Pi \circ \gamma(t)) \partial_v$ for some $c>0$.
			
			\item  \label{Riccibounded1} $t \rightarrow Ric(\dot{\gamma}(t),\dot{\gamma}(t))= c^2 \cdot \Omega^{-4} |\partial_v \phi|^2$ is bounded on $(-\epsilon,0)$.
		\end{enumerate}

	\end{prop}
	While Proposition \ref{geomextendJonathan} is proved in the context of an uncharged and massless scalar fields in \cite{JonathanStab}, it does not rely on the precise matter model, and only assumes spherical symmetry and the a priori boundary characterization in the two-ended case of Theorem \ref{twoendedapriori}, which is also valid for the charged and/or massive scalar field model, c.f. \cite{Kommemi}. Therefore, Proposition \ref{geomextendJonathan} holds in our context, with the same proof.
	
	\begin{prop} \label{Ricciblowupprop} Under the same assumptions as Proposition \ref{geomextendJonathan}, we have the following blow up of the Ricci curvature: $$ \limsup_{v \rightarrow +\infty} Ric(\Omega^{-2} \partial_v, \Omega^{-2} \partial_v)(u_p,v)=+\infty.$$
		This contradicts statement \ref{Riccibounded1} of Proposition \ref{geomextendJonathan}, thus by contradiction $(M,g)$ is $C^2$-future-inextendible.
	\end{prop}
	\begin{proof} The proof follows immediately from \eqref{trappedestimate10}: by the pigeon-hole principle, there exists a dyadic sequence $v_n = 2^n \cdot v(u)$ such that $Ric(V,V)(u,v_n) \geq C \cdot v_n^{-1} \cdot  e^{1.98 |K_-| \cdot v_n}$ for all $u < u_{\CH}$; so in particular at $(u_p,+\infty) \in \CH$, hence the $\limsup$ is infinite, as claimed, and the contradiction follows.

	\end{proof}
	
	Proposition \ref{Ricciblowupprop} provides a proof of Theorem \ref{conditionnalSCCtheoremtwoended}.
	\subsection{$C^2$-inextendibility across the Cauchy horizon in the one-ended case} \label{oneendedext1}
	
	As we have seen, the two-ended case follows immediately from \eqref{trappedestimate10}, as we used directly the geometric setting of \cite{JonathanStab}. In the one-ended case, we have to re-prove some of the claims of \cite{JonathanStab}, but the approach is essentially similar and, as before, the key ingredient is the blow up of the Ricci curvature given by \eqref{trappedestimate10}. Before going to the proof, we recall a Lemma from Dafermos--Rendall \cite{DafermosRendall} concerning $C^2$ extensions, which appears also as Lemma 11.2 in \cite{JonathanStab}: \begin{lem}[\cite{DafermosRendall}, Lemma 11.2 in \cite{JonathanStab}] \label{DafermosRendallLemma}If $\tilde{M}$ is a $C^2$ extension, the standard rotations extend continuously to $\partial M$. $\partial M$ is the boundary of $M$ inside a $C^2$ extension $\tilde{M}$.\end{lem}As a result, we obtain immediately the following corollary on the projection on the quotiented space-time: 

	\begin{cor} \label{Piextension}
		If $\tilde{M}$ is a $C^2$ extension, the map $\Pi \circ i^{-1}: i(M) \color{black} \rightarrow \mathcal{Q}^+$ extends as a continuous map 	$\tilde{\Pi}: i(M) \cup\check{(\partial M)}  \rightarrow \overline{\mathcal{Q}^+}$, where $\check{(\partial M)}$ is a subset of $\partial M$ of full measure with the property that for all $p\in \check{(\partial M)}$,  $\partial M$ is differentiable at $p$. 
	\end{cor}\begin{proof} 
		Recall that $\partial M$ is the topological boundary of $i(M)$ in $\tilde{M}$ and $\overline{\mathcal{Q}^+}$ is the closure in $\RR^{2}$ of the bounded set $\mathcal{Q}^+$. By Lemma 11.1 in \cite{JonathanStab}, $\partial M$ is a locally lipschitz achronal hypersurface in $\tilde{M}$ thus by Rademacher theorem, it is differentiable almost every-where thus $T_p (\partial M)$ exists for almost every $p \in \partial M$ (an argument first given in Lemma 11.5 in \cite{JonathanStab}). We denote  $\check{(\partial M)} \subset \partial M$  the set of all such $p$: this is a set of full measure in $\partial M$.
		
		We will thus construct the extension $\tilde{\Pi}(p)$ for $p\in \partial M$ such that $T_p(\partial M)$ is well-defined.

		Consider two sequences $(p_n)_{n \in \mathbb{N}} \in i(M)$ and  $(q_n)_{n \in \mathbb{N}} \in i(M)$ which converge to $p$ in $\tilde{M}$. Since $\overline{\mathcal{Q}^+}$ is compact, we can assume (up to sequence extraction) that both $\Pi(p_n)$ and $\Pi(q_n)$ converge in  $\overline{\mathcal{Q}^+} \subset \RR^2$, with no loss of generality. We will denote $\bar{p}  \in \overline{\mathcal{Q}^+}$ and $\bar{q}  \in \overline{\mathcal{Q}^+}$ their respective limit. We want to prove that $\bar{p}=\bar{q}$. First, we establish some preliminaries.

		By Lemma \ref{DafermosRendallLemma} we know that the Killing rotations $\{ O_1, O_2, O_3\}$ (here $O_i$, vector fields in $i(M)$, are the push-forwards of the Killing rotations in $M$ by the extension map $i$) extend to continuous vector fields $\{ \Oo, \Ot, \Oth\}$ on $i(M) \cup \partial M$.

		In view of the formula $r^2= \frac{1}{2} \sum_{i=1}^3 g(O_i,O_i)$, it is also clear that the area-radius function $r$ extends continuously (abusing notation we will still call $r$ this extension) to $\partial M$ (this was already proven in Lemma 11.3 in \cite{JonathanStab}). 
		
		We claim that $r(p)>0$. Suppose not i.e.\ suppose that $r(p)=0$. Then, by the proof of Lemma 11.5 in \cite{JonathanStab}, there exists a future oriented time-like geodesic $\gamma:(-\epsilon,\epsilon) \rightarrow \tilde{M}$ such that $\gamma(0)=p$ for some $\epsilon>0$. Since $\tilde{M}$ is a $C^2$ manifold, the Kretschmann scalar $R_{\alpha \beta \mu \nu}R^{\alpha \beta \mu \nu}$ must be bounded on $\gamma$. But this contradicts the fact that  $R_{\alpha \beta \mu \nu}R^{\alpha \beta \mu \nu}$ blows up at $\{r=0\}$ (see the proof of Lemma 11.6 in \cite{JonathanStab} for details). Therefore we proved that $r(p)>0$.
		
		Since $r(p)>0$, then by Lemma 11.4 in \cite{JonathanStab}, $\{ \Oo(p), \Ot(p), \Oth(p)\}$ generate a two-dimensional space-like vector sub-space of $T_p \tilde{M}$ (with respect to the metric $\tilde{g}$), that we denote $\tilde{\mathcal{G}}$. We also denote $\tilde{\mathcal{G}}^{\perp}$ the orthogonal of $\tilde{\mathcal{G}}$ with respect to $\tilde{g}$. Since $\tilde{\mathcal{G}}^{\perp}$ is a time-like plane, it contains two non-proportional past directed vector fields that we denote $L$ and $\barL$. Recall that $T_p(\partial M)$ is well-defined: we claim that either $L \notin T_p(\partial M)$ or  $\barL \notin T_p(\partial M)$. Indeed, if both $L \in T_p(\partial M)$ and  $\barL \in T_p(\partial M)$, then $L+\barL \in T_p(\partial M)$ which contradicts the achronality of $\partial M$ (an argument we borrowed from \cite{JonathanStab}). Without loss of generality, assume that $L \notin T_p(\partial M)$.

		With these preliminaries being proven, denote $\bar{p}=(u_{\bar{p}},v_{\bar{p}}) \in \overline{\mathcal{Q}^+} \subset \RR^{2}$, $\bar{q}=(u_{\bar{q}},v_{\bar{q}}) \in  \overline{\mathcal{Q}^+}\subset \RR^{2}$ and assume by contradiction that $\bar{p} \neq \bar{q}$. Without loss of generality, assume that the space-time is one-ended, thus by Theorem \ref{oneendedapriori},  $\overline{\mathcal{Q}^+}$ is given by Penrose diagram of Figure \ref{Fig1} and that $\bar{p} \in \CH$.
		
		We first claim that $v_{\bar{p}}=v_{\bar{q}}$; suppose not: then this implies that $\bar{q} \notin \CH \cup \mathcal{S}_{i^+} $. Since $r(\bar{q})>0$, then by elimination it implies that $\bar{q} \in \mathcal{CH}_{\Gamma}$. Therefore $J^{-}_{\overline{\mathcal{Q}^+}}(\bar{p}) \cap J^{-}_{\overline{\mathcal{Q}^+}}(\bar{q})$ (the intersection of the causal past of $\bar{q}$ and the causal past of $\bar{p}$ in $\overline{\mathcal{Q}^+}$) is a compact subset of the open set $\mathcal{Q}^+$. Since $\mathbb{S}^2$ is a compact set, the projection $\Pi \circ i^{-1} : i(M) \rightarrow \mathcal{Q}^+$ is a proper map, therefore $(\Pi \circ i^{-1})^{-1}(J^{-}_{\overline{\mathcal{Q}^+}}(\bar{p}) \cap J^{-}_{\overline{\mathcal{Q}^+}}(\bar{q}))$ is a compact subset of $i(M) \subset \tilde{M}$. But $p  \in(\Pi \circ i^{-1})^{-1}(J^{-}_{\overline{\mathcal{Q}^+}}(\bar{p}) \cap J^{-}_{\overline{\mathcal{Q}^+}}(\bar{q}))$ (since it belongs to its closure in $\tilde{M}$, and that it is compact set in $i(M)$) and since $p\in \partial M$, any sequence converging to $p$ must leave every compact set in $i(M)$: we reach a contradiction. Thus we proved indeed that $v_{\bar{p}}=v_{\bar{q}}$.
		
		
		It remains to prove that  $u_{\bar{p}}=u_{\bar{q}}$. For this, define $\gamma: (-\epsilon,\epsilon) \rightarrow \tilde{M}$ the unique past directed geodesic such that $\gamma(0)=p$ and $\gamma'(0)=L$.  By the proof of Lemma 11.5 in \cite{JonathanStab}, and since $L\notin T_p( \partial M)$, we know that $\gamma(0,\epsilon) \subset i(M)$ and that $\gamma_{|(0,\epsilon)}$ is radial. Thus, there exists $u_0 \in \RR$ such that $\Pi \circ \gamma_{|(0,\epsilon)} \subset \{ (u_0,v), v\in \mathcal{V}^+(u_0) \} \subset \mathcal{Q}^+$, where $\mathcal{V}^+(u):=\{ v\in \RR/ (u,v) \in \mathcal{Q}^+\}$ is a bounded subset of $\RR$. Now there are three possibilities: $u_p < u_0$, or $u_p >u_0$ or $u_p=u_0$. If $u_p>u_0$, then there an open neighborhood $A$ of $\Pi \circ \gamma(0,\epsilon)$ in $\mathcal{Q}^+$ and $N \in \mathbb{N}$ such that $\bigcup_{n \geq N} I^+_{ \mathcal{Q}^+}(p_n) \cap A = \emptyset$ (here $I^+_{ \mathcal{Q}^+}(p_n)$ denotes the time-like future of $p_n$ in $\mathcal{Q}^+$ and is an open set in  $\mathcal{Q}^+$). Note that the pre-image $I:=(\Pi\circ i^{-1})^{-1}(\bigcup_{n \geq N} I^+_{ \mathcal{Q}^+}(p_n) )$ is an open (as $\Pi\circ i^{-1}$ is continuous on $i(M)$) subset of $i(M)$ and that $p \in \bar{I}$, denoting $\bar{I}$, the closure of $I$ in $i(M)$; note that $ \bar{I} \cap i(M) = I$. Since $p \in \overline{ \gamma(0,\epsilon)}$ as well, there exists a neighborhood $\tilde{V}$ of $p$ in $\tilde{M}$ such that for all open set $ \tilde{U} \subset \tilde{V}$, we have $\tilde{U} \cap \bar{I} \cap i(M)= \tilde{U} \cap I \neq \emptyset$ and $\tilde{U} \cap \overline{ \gamma(0,\epsilon)} \cap i(M)= \tilde{U} \cap  \gamma[0,\epsilon] \neq \emptyset$. Denote $U:= \tilde{U}  \cap i(M)$: this implies that $\Pi\circ i^{-1}(U) \cap \Pi \circ i^{-1}(I) \neq \emptyset$ and $\Pi\circ i^{-1}(U) \cap \Pi \circ i^{-1} \circ \gamma(0,\epsilon) \neq \emptyset$. Now choose $\tilde{U}$ small enough so that $\Pi \circ i^{-1} (U) \subset A$ (this is possible since $A$ is an open set containing $\gamma(0,\epsilon)$). Taking the intersection with $\Pi \circ i^{-1}(I) $, we obtain a contradiction since we have constructed $A$ such that $A \cap \Pi \circ i^{-1}(I) = \emptyset$. If $u_p<u_0$, a similar argument holds replacing the future $I^+(p_n)$ by the past $I^-(p_n) $. Therefore, we proved that $u_{\bar{p}}=u_{\bar{q}}=u_0$.
		
		Thus, we proved that $\bar{p}=\bar{q}$. For all $p \in \breve{(\partial M)}$, we define the extension of $\Pi$ as $\bar{\Pi}(p):=\bar{p}$. By construction, this map is continuous for the topology induced $i(M) \cup \breve{(\partial M)}$. While the proof was done in the one-ended case for the sake of exposition, note that we can repeat the same arguments in the two-ended case with no additional work.

	\end{proof} 
	
	In what follows, we will identify the manifolds $i(M)$ and $M$, since they are isometric, therefore we also identify $\Pi$ with $\Pi \circ i^{-1}$, thus we view $\tilde{\Pi}$ as an extension of $\Pi$. \color{black}Now we turn to the proof of the main proposition:

	\begin{prop} \label{geomextendoneened}
		We consider $(M,g_{\mu \nu}, \phi,F)$ the maximal development of smooth, spherically symmetric and admissible \textbf{one-ended} initial data with a priori boundary $\mathcal{B^+}$ given by the decomposition of Theorem \ref{oneendedapriori}.
		
		Assume that $(M,g)$ is $C^2$-future-extendible across $\CH$ in the sense of Definition \ref{CHinext}; we call the extension $(\tilde{M},\tilde{g})$.
		
		Then, there exists $p \in \CH$ and a future-directed null geodesic $\gamma_1:(-\epsilon,\epsilon) \rightarrow \tilde{M}$ such that  \begin{enumerate}
			\item $\gamma_1(-\epsilon,0) \subset M$.
			\item $\gamma_{1}$ is radial on $(-\epsilon,0)$.
			\item  $\Pi \circ \gamma_{1}(-\epsilon,0) \subset \{ (u_p,v), v \geq v_0\}$ for $u_p<u_{\CH}$, such that $p=(u_p,+\infty)$ in $(u,v)$ coordinates and $v_0 \in \RR$.
			\item For all $-\epsilon < t <0$, we have $\frac{d}{dt}(\Pi \circ \gamma_1(t))=c  \cdot \Omega^{-2}(\Pi \circ \gamma_1(t)) \partial_v$ for some $c>0$.
			\item  \label{Riccibounded2} $t \rightarrow Ric(\dot{\gamma_1}(t),\dot{\gamma_1}(t))= c^2 \cdot \Omega^{-4} |\partial_v \phi|^2$ is bounded on $(-\epsilon,0)$.
		\end{enumerate} 
	\end{prop}
	
	\begin{proof}  	We mostly follow the strategy of \cite{JonathanStab}, section 11, which we accommodate to our setting. \\ First, by Definition \ref{CHinext}, there exists $p \in \check{(\partial M)}$ (i.e.\ such that $T_p (\partial M)$ exists, see the notations of Corollary \ref{Piextension}) and a continuous curve $\gamma:(-\epsilon,\epsilon) \rightarrow \tilde{M}$ with $\gamma(0)=p$, $\gamma(-\epsilon,0) \subset M$ and $ \overline{\Pi(\gamma(-\epsilon,0))}^{\mathcal{Q}^+} \cap \CH \neq \emptyset.$ This implies that there exists $p_{\CH} \in \CH \subset \mathcal{B}^+$ and a sequence $t_n \rightarrow 0$ as $n \rightarrow +\infty$ such that $\Pi( \gamma(t_n)) \rightarrow p_{\CH}$ as $n \rightarrow +\infty$.  By the continuity of the map $\tilde{\Pi}$ from Corollary \ref{Piextension}, this implies that $\tilde{\Pi}(p)=p_{\CH}$.
		
		Now, using the achronality of $\partial M$ (Lemma 11.1 in \cite{JonathanStab}), the extension of the Killing rotations to $\partial M$ (Lemma 11.2 in \cite{JonathanStab})	
		and the fact that $T_p (\partial M)$ exists by definition, one can construct, using the same method \footnote{The proof of \cite{JonathanStab}, which does not use the equations, can be exactly reproduced in our setting.} as in \cite{JonathanStab}, a null geodesic $\gamma_1:(-\epsilon,\epsilon) \rightarrow \tilde{M}$ such that  \begin{enumerate}
			\item $\gamma_1(0)=p$.
			\item $\gamma_1(-\epsilon,0) \subset M$.
			\item $\gamma_{1}$ is radial on $(-\epsilon,0)$.
			\item  $\Pi \circ \gamma_{1}(-\epsilon,0) \subset \{ (u_0,v), v \geq v_0\}$ for some $u_0<u_{\CH}$ and $v_0 \in \RR$.
			\item For all $-\epsilon < t <0$, we have $\frac{d}{dt}(\Pi \circ \gamma_1(t))=c  \cdot \Omega^{-2}(\Pi \circ \gamma_1(t)) \partial_v$,
		\end{enumerate} the last point using the fact that  $\Omega^{-2} \partial_v$ is a radial geodesic vector field.

		By continuity of $\tilde{\Pi}$ again, we obtain \color{black} $$ \lim_{ t\rightarrow 0} \Pi(\gamma_1(t))= \tilde{\Pi}(p)=p_{\CH}.$$ This also implies that $u_0=u_p$ i.e. that $\Pi \circ \gamma_{1}(-\epsilon,0) \subset \{ (u_p,v), v \geq v_0\}$. Since $\gamma_1$ is a geodesic and that $\tilde{M}$ is a $C^2$ manifold, we immediately obtain the boundedness of $t \rightarrow Ric(\dot{\gamma_1}(t),\dot{\gamma_1}(t))$ which concludes the proof.

	\end{proof}

	\begin{prop} \label{inextproponeended} Under the same assumptions as Proposition \ref{geomextendoneened}, we have the following blow up of the Ricci curvature: $$ \limsup_{v \rightarrow +\infty} Ric(\Omega^{-2} \partial_v, \Omega^{-2} \partial_v)(u_p,v)=+\infty.$$
		This contradicts statement \ref{Riccibounded2} of Proposition \ref{geomextendoneened} so by contradiction $(M,g)$ is $C^2$-future-inextendible across $\CH$.
	\end{prop}
	\begin{proof} We can repeat exactly the same proof that we used for Proposition \eqref{Ricciblowupprop}.
	\end{proof} 
	
	Proposition \ref{inextproponeended} provides a proof of Theorem \ref{CHinexttheorem}.
	\subsection{$C^2$-inextendibility if $\mathcal{CH}_{\Gamma}=\emptyset$ in the one-ended case}
	Now, we turn to the $C^2$-future-inextendibility of the space-time in the one-ended case, if we assume additionally that the ``Cauchy horizon emanating from the center'' $\mathcal{CH}_{\Gamma}$ is empty, an assumption which is conjectured to be generic (see Conjecture \ref{trappedsurfaceconj}, and related to Weak Cosmic Censorship, see section \ref{connected}. The proof is similar to that of section \ref{oneendedext1}.
	
	\begin{prop} \label{oneendedCHgammaprop}
		We consider $(M,g_{\mu \nu}, \phi,F)$ the maximal development of smooth, spherically symmetric and admissible \textbf{one-ended} initial data with a priori boundary $\mathcal{B^+}$ given by the decomposition of Theorem \ref{oneendedapriori}.
		
		Assume that $\mathcal{CH}_{\Gamma}=\emptyset$ and moreover that $(M,g)$ is $C^2$-future-extendible.
		
		Then there exists $p \in \partial M$, with $r(p)=0$ and a future time-like geodesic $\gamma:(-\epsilon,\epsilon) \rightarrow \tilde{M}$ such that $\gamma(0)=p$ and $\gamma((-\epsilon,0)) \subset M$.
	\end{prop}
	\begin{proof}
		Since $\partial M$ is Liptschitz (Lemma 11.1 in \cite{JonathanStab}), by the Rademacher theorem it is almost everywhere differentiable, so one can find $p \in \partial M$ at which $\partial M$ is differentiable (c.f. \cite{JonathanStab}). Since $r$ extends continuously to $\partial M$ (Lemma 11.3 in \cite{JonathanStab}), we either have $r(p) \neq 0$ or $r(p)=0$. If $r(p) \neq 0$, then by Theorem \ref{oneendedapriori}, $p \in \CH$ since $\mathcal{CH}_{\Gamma}=\emptyset$. In this case, one obtains a contradiction using the same argument as in Proposition \ref{geomextendoneened} and Proposition \ref{inextproponeended}. So we can assume for now that $r(p)=0$. Then, we can repeat the argument of Lemma 11.5 of \cite{JonathanStab}, which yields the result.
	\end{proof}
	
	To obtain the result, one can use a Lemma from \cite{JonathanStab}, proven in the very same way: 
	\begin{lem}[Lemma 11.6, \cite{JonathanStab}] \label{lastlemma}
		The existence of a time-like geodesic $ \gamma:(-\epsilon,\epsilon) \rightarrow \tilde{M}$ such that $\gamma(0)=p$, $r(p)=0$ and $\gamma((-\epsilon,0)) \subset M$ contradicts the fact that $\tilde{M}$ is a $C^2$ extension.
	\end{lem}
	
	As for Proposition \ref{geomextendJonathan}, this result does not use the specific structure of the uncharged massless field equations and is also valid in our context, as it only uses soft estimates (namely the blow-up of the Kretschmann scalar at boundary points $p$ where $r(p)=0$) which were already proven in \cite{Kommemi}.  Proposition \ref{oneendedCHgammaprop} and Lemma \ref{lastlemma} provide a proof of Theorem \ref{conditionnalSCCtheoremoneended}.

	\appendix

	\section{Construction of a Cauchy horizon of mixed type for the Einstein-null-dust model} \label{appendix}
	
	In  this appendix, we construct an example of a Cauchy horizon of mixed type, following Definition \ref{mixeddef}. In the second part of our development, we prove the blow up of the Hawking mass for mixed and dynamical type Cauchy horizons in the Einstein--Maxwell-null-dust model, following Poisson and Israel \cite{PoissonIsrael} and \footnote{Note that Ori considers Cauchy horizons of mixed type in \cite{Ori}.} Ori \cite{Ori}, the first instances of the mass inflation scenario in the literature. Notice that their model is very elementary, as the dust clouds are simply transported linearly in the null directions, and only interact indirectly, via the metric: thus, the ``scattering theory'' is trivial. Additionally, such a model does not allow for one-ended regular solutions, unlike the charged/massive scalar field model.  In this appendix we sketch a mathematical proof of mass inflation, using the methods of \cite{Moi}; this is also an opportunity to illustrate the terminology of the present paper in the simpler setting of dust, where little analysis is required and monotonicity suffices. \color{black} The system \eqref{dust1}, \eqref{dust1.5}, \eqref{Maxwelldust}, \eqref{dust2}, \eqref{dust3}, \eqref{dust5}, \eqref{dust6} can be expressed in spherical symmetry as
	
	\begin{equation} \label{RadiusDust}
		\partial_{u}\partial_{v} r= \partial_u \lambda=\partial_v \nu= -\frac{\Omega^2}{2r} \cdot (\frac{\varpi}{r}-\frac{e^2}{r^2}) =\frac{\lambda \nu}{r}-\frac{\Omega^2}{4r} \cdot (1- \frac{e^2}{r^2})\end{equation} 	\begin{equation} \label{OmegaDust}
		\partial_{u}\partial_{v} \log(\Omega^2)=  \frac{\Omega^2}{2r^2} \cdot ( \frac{2\varpi}{r}-\frac{3e^2}{r^2})= \frac{2 \nu \lambda}{r^{2}}+ \frac{ \Omega^{2}}{2r^{2}} \cdot (1- \frac{ 2 e^{2}}{r^{2}}),\end{equation} 	
	\begin{equation}\label{massUdust}\partial_{u} \varpi =\frac {r^2}{2} \kappa^{-1}|  f_R|^{2}, \end{equation} \begin{equation} \label{massVdust}\partial_{v} \varpi =\frac {r^2}{2} \iota^{-1}|  f_L|^{2},\end{equation}  \begin{equation} \label{dustU}
		\partial_v (r f_R)=0,
	\end{equation}\begin{equation} \label{dustV}
		\partial_u (r f_L)=0, 	\end{equation} 	\begin{equation}\label{RaychUdust}\partial_{u}(\frac {-4\nu}{\Omega^{2}})= \partial_{u}(\kappa^{-1})=\frac {4r}{\Omega^{2}}|  f_R|^{2}, \end{equation} 	\begin{equation} \label{RaychVdust}\partial_{v}(\frac {-4\lambda}{\Omega^{2}})= \partial_v(\iota^{-1})=\frac {4r}{\Omega^{2}}|f_L|^{2},\end{equation} \begin{equation}
		\label{Fdust}
		F_{\mu \nu}= \frac{e^2 \Omega^2}{ 2r^2} du \wedge dv.	\end{equation}

	\begin{prop} \label{propappendix}  	We consider the maximal development of smooth, spherically symmetric and admissible (in the sense of Definition \ref{admissibilitydef}) two-ended initial data  $(M=\mathcal{Q}^+ \times_r \mathcal{S}^2,g_{\mu \nu}, \phi,F_{\mu \nu})$ satisfying the Einstein--Maxwell-null-dust system. Suppose that Assumption \ref{blackholehyp}, \ref{affinecomplete} and \ref{subexthyp} of Theorem \ref{previous} are satisfied, and we also work under the same $(U,v)$ gauge choice. We also choose $U_{max}  \in \RR$ such that \color{black}$ [0,U_{max}]\times \{v=+\infty\} \subset \CH$  in the Penrose diagram $\mathcal{Q}^+$. We consider data  $f_L^0$ on the event horizon $\mathcal{H}^+= \{0\} \times  [v_0,+\infty)$  and $f_R^0$ on an ingoing cone $\underline{C}_{v_0}= [0,U_{max}]\times \{v_0\}$; we assume that both $f_R^0$ and $f_L^0$ are smooth in the $(U,v)$ coordinate system. We also make assumptions analogous to Assumption \ref{fieldevent} and \ref{fieldUevent} of Theorem \ref{previous} i.e.\ for some $s>\frac{1}{2}$ and some $C>0$:	$$\sup_{v \geq v_0 }|f_L^0|_{|\mathcal{H}^+}(v) \cdot v^{s}  \leq  C \cdot ,$$	$$ \sup_{0 \leq U \leq U_{max}}|f_R^0| (U, v_0) \leq C. $$   \color{black} Then by continuity of $f^0_R$ on $[0,U_{max})$,     there are three possibilities: \begin{enumerate}[a]
			\item \label{dynamicaldust}  For all $0<U_s<U_{max}$, there exists $0 \leq U \leq U_s$ such that $f^0_R(U) \neq 0$. \color{black} We call this the dynamical case.
			\item \label{staticdust} For all $U\in [0,U_{max})$, $f^0_R(U) = 0$. We call this the static case.
			\item \label{mixeddust} There exists $U_T \in (0,U_{max})$ such that for all $U\in [0,U_{T}]$, $f^0_R(U) = 0$, and such that for every $\epsilon>0$, there exists $U\in (U_{T},U_{T}+\epsilon]$ such that $f^0_R(U) \neq 0$\color{black}. We call this the mixed case.
		\end{enumerate}
		Then we have, in the three different cases: \begin{enumerate}[a]
			\item In the dynamical case, $r$ extends to $\CH$ as function $r_{CH}$  which is strictly decreasing on $[0,U_{max})$\color{black}.
			\item In the static case, $r$ extends to $\CH$ as a constant $r_{-}>0$.
			\item In the mixed case, $r$ extends to $\CH$ as a function $r_{CH}$, which is constant on $[0,U_T]$, and strictly decreasing on $(U_T,U_{max} )$.
		\end{enumerate} If we additionally assume that $f_L^0(v)=C \cdot v^{-s}+o(v^{-s})$, for some $s>\frac{1}{2}$, where $v$ is defined by gauge \eqref{gauge2}, then we have \begin{enumerate}[a]
			\item In the dynamical case, the Hawking mass $\rho$ blows up on $\CH \cap [0,U_{max}] $.
			\item In the static case, the Hawking mass $\rho$ is constant 
			on $\CH \cap [0,U_{max}] $.
			\item  In the mixed case, the Hawking mass $\rho$ is constant on $\CH \cap [0,U_{T}] $ and blows up on $\CH \cap (U_T,U_{max}]$.
		\end{enumerate}
	\end{prop}
	\begin{rmk}
		Note that $\CH \neq \emptyset$, applying (an easier version of) the argument of \cite{Moi}, which proves the stability of the Reissner--Nordstr\"{o}m Cauchy horizon for the Einstein--Maxwell--Klein--Gordon model in spherical symmetry.
	\end{rmk}
	\begin{proof} We first prove that the possibilities reduce to the three cases \ref{dynamicaldust}, \ref{staticdust}, \ref{mixeddust}. Define the set $E:=\{U\in[0,U_{max}),  f_R^0(U)\neq 0\}$. If $E=\emptyset$ then we are in case \ref{staticdust}. If $E \neq\emptyset$, define $U_T=\inf E$. There are two possibilities: either $U_T=0$ or $U_T>0$. If $U_T=0$ then we are in case  \ref{dynamicaldust} by definition of $E$. If $U_T>0$ then for all $0 \leq U < U_T$, $f_R^0(U)= 0$. By continuity of $f_R^0$ we have also $f_R^0(U_T)= 0$. Again by definition of $E$, we are therefore in case \ref{mixeddust}.

		\color{black}
		We work in the gauge $\kappa_{|\mathcal{H}^+} \equiv 1$ and we denote $\varphi_R(u)= r(u,v) f_R(u,v)$, $\varphi_L(v)= r(u,v) f_L(u,v)$ by \eqref{dustU}, \eqref{dustV}.
		
		By the stability estimates \footnote{Note that while the estimates of \cite{Moi} concern the scalar field case, they can be immediately transposed to the null dust case with no further work, as the Einstein equations are the same, except with fewer terms and moreover the propagation of both dust clouds $f_R$ and $f_L$ is trivial.} of \cite{Moi} (see section \ref{LB}), one can show that for some sub-extremal parameters $0<|e|<M$, $$ |\varpi(U_{\gamma}(v),v)-M|+ |r(U_{\gamma}(v),v)-r_-(M,e)|+ |\partial_v \log(\Omega^2)(U_{\gamma}(v),v)-2K_-(M,e)| + v \cdot |\partial_v r| (U_{\gamma}(v),v)\lesssim v^{1-2s}$$ holds on some space-like curve $\gamma$ terminating at $i^+$,  for $v$ large enough.  We will then establish estimates for all $U \in [0,U_{max}]$: combining \eqref{RadiusDust} and \eqref{OmegaDust}, we get $ \partial_U \partial_v \log(r \Omega^2) = \frac{ \Omega^{2}}{2r^{2}} \cdot (1- \frac{ 3 e^{2}}{r^{2}})$. In view of $|e| \color{black} > r_-(M,e)$ and the monotonicity of $r$ (which implies that $r(U,v)\leq r(U_{\gamma}(v),v) $ in our region of interest), we have  $ \partial_U \partial_v \log(r \Omega^2)(U,v) <0$ for all $U \in [0,U_{max}]$. Integrating in $U$ from $\gamma$,  assuming $v$ large enough (in particular $v\geq v_{\gamma}(U_s)$)\color{black}, we have, for all $U \in [0,U_{max}]$: $$ \partial_v\log(r\Omega^2)(U,v)  \leq \partial_v\log(r\Omega^2)(U_{\gamma}(v),v) \leq 2K_-(M,e) - C \cdot v^{1-2s}\color{black} \leq K_-(M,e)<0,$$ where the last inequality follows from $v$ being large enough. This proves that for all $U \in [0,U_{max}]$ \begin{equation} \label{omegadust}
			\Omega^2(U,v) \lesssim e^{K_-(M,e) \cdot v }.
		\end{equation}
		Using \eqref{omegadust} with $\partial_U \partial_v \log(r \Omega^2) = \frac{ \Omega^{2}}{2r^{2}} \cdot (1- \frac{ 3 e^{2}}{r^{2}})$ and the fact that $ C(M,e)^{-1} \lesssim |1- \frac{ 3 e^{2}}{r^{2}}| \lesssim C(M,e)$ by monotonicity of $r$ and the fact that $r$ is lower bounded, we get an improved estimate, still for all  $U \in [0,U_{max}]$: \begin{equation}  \label{omegadust2}
			|\partial_v \log(r \Omega^2)(U,v) -2K_-(M,e) | \lesssim v^{1-2s}, \hskip 5 mm  e^{3K_-(M,e) \cdot v } \lesssim   \Omega^2(U,v) \lesssim e^{K_-(M,e) \cdot v }.
		\end{equation}
		
		Now, we can write \eqref{RadiusDust}, since $\nu \leq 0$, as  $\partial_v( r|\nu|) =\Omega^2 \cdot (1- \frac{e^2}{r^2})$. Since $r$ is decreasing in $U$, $(1- \frac{e^2}{r^2(U,v)})<0$ for all  $U \in [0,U_{max}]$, as $r<r_-(M,e)<|e|$. Hence $v \rightarrow r\nu(U,v)$ is strictly decreasing for all fixed $U$, thus has a limit $r\nu_{\CH}(U) \in \RR_-$. Moreover, we establish, using \eqref{RadiusDust}, the following estimate for all  $U \in [0,U_{max}]$: \begin{equation} \label{kappaDustestimate}
			| r\nu(U,v) - r\nu_{\CH}(U) |\color{black} \lesssim \Omega^2(U,v) .
		\end{equation} Now, writing \eqref{RadiusDust} as $\partial_U(- r\lambda) =\Omega^2 \cdot (1- \frac{e^2}{r^2})$, we get decay for $\lambda$ so we can extend $r$ to $\CH$ into a function $r_{\CH} >0$. Therefore, for any fixed $U<U_{max}$, $v \rightarrow \nu(U,v)$ has a limit as $v\rightarrow +\infty$ that we consistently denote $\nu_{\CH}(U)$.  \color{black}
		
		Since $r$ is lower bounded in this region, we have $\partial_U (\kappa^{-1})(U,v) \gtrsim C \cdot \varphi_R^2(U) \cdot e^{|K_-|(M,e) v}$ hence, integrating from $\gamma$: \begin{equation} \label{kappamaxdust}
			\kappa^{-1}(U,v)  \gtrsim  e^{|K_-|(M,e)  v}  \cdot \int_{U_{\gamma(v)}}^U \varphi_R^2(U')dU'=  e^{|K_-|(M,e)  v}  \cdot \int_{U_{\gamma(v)}}^U r^2(U',v_0) \cdot (f_R^0(U'))^2dU' \gtrsim e^{|K_-|(M,e)  v}  \cdot \int_{U_{\gamma(v)}}^U  (f_R^0(U'))^2dU'.
		\end{equation}
		Then, there are our three possibilities, starting with the easiest: \begin{enumerate}
			\item \textit{Dynamical case: then for all $U\in [0,U_{max})$, $\int_{0}^U(f^0_R)^2(U')dU' > 0$.} 
			
			Indeed by \eqref{dynamicaldust}, for all $U\in [0,U_{max})$, there exists $U'<U$ with $(f^0_R(U'))^2 > 0$, hence by continuity there exists $\epsilon>0$ such that  $(f^0_R(U'))^2 > 0$ on $[U',U'+\epsilon]$, which implies $\int_{0}^U(f^0_R)^2(U')dU' > 0$ .

			In view of the limit $U_{\gamma}(v) \rightarrow 0$ as $v\rightarrow +\infty$ (recall that $\gamma$ terminates at $i^+$) \color{black} then $\lim_{v \rightarrow +\infty} \int_{U_{\gamma(v)}}^U  (f_R^0(U'))^2dU' >0$ for all $U\in [0,U_{max})$, hence $\kappa^{-1}(U,v) \rightarrow +\infty$. If $ \nu_{CH}(U)=0$ for some $0<U < U_{max}$, then \eqref{kappaDustestimate} is contradicted (after dividing both sides of \eqref{kappaDustestimate} by $r \Omega^2$)\color{black}. Hence $\nu_{\CH}(U)<0$ for all $U\in [0,U_{max}]$  thus $r_{\CH}$ is strictly decreasing on $[0,U_{max})$ as desired \color{black}. Moreover $ \kappa^{-1}(U,v) \sim \frac{4|\nu_l|(U)}{\Omega^2(U,v)}$ by \eqref{kappaDustestimate}.\color{black}
			
			If additionally we have a lower bound $f_L^0(v)=C \cdot v^{-s}+o(v^{-s})$ then by the stability estimates of \cite{Moi}, $-\lambda \sim C \cdot v^{-2s}$, at least for $U$ small enough and in the future of the curve $\gamma$. Since $2\rho-r= \kappa^{-1} \cdot (-\lambda)$, $\rho(U,v) \rightarrow +\infty$ as $ v \rightarrow+\infty$, for small $U$.
			A forciciori, $\varpi(U,v) \rightarrow +\infty$ for small $U$, and by the monotonicity of \eqref{massUdust}, $\varpi(U,v) \rightarrow +\infty$ for  all $U\in [0,U_{max})$. Hence, for  all $U\in [0,U_{max})$,  $\rho(U,v) \rightarrow +\infty$ as $ v \rightarrow+\infty$.
			\item \textit{Static case: for all $U\in [0,U_{max})$, $f^0_R(U) = 0$.}
			
			Then $\kappa(U,v) = 1$ as $\partial_U (\kappa^{-1})=0$. By \eqref{kappaDustestimate}, it means that $\nu_{\CH}(U)=0$ for all $U\in [0,U_{max})$ hence $r_{\CH}$ is constant and by \eqref{massUdust} $\varpi_{\CH}$ is constant. Their values are $ r_{\CH} \equiv  r_-(M,e)>0$ and $ \varpi_{\CH} \equiv M >0$ by the stability estimates of \cite{Moi}. This implies that $\rho$ is constant on $\CH$: $\rho(U,v) \rightarrow M-\frac{e^2}{r_-(M,e)}$ as $v \rightarrow +\infty$ for all $U\in [0,U_{max})$.
			
			\item \textit{Mixed case: $\int_{0}^{U_T}(f^0_R)^2(U')dU' = 0$ but $\int_{U_T}^{U}(f^0_R)^2(U')dU' > 0$ for all $U_T<U<U_{max}$.}
			
			Indeed, the above statement follows from the same logic as for the dynamical case, see earlier discussion. 
			
			Then similarly, $r_{\CH}(U)= r_-(M,e)$,  $\varpi_{\CH}(U)= M$, and $\rho_{\CH}(U)= M-\frac{e^2}{r_-(M,e)}$ for all $U \in [0,U_T]$. Yet,  $\kappa^{-1}(U,v) \rightarrow +\infty$ for all $U \in (U_T,U_{max})$, as $\int_{0}^U  (f_R^0(U'))^2dU' >0$. Thus $\nu_{\CH}(U)<0$ for all $U \in (U_T,U_{max})$ and $r_{\CH}$ is strictly decreasing on $(U_T,U_{max})$. If additionally we have a lower bound $f_L^0(v)=C \cdot v^{-s}+o(v^{-s})$, then we have the blow up of $\rho$, as in the dynamical case.

	\end{enumerate} \end{proof} \begin{rmk}
		Notice that the proof was considerably easier for the uncharged dust than for a charged scalar field, as we used in this section some special monotonicity properties that are not exploitable in the charged case. In contrast, in the earlier sections, we proved estimates that were harder to obtain but also more robust as they do not rely on monotonicity.
	\end{rmk}

	\end{document}